%% file: Bayesian_pyramids_R2.tex
\documentclass[12pt]{article}
\usepackage[margin=1in]{geometry}

\usepackage{graphicx}
\usepackage{natbib}
\usepackage{etex}
\usepackage{amsmath}
\usepackage{graphicx,psfrag,epsf}
\usepackage{enumerate}
\usepackage{url}
\usepackage{setspace}
\usepackage{caption}
\usepackage{xr}
\usepackage{bbm}
\usepackage[normalem]{ulem}
\usepackage{amsthm}
\usepackage{mleftright}

\usepackage{gyqcommands}
\allowdisplaybreaks

\usepackage{lipsum}

\hypersetup{
  colorlinks,
  linkcolor={red!50!black},
  citecolor={blue},
  urlcolor={blue}
}

\begin{document}

  \title{Bayesian Pyramids: Identifiable Multilayer Discrete Latent Structure Models for Discrete Data}
  \author{Yuqi Gu\footnote{Address for correspondence: Yuqi Gu, Department of Statistics, Columbia University, New York, NY,
USA. E-mail: yuqi.gu@columbia.edu.}
  \hspace{.2cm}\\
    Department of Statistics, Columbia University\\
    and \\
    David B. Dunson\footnote{Email: dunson@duke.edu.}\\
    Department of Statistical Science, Duke University}
    \date{}
  \maketitle

\def\spacingset#1{\renewcommand{\baselinestretch}%
{#1}\small\normalsize} \spacingset{1}

\begin{abstract}
High dimensional categorical data are routinely collected in biomedical and social sciences. It is of great importance to build interpretable parsimonious models that perform dimension reduction and uncover meaningful latent structures from such discrete data. Identifiability is a fundamental requirement for valid modeling and inference in such scenarios, yet is challenging to address when there are complex latent structures. 
{In this article, we propose a class of identifiable multilayer (potentially deep) discrete latent structure models for discrete data, termed \textit{Bayesian Pyramids}. 
We establish the identifiability of Bayesian Pyramids by developing novel transparent conditions on the pyramid-shaped deep latent directed graph.}
The proposed identifiability conditions can ensure Bayesian posterior consistency under suitable priors. As an illustration, we consider the two-latent-layer model and propose a Bayesian shrinkage estimation approach. Simulation results for this model corroborate the identifiability and estimability of model parameters. Applications of the methodology to DNA nucleotide sequence data uncover useful discrete latent features that are highly predictive of sequence types. The proposed framework provides a recipe for interpretable unsupervised learning of discrete data, and can be a useful alternative to popular machine learning methods.
\end{abstract}

%
%\keywords{Bayesian inference, 
%deep generative models, 
%identifiability,
%interpretable machine learning,
%latent class, 
%multivariate categorical data.}

\noindent%

\spacingset{1.45} % DON'T change the spacing!

\section{Introduction}

High-dimensional unordered categorical data are ubiquitous in many scientific disciplines, including the DNA nucleotides of A, G, C, T in genetics \citep{pokholok2005cell,nguyen2016dna},
occurrences of various species in ecological studies of biodiversity \citep{ovaskainen2016eco,ovaskainen2020eco},
responses from psychological and educational assessments or social science surveys \citep{eysenck2020person,skinner2019survey}, 
and document data gathered from huge text corpora or publications \citep{blei2003lda,erosheva2004pnas}.
Modeling and extracting information from multivariate discrete data require different statistical methods and theoretical understanding from those for continuous data.
In an unsupervised setting, it is an important task to uncover reliable and meaningful latent patterns from the potentially high-dimensional and heterogeneous discrete observations. 
Ideally, the inferred lower-dimensional latent representations should not only provide scientific insights on their own, but also aid downstream statistical analyses through effective dimension reduction.

Recently, there has been a surge of interest in {interpretable} machine learning, see \cite{doshi2017int}, \cite{rudin2019nature}, and \cite{murdoch2019int}, among others. 
Latent variable approaches, however, have received limited attention in this emerging literature, likely due to the associated complexities.
Indeed,  deep learning models with many layers of latent variables are usually considered as uninterpretable black boxes.
For example, the deep belief network  \citep{hinton2006dbn,lee2009dbn} is a very popular deep learning architecture, but it is generally not reliable to interpret the inferred latent structure.
However, for high-dimensional data, it is highly desirable to perform dimension reduction to extract the key signals in the form of lower-dimensional latent representations.
If the latent representations are themselves reliable, then they can be viewed as surrogate features of the data and then passed along to existing interpretable machine learning methods for downstream tasks.
A key to success in such modeling and analysis processes is the interpretability of the latent structure. 
This in turn relies crucially on the identifiability of the statistical latent variable model being used.

In statistical terms, a set of parameters for a family of statistical models are said to be identifiable if distinct values of the parameters correspond to distinct distributions of the observed data.
Studies under such an identifiability notion date back to \cite{koopmans1950identification}, \cite{teicher1961mix}, and \cite{goodman1974lcm}.
Model identifiability is a fundamental prerequisite for valid statistical estimation and inference. 
In the latent variable context, if one wishes to interpret the parameters and latent representations learned using a latent variable model, then identifiability is necessary for making such interpretation meaningful and reproducible.
Early considerations of identifiability in the latent variable context can be traced to the seminal work \cite{anderson1956fa} for  traditional factor analysis.
For modern latent variable models potentially containing nonlinear, non-continuous, and even deep layers of latent variables, identifiability issues can be challenging to address.

Recent developments on the identifiability of continuous latent variable models include
\cite{drton2011global}, \cite{anandkumar2013dag},  and \cite{chen2019factor}.
Discrete latent variable models are an appealing alternative to their continuous counterparts in terms of the combination of interpretability and flexibility.    
Finite mixture models \citep{mclachlan1988mixture} routinely used for model-based clustering are a canonical example involving a single discrete latent variable.
Such relatively simple approaches are insufficiently flexible for complex data sets.  Extensions with multiple latent variables and/or multilayer structure have distinct advantages in such settings, but come with increasingly complex identifiability issues.
{In this work, we are motivated to build  identifiable deep latent variable models, which are flexible enough to capture the complex dependencies in real-world data, yet also with appropriate restrictions and parsimony to yield identifiability.}

{We propose a family of multilayer, potentially deep, discrete latent variable models and propose novel identifiability conditions for them.
We establish identifiability for hierarchical latent structures organized in a ``pyramid''-\,shaped Bayesian network. In such a \textit{Bayesian Pyramid}, observed variables are at the bottom and multilayer latent variables above them describe the data generating process.
Sparse graphical connections occur between layers, and our identifiability conditions impose structural and size constraints on these between-layer graphs. 
Technically, we tackle identifiability by first reformulating the Bayesian Pyramid as a \textit{constrained} latent class model \citep[LCM;][]{goodman1974lcm} in a layerwise manner. Then we derive a nontrivial algebraic property of LCMs under such parameter constraints (Proposition 2) and combine it with Kruskal's theorem \citep{kruskal1977, allman2009} on tensor decompositions to establish identifiability.
Our identifiability results are not only technically novel, but also provide insights into methodology development. 
Indeed, the identifiability theory directly inspires the specification of deep latent architecture in Bayesian Pyramids, which features fewer latent variables deeper up the hierarchy.
The identifiability results offer a theoretical basis for learning potentially deep and interpretable latent structures from high-dimensional discrete data.}

A nice consequence of the identifiability results is the posterior consistency of Bayesian procedures under suitable priors.
As an illustration, we consider the two-latent-layer model and propose a Bayesian shrinkage estimation approach. We develop a Gibbs sampler with data augmentation for computation.
Simulation studies corroborate identifiability and show good performance of Bayes estimators.  
We apply the proposed approach to two DNA nucleotide sequence datasets \citep{uci2019}.  
For the splice junction data, when using latent representations learned from our two-latent-layer model in downstream classification of nucleotide sequence types, we achieve a remarkable accuracy  comparable to fine-tuned convolutional neural networks \citep[i.e., in][]{nguyen2016dna}.
This suggests that the developed recipe of unsupervised learning of discrete data may serve as a useful alternative to popular machine learning methods.

The rest of this paper is organized as follows.
{Section \ref{sec-layers} proposes Bayesian Pyramids, a new family of pyramid-shaped deep latent variable models for discrete data, and reformulates a Bayesian Pyramid into Constrained Latent Class Models (CLCMs) in a layerwise manner.
Section \ref{sec-id} first considers the identifiability of the general CLCMs, and then proposes identifiability conditions for the multilayer deep Bayesian Pyramids.}
To illustrate the proposed framework, Section \ref{sec-method} focuses on a two-latent-layer Bayesian Pyramid and proposes a Bayesian estimation approach.
Section \ref{sec-simu} provides simulation studies that examine the performance of the proposed methodology and corroborate the identifiability theory.
Section \ref{sec-data} applies the method to real data on nucleotide sequences.
Finally, Section \ref{sec-disc} discusses implications and future directions.
Technical proofs, posterior computation details, and additional data analyses are included in the Supplementary Material.
{Matlab code implementing the proposed method is available at \texttt{https://github.com/yuqigu/BayesianPyramids}.}

\section{Bayesian Pyramids: Multilayer Latent Structure Models}\label{sec-layers}

This section proposes Bayesian Pyramids to model the joint distribution of multivariate unordered categorical data.
For an integer $m$,  denote $[m]=\{1,2,\ldots,m\}$.
Suppose for each subject, there are $p$ observed variables $\bo y=(y_1,\ldots,y_p)^\top$, where $y_j\in[d_j]$ for each variable $j\in[p]$. The $d_j$ is the number of categories that the $j$th observed variable  can potentially take.
We mainly consider multiple (potentially deep) layers of \textit{binary} latent variables, motivated by better computational tractability and also the simpler interpretability, with each variable encoding presence or absence of a certain latent construct. 
{The stack of multiple layers of binary latent variables induces a model resembling}
deep belief networks  \citep{hinton2006dbn}. However, our proposed class of models is more general in terms of distributional assumptions. 
In the following, we first describe in detail our proposed Bayesian Pyramids in Section \ref{subsec-dp}, and then connect them to a latent class model \citep{goodman1974lcm} subject to certain constraints.

\subsection{Multilayer Bayesian Pyramids}\label{subsec-dp}
The proposed models belong to the broader family of Bayesian networks \citep{pearl2014}, which are directed acyclic graphical models that can encode rich conditional independence information.
We propose a ``pyramid''-like Bayesian network with one latent variable at the root and more and more latent variables in downward layers, where the bottom layer consists of the $p$ observed variables $y_{1},\ldots,y_{p}$.
Denote the number of latent layers in this Bayesian network by $D$, which can be viewed as the depth.
Specifically, let the layer of latent variables consecutive to the observed $\bo y$ be $\aaa^{(1)}=(\alpha_{1},\ldots,\alpha_{K_1})\in\{0,1\}^{K_1}$ with $K_1$ variables, 
and let a deeper layer of latent variables consecutive to $\aaa^{(m)}\in\{0,1\}^{K_m}$ be $\aaa^{(m+1)}\in\{0,1\}^{K_{m+1}}$ for $m=1,2,\ldots, D-1$. 
Finally, at the top and the deepest layer of the pyramid, we specify a single discrete latent variable $z$, or equivalently $z^{(D)} \in\{1,\ldots,B\}$.
In this Bayesian network, all the directed edges are pointing in the top-down direction only between two consecutive layers, and there are no edges between variables within a layer. This gives the following factorization of the joint distribution of $\bo y$ and latent variables, 
{where the subscript $i$ denotes an index of a random subject,
\begin{align}\label{eq-model}
	&\MP(\bo y_i,\{\aaa^{(m)}_i\},z^{(D)}_i)
	= \MP(\bo y_i\mid \aaa^{(1)}_i)
	\prod_{m=1}^{D-2} \MP(\aaa^{(m)}_i\mid \aaa^{(m+1)}_i)
	\MP(\aaa^{(D-1)}_i \mid z^{(D)}_i)
	 \mathbb P(z^{(D)}_i);\\
\label{eq-modelfac}
	%\text{where}\quad
	&\MP(\bo y_i\mid \aaa^{(1)}_i) 
	= \prod_{j=1}^p \MP\Big(y_{i,j} \,\Big{|}\, \aaa^{(1)}_{i,\pa(j)}\Big),
	\quad
	\MP(\aaa^{(m)}_i \mid \aaa^{(m+1)}_i) 
	= \prod_{k=1}^{K_m} \MP\Big(\alpha^{(m)}_{i,k} \,\Big{|}\, \aaa^{(m+1)}_{i,\pa(k^{(m)})} \Big);
\end{align}
In the above display, for each $j\in[p]$, the $\pa(j)\subseteq [K_1]$ is a collection of first-latent-layer variables which are parents of $y_j$.
Similarly, for each $k\in[K_m]$, the $\pa(k^{(m)})\subseteq[K_{m+1}]$ is a collection of $(m+1)$-latent-layer variables which are parents of the $m$-latent-layer variable $\alpha_{k}$.
}% end of darkblue

%% multiple layers
\begin{figure}[h!]\centering
\resizebox{\textwidth}{!}{
    \begin{tikzpicture}[scale=1.7]
  
    \node (h_root)[hidden] at (4,1.4) {$z_i$};
    
    \node (h21)[hidden] at (1,0.2) {$\alpha^{(2)}_{i,1}$};
    \node (h22)[hidden] at (4,0.2) {$\cdots$};
    \node (h23)[hidden] at (7,0.2) {$\alpha^{(2)}_{i,K_2}$};
    
    \node (hb1)[hidden] at (0,-1) {$\alpha^{(1)}_{i,1}$};
    \node (hb2)[hidden] at (1,-1) {$\alpha^{(1)}_{i,2}$};
    \node (hb3)[hidden] at (2,-1) {$\alpha^{(1)}_{i,3}$};
    \node (hb4)[hidden] at (3,-1) {$\cdots$};
    \node (hb5)[hidden] at (4,-1) {$\cdots$};
    \node (hb6)[hidden] at (5,-1) {$\cdots$};
    \node (hb7)[hidden] at (6,-1) {$\cdots$};
    \node (hb8)[hidden] at (7,-1) {$\cdots$};
    \node (hb9)[hidden] at (8,-1) {$\alpha^{(1)}_{i,K_1}$};
    
    \node (v1)[neuron] at (-0.8,-2.5) {$y_{i,1}$};
    \node (v2)[neuron] at (0,-2.5) {$y_{i,2}$};
    \node (v3)[neuron] at (0.8,-2.5) {$y_{i,3}$};
    
    \node at (1.5,-2.5) {$\cdots$};
    
    \node (v4)[neuron] at (2.2,-2.5) {$\cdots$};
    \node (v5)[neuron] at (3,-2.5) {$\cdots$};
    \node (v6)[neuron] at (3.8,-2.5) {$\cdots$};

    \node at (5,-2.5) {$\cdots$};
    \node at (6,-2.5) {$\cdots$};
    
    \node at (4,-1.8) {$\cdots$};
    \node at (5,-1.8) {$\cdots$};
    \node at (6,-1.8) {$\cdots$};
    \node at (7,-1.8) {$\cdots$};

    \node (v7)[neuron] at (7.2,-2.5) {$\cdots$};
    \node (v8)[neuron] at (8,-2.5) {$\cdots$};
    \node (v9)[neuron] at (8.8,-2.5) {$y_{i,p}$};

    \draw[arr] (hb1) -- (v1);
    \draw[arr] (hb1) -- (v2);
    \draw[arr] (hb1) -- (v3);
    
    \draw[arr] (hb2) -- (v2);
    
    \draw[arr] (hb4) -- (v4);
    \draw[arr] (hb4) -- (v5);
    \draw[arr] (hb4) -- (v6);
    
    \draw[arr] (hb9) -- (v7);
    \draw[arr] (hb9) -- (v8);
    \draw[arr] (hb9) -- (v9);

    \draw[arr,dotted] (h_root) -- (h21);
    \draw[arr,dotted] (h_root) -- (h22);
    \draw[arr,dotted] (h_root) -- (h23);
    
    \draw[arr] (h21) -- (hb1);
    \draw[arr] (h21) -- (hb2);
    \draw[arr] (h21) -- (hb3);
    
    \draw[arr] (h22) -- (hb4);
    \draw[arr] (h22) -- (hb5);
    \draw[arr] (h22) -- (hb6);
    
    \draw[arr] (h23) -- (hb7);
    \draw[arr] (h23) -- (hb8);
    \draw[arr] (h23) -- (hb9);

    % graph structures beyond trees
    \draw[arr] (h21) -- (hb4);
    \draw[arr] (h22) -- (1,-0.73);
    \draw[arr] (h22) -- (hb7);
    \draw[arr] (h23) -- (4,-0.73);
    
    \draw[arr] (hb1) -- (v4);

    \node[anchor=west] (g1) at (9, -1.6) {$\mathbf G^{(1)}$};
    \node[anchor=west] (g2) at (9, -0.2) {$\mathbf G^{(2)}$};
    \node[anchor=west] (g3) at (9, 1.2) {$\bo\vdots$};
    
\end{tikzpicture}
}% end of resize
\caption{Multiple layers of binary latent traits $\aaa^{(d)}_i$s model the distribution of observed $\bo y_i$. Binary matrices $\GG^{(1)},\GG^{(2)},\ldots$ encode the sparse connection patterns between consecutive layers. Dotted arrows emanating from the root variable $z_i$ summarize omitted layers $\{\aaa^{(d)}_i\}$.
}
\label{fig-layers}
\end{figure}
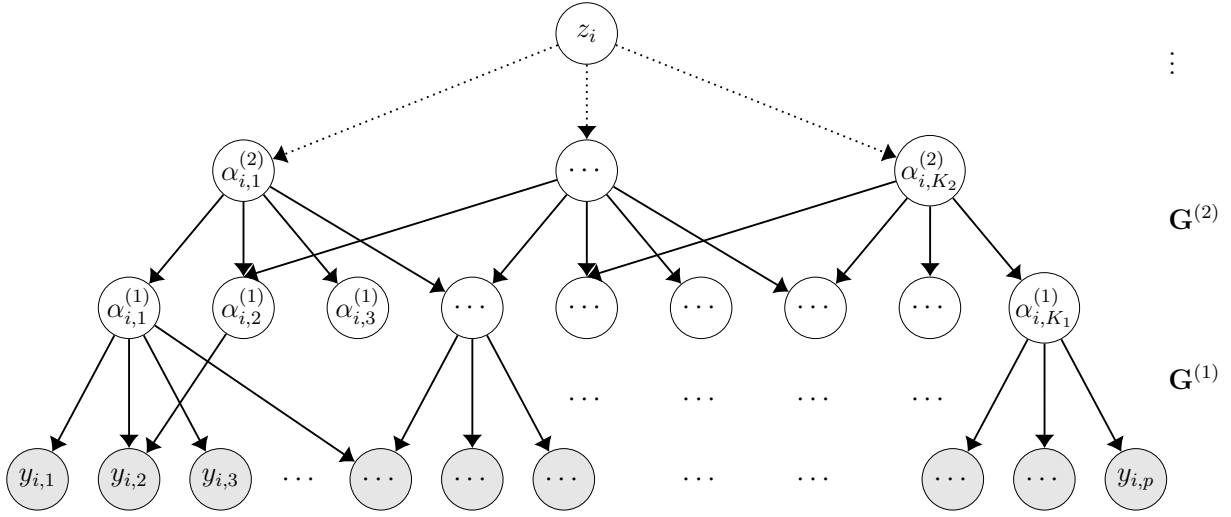

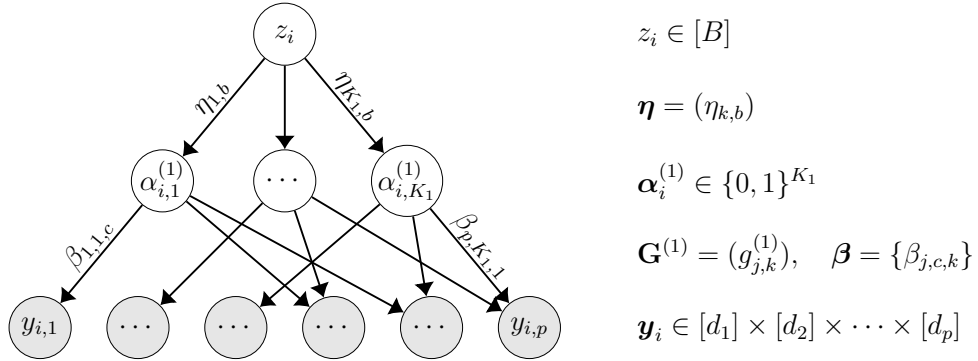
\begin{figure}[h!]\centering
\resizebox{0.8\textwidth}{!}{
    \begin{tikzpicture}[scale=1.8]

    \node (v1)[neuron] at (0, 0) {$y_{i,1}$};
    \node (v2)[neuron] at (0.8, 0) {$\cdots$};
    \node (v3)[neuron] at (1.6, 0) {$\cdots$};
    \node (v4)[neuron] at (2.4, 0) {$\cdots$};
    \node (v5)[neuron] at (3.2, 0) {$\cdots$};
    \node (v6)[neuron] at (4, 0)   {$y_{i,p}$};

    \node (h1)[hidden] at (1.0, 1.2) {$\alpha^{(1)}_{i,1}$};
    \node (h2)[hidden] at (2.0, 1.2) {$\cdots$};
    \node (h3)[hidden] at (3.0, 1.2) {$\alpha^{(1)}_{i,K_1}$};
    
    \node (h0)[hidden] at (2, 2.4) {$z_i$};

    \node[anchor=west] (z) at (4.8, 2.4) {$z_i\in [B]$};
    
    \node[anchor=west] (eta) at (4.8, 1.8) {$\bo\eta = (\eta_{k,b}) $};
    
    \node[anchor=west] (h) at (4.8, 1.2) {$\aaa^{(1)}_i \in \{0, 1\}^{K_1}$};
   
    \node[anchor=west] (g1) at (4.8, 0.6) {$\GG^{(1)} = (g^{(1)}_{j,k}),\quad \bo\beta=\{\beta_{j,c,k}\}$};
    
    \node[anchor=west] (v) at (4.8, 0) {$\bo y_i \in [d_1]\times [d_2]\times \cdots \times[d_p]$};

    \draw[qedge] (h0) -- (h1) node [midway,above=-0.12cm,sloped] {\textcolor{black}{$\eta_{1,b}$}}; 
    \draw[qedge] (h0) -- (h2);
    \draw[qedge] (h0) -- (h3) node [midway,above=-0.12cm,sloped] {\textcolor{black}{$\eta_{K_1,b}$}};

    \draw[qedge] (h1) -- (v1) node [midway,above=-0.12cm,sloped] {\textcolor{black}{$\beta_{1,1,c}$}}; 
    
    \draw[qedge] (h1) -- (v4) node [midway,above=-0.12cm,sloped] {};  
    
    \draw[qedge] (h2) -- (v2) node [midway,above=-0.12cm,sloped] {}; 
    
    \draw[qedge] (h2) -- (v4) node [midway,above=-0.12cm,sloped] {}; 
    
    \draw[qedge] (h1) -- (v5) node [midway,above=-0.12cm,sloped] {}; 
    
    \draw[qedge] (h3) -- (v3) node [midway,above=-0.12cm,sloped] {}; 
    
    \draw[qedge] (h3) -- (v5) node [midway,above=-0.12cm,sloped] {}; 
    
    \draw[qedge] (h2) -- (v6) node [midway,above=-0.12cm,sloped] {}; 

    \draw[qedge] (h3) -- (v6) node [midway,above=-0.12cm,sloped] {\textcolor{black}{$\beta_{p,K_1,1}$}}; 
    
\end{tikzpicture}
}% end of resize
\caption{Two-latent-layer model. Latent variables $z_i, \alpha_{i,1},\ldots,\alpha_{i,K_1}$ and observed variables $y_{i,1},\ldots, y_{i,p}$ are subject-specific, and model parameters $\GG^{(1)}, \bo\beta, \bo\tau, \bo\eta$ are population quantities.}
\label{fig-2layer}
\end{figure}

Fig.~\ref{fig-layers} gives a graphical visualization of the proposed model. 
%%%
To make clear the sparsity structure of this graph, we introduce binary matrices $\GG^{(1)},\GG^{(2)},\cdots,\GG^{(D-1)}$, termed as \textit{graphical matrices}, to encode the connecting patterns between consecutive layers. 
That is, the $\GG^{(d)}$ summarizes the parent-child graphical patterns between the two consecutive layers $d$ and $d+1$.
Specifically, matrix $\GG^{(1)}=(g^{(1)}_{j,k})$ has size $p\times K_1$ and matrix $\GG^{(m)}=(g^{(m)}_{k,k'})$ has size $K_{m-1}\times K_m$ for $m=2,\ldots, D-1$, with entries
\begin{align*}
	&g^{(1)}_{j,k} = 1,~~~\text{if}~~\alpha^{(1)}_k~~\text{is a parent of}~~y_j;\\
	&g^{(m)}_{k,k'} = 1,~~\text{if}~~\alpha^{(m)}_k~~\text{is a parent of}~~\alpha^{(m-1)}_{k'},~~ 2\leq m\leq D-1.
\end{align*}
%%%
Each variable in the graph is subject-specific, {implying that all the \textit{circles} in Fig. \ref{fig-layers} represent subject-specific quantities.}
Namely, if there are $n$ subjects in the sample, each of them would have its own realizations of $\bo y$ and $\aaa^{(d)}$s.
The proposed directed acyclic graph is not necessarily a tree, as shown in Fig.~\ref{fig-layers}. That is, each variable can have multiple parent variables in the layer above it, while in a directed tree each variable can only have one parent.
{As a simpler illustration, we also provide a two-latent-layer Bayesian Pyramid in Fig.  \ref{fig-2layer}, which features a simpler yet still quite expressive architecture. 
For example, we can specify the conditional distribution of each observed $y_{j}$ using a multinomial or binomial logistic model with its parent variables as linear predictors, and specify the distribution of $\aaa$ given $z$ using a latent class model. Namely, for a random subject $i$,
\begin{align}\notag
\mathbb P(y_{i,j}=c\mid \aaa^{(1)}_i=\aaa) 
    &=\frac{ \exp\left(\beta_{j,c,0} + \sum_{k=1}^{K_1} \beta_{j,c,k} g^{(1)}_{j,k}\alpha_k \right)}{ \sum_{m=1}^{d_j} \exp\left(\beta_{j,m,0} + \sum_{k=1}^{K_1} \beta_{j,m,k} g^{(1)}_{j,k}\alpha_k \right)},
	~~ j\in[p],~ c \in [d_j],~  \aaa\in\{0,1\}^{K_1};
	\\ 
	\label{eq-2layer}
    \mathbb P(\aaa^{(1)}_i=\aaa) 
    &= \sum_{b=1}^B \mathbb P(z_i=b) \mathbb P(\aaa^{(1)}_i=\aaa\mid z_i=b) 
    = \sum_{b=1}^B \tau_b \prod_{k=1}^{K_1} \eta_{k,b}^{\alpha_k} (1-\eta_{k,b})^{1-\alpha_k}.
\end{align}
Later in Section \ref{sec-method}, we will focus on the two-latent-layer model when describing the estimation methodology; please refer to that part for further details.}

{We provide some discussions on the conceptual motivation for the proposed multilayer model. Intuitively, for each subject $i$, the $z_i\in[B]$ in the deepest layer of the pyramid can encode some coarse-grained latent categories of subjects, while the vector $\aaa^{(1)}_i\in\{0,1\}^{K_1}$ encodes more fine-grained latent features of subjects. 
The hierarchical multilayer structure is conceptually appealing as it can provide latent representations of data in \textit{multiple resolutions}, where there are nonlinear compositions between different latent resolutions that boost model flexibility.
The model \eqref{eq-2layer} for generating $\aaa^{(1)}$ given $z$ is a nonlinear latent class model (LCM). Hence $z$ cannot be viewed as a linear projection of the $\aaa^{(1)}$-layer; rather, the nonlinear LCM (with a deeper variable $z$ and conditional independence of $\alpha_k$'s given $z$) can encode and explain rich and complex dependencies between the variables $\alpha^{(1)},\ldots,\alpha^{(1)}_{K-1}$ using parsimonious parametrizations.}

%%%%%%%%%%%%%%%%%%%%

\subsection{Reformulating the Bayesian Pyramid as a Constrained Latent Class Model}\label{sec-sub-spec}
{In this subsection, we reveal an interesting equivalence relationship between our Bayesian Pyramid models and Latent Class Models \citep[LCMs][]{goodman1974lcm} with certain equality constraints. Such equivalence will pave the way for investigating the identifiability of Bayesian Pyramids.}
The traditional latent class model (LCM) \citep{lazarsfeld1950,goodman1974lcm} posits the following joint distribution of $\bo y_i$ for a random subject $i$,
\begin{align}\label{eq-lca}
	\mathbb P(y_{i,1}=c_1,\ldots,y_{i,p}=c_p) = 
	\sum_{h=1}^H \nu_{h}\prod_{j=1}^p 
	  \lambda^{(j)}_{c_j,h},
	  \quad \forall c_j\in[d_j],~j\in[p].
\end{align}
The above equation specifies a finite mixture model with one discrete latent variable having $H$ categories (i.e., $H$ latent classes). 
In particular, given a latent variable $z\in[H]$, $\nu_h = \mathbb P(z = h)$ denotes the probability of $z$ belonging to the $h$th latent class,  and $\lambda^{(j)}_{c_j,h} = \mathbb P(y_j=c_j\mid z_j=h)$ denotes the conditional probability of response $c_j$ for variable $y_j$ given the latent class membership $h$.
Expression \eqref{eq-lca} implies the $p$ observed variables $y_1,\ldots,y_p$ are conditionally independent given the latent $z$.
In the Bayesian literature, the model in \cite{dunson2009} is a nonparametric generalization of \eqref{eq-lca}, which allows an infinite number of latent classes.

Latent class modeling under \eqref{eq-lca} is widely used in social and biomedical sciences, where researchers often hope to infer subpopulations of individuals having different profiles \citep{collins2009lca}. However, the overly-simplistic form of \eqref{eq-lca} can lead to poor performance in inferring distinct and interpretable subpopulations.  In particular, the model assumes that individuals in different subpopulations have completely different probabilities $\lambda^{(j)}_{c_j,h}$ for all $c_j \in [d_j]$ and $j \in [p]$, and conditionally on subpopulation membership all the variables are independent.  
These restrictions can force the number of classes to increase in order to provide an adequate fit to the data, which can degrade interpretability of a plain latent class model.
%and moreover the linear growth in the number of parameters in \eqref{eq-lca} with $k$ can lead to large uncertainty in parameter estimation.

%%%%%%%%%%%%%%%%%%%%%%%%%%%%%%%%%%%%%%%%%%%%%%%%%%%%%%%%%%%%%%%%%%%%%%%%%%%%%

We introduce some notation before proceeding. 
Denote a $p\times H$ all-one matrix by $\mathbf 1_{p\times H}$.
For a matrix $\mathbf S$ with $p$ rows and a set $\mc A \subseteq [p]$, denote by $\mathbf S_{\mc A,\bcolon}$ the submatrix of $\mathbf S$ consisting of rows indexed in $\mc A$.
%%%%%
{Consider} a family of constrained latent class models, which enable learning of a potentially large number of interpretable, identifiable, and diverse latent classes.
A key idea is sharing of parameters within certain latent classes for each observed variable.
%%%
We introduce a binary \textit{constraint matrix} $\mathbf S=(S_{j,h})$ of size $p\times H$, which has rows indexed by the $p$ observed variables and columns indexed by the $H$ latent classes. 
The binary entry $S_{j,h}$ indicates whether the conditional probability table $\lambda^{(j)}_{1:d_j,h}=(\lambda^{(j)}_{1,h},\ldots,\lambda^{(j)}_{d_j,h})$ is free or instead equal to some unknown baseline.
Specifically, if $S_{j,h}=1$ then $\lambda^{(j)}_{1:d_j,h}$ is free; while for those latent classes $h\in[H]$ such that $S_{j,h}=0$, their conditional probability tables $\lambda^{(j)}_{1:d_j,h}$'s are constrained to be equal.
Hence, $\mathbf S$ puts the following equality constraints on the parameters of \eqref{eq-lca},
\begin{align}\label{eq-equa}
	~\text{ if }~
	S_{j,h_1}=S_{j,h_2}=0,
	~\text{ then }~
	\lambda^{(j)}_{1:d_j,\,h_1} = \lambda^{(j)}_{1:d_j,\,h_2}.
\end{align}
We also enforce a natural inequality constraint for identifiability,
\begin{align}\label{eq-neq}
	~\text{ if }~
	S_{j,h_1} \neq S_{j,h_2},
	~\text{ then }~
	\lambda^{(j)}_{c_j,h_1} \neq \lambda^{(j)}_{c_j,h_2}.
	%~\forall~ c_j\in[d_j].
\end{align}
If $\mathbf S=\mathbf 1_{p\times H}$ then there are no active constraints and the original latent class model \eqref{eq-lca} is recovered.
We generally denote the conditional probability parameters by $\bo\Lambda=\{\lambda^{(j)}_{c_j, h};\,j\in[p],c_j\in[d_j],h\in[H]\}$ where $\lambda_{c_j, h} := \lambda_{c_j, 0}$ if $S_{j,h}=0$.
{As will be revealed soon, such constrained latent class models are related to our proposed Bayesian Pyramids via a neat mathematical transformation.}

Viewed from a different perspective, a latent class model \eqref{eq-lca} specifies a decomposition of the $p$-way probability tensor $\bo\Pi = (\pi_{c_1\ldots c_p})$, where $\pi_{c_1\cdots c_p} = \mathbb P(y_1= c_1,\ldots,y_p=c_p \mid \LLambda,\mb S,\nnu)$. This corresponds to the PARAFAC/CANDECOMP (CP) decomposition in the tensor literature \citep{koldabader2009}, which can be used to factorize general real-valued tensors, while our focus is on probability tensors.
The proposed equality constraint \eqref{eq-equa} induces a family of constrained CP decompositions.  

This family is connected with the sparse CP decomposition of \cite{zhoudunson2015}, with both having  equality constraints summarized by a $p\times H$ binary matrix. 
However, \cite{zhoudunson2015} encourage different observed variables to share parameters, while our proposed model encourages different latent classes to share parameters through \eqref{eq-equa}.
We have the following proposition linking the proposed multilayer Bayesian Pyramid to the constrained latent class model under equality constraint \eqref{eq-equa}.
For two vectors $\bo a, \bo b$ of the same length $M$, denote $\bo a\succeq\bo b$ if $a_i\geq b_i$ for all $i\in[M]$.
Denote by $\mathbbm{1}(\bcdot)$ the binary indicator function, which equals one if the statement inside is true and zero otherwise.

\begin{proposition}\label{prop-relation} 
    Consider the multilayer Bayesian Pyramid with binary graphical matrices $\GG^{(1)},\GG^{(2)},\cdots,\GG^{(D-1)}$. 
    \begin{itemize}
    \item[(a)]
    
	In marginalizing out all the latent variables \textbf{except} $\aaa^{(1)}$, the distribution of $\bo y$ is a constrained latent class model with $2^{K_1}$ latent classes, where each latent class can be characterized as one configuration of the $K_1$-dimensional binary vector $\aaa^{(1)}\in\{0,1\}^{K_1}$.
	The corresponding $p\times 2^{K_1}$ constraint matrix $\mb S^{(1)}$ is determined by the bipartite graph structure between the $\aaa^{(1)}$-layer and the $\bo y$-layer, with entries being
	\begin{align}\label{eq-relation} 
		S^{(1)}_{j,\aaa^{(1)}}
		&= 1-\mathbbm{1}\left(\alpha^{(1)}_k\geq \GG^{(1)}_{j,k}~\normalfont{\text{for all}}~k = 1,\ldots,K_1\right)\\
		\notag
	    &= 1-\mathbbm{1}\left(\aaa^{(1)}\succeq \GG^{(1)}_{j,\bcolon}\right), \quad
	    j\in[p],~~
	    \aaa^{(1)}\in\{0,1\}^{K_1}.
	\end{align}
	
	\item[(b)]
	Further, in considering the distribution of $\aaa^{(m)}\in\{0,1\}^{K_m}$ and marginalizing out all the latent variables deeper than $\aaa^{(m)}$ \textbf{except} $\aaa^{(m+1)}$, the distribution of $\aaa^{(m)}$ is also a constrained latent class model with $2^{K_m+1}$ latent classes, where each latent class is characterized as one configuration of the $K_{m+1}$-dimensional binary vector $\aaa^{(m+1)}\in\{0,1\}^{K_{m+1}}$.
	Its corresponding constraint matrix $\mb S^{(m)}$ is determined by the bipartite graph structure between the $m$-th and the $(m+1)$-th latent layers, with entries being
	\begin{align}
		\label{eq-relation2}
		S^{(m+1)}_{k,\, \aaa^{(m+1)}}
		&
		 = 1-\mathbbm{1}\left( \aaa^{(m+1)} \succeq \GG^{(m+1)}_{k,\bcolon} \right),
		 \quad
		 k\in[K_m],~~
		 \aaa^{(m+1)} \in \{0,1\}^{K_{m+1}}.
	\end{align}
	\end{itemize}
\end{proposition}

We present a toy example to illustrate Proposition \ref{prop-relation} and discuss its implications.
\begin{example}\label{exp-toy}
	Consider a multilayer Bayesian Pyramid with $p=6$ and $K_1 = 3$, with the graph between $\aaa^{(1)}$ and $\bo y$ displayed in Fig.~\ref{fig-2layer}. 
	The $6\times 3$ graphical matrix $\GG^{(1)}$ is presented below. 
	Proposition \ref{prop-relation}(a) states that there is a $p \times 2^{K_1}$ constraint matrix $\mb S^{(1)}$ taking the  form,
	
	\vspace{-6mm}
	{\singlespacing
	$$
		\GG^{(1)} = 
		\begin{pmatrix}
			1 & 0 & 0 \\
			0 & 1 & 0 \\
			0 & 0 & 1 \\
			1 & 1 & 0 \\
			1 & 0 & 1 \\
			0 & 1 & 1 \\
		\end{pmatrix},~~
% 		~\Longrightarrow~
		\mathbf S^{(1)} = 
		\begin{blockarray}{cccccccc}
			(000) & (100) & (010) & (001) & (110) & (101) & (011) & (111)\\
			\begin{block}{(cccccccc)}
			 1     & 0     & 1     & 1     & 0     & 0     & 1     & 0 \\
			 1     & 1     & 0     & 1     & 0     & 1     & 0     & 0 \\
			 1     & 1     & 1     & 0     & 1     & 0     & 0     & 0 \\
			 1     & 1     & 1     & 1     & 0     & 1     & 1     & 0 \\
			 1     & 1     & 1     & 1     & 1     & 0     & 1     & 0 \\ 
			 1     & 1     & 1     & 1     & 1     & 1     & 0     & 0 \\
		    \end{block}
		\end{blockarray}\,.
	$$
	}% end of singlespacing
    Entries of $\mb S^{(1)}$ are determined according to \eqref{eq-relation}, for example, $S^{(1)}_{1,(000)} = 1-\mathbbm{1}((000)\succeq \GG^{(1)}_{1,\bcolon}) = 1-\mathbbm{1}((000)\succeq(100)) = 1$; and $S^{(1)}_{6,(111)} = 1-\mathbbm{1}((111)\succeq \GG^{(1)}_{6,\bcolon}) = 1-\mathbbm{1}((111)\succeq(011)) = 0$.
    Each column of the constraint matrix $\mb S^{(1)}$ is indexed by a latent class characterized by a configuration of a $K_1$-dimensional binary vector.
	This implies that if only considering the first latent layer of variables $\aaa^{(1)}$, all the subjects are naturally divided into $2^{K_1}$ latent classes, each endowed with a binary pattern.
	%The above display also shows that in order to 
\end{example}

Proposition \ref{prop-relation} gives a nice structural characterization of multilayer Bayesian Pyramids. 
This characterization is achieved by relating the multilayer sparse graph to the constrained latent class model in a layer-wise manner.
This proposition provides a basis for investigating identifiability of multilayer Bayesian Pyramids; {see details in Section \ref{sec-id}.}

\subsection{Connections to Existing Models and Studies}\label{sec-connect}
We next briefly review connections between Bayesian Pyramids and existing models.
In educational measurement research, \cite{haertel1989rlcm} first used the term {restricted} latent class models.
Further developments along this line in the psychometrics literature led to a popular family of cognitive diagnosis models \citep{rupp2008review, dela2011gdina, von2019handbook}.
{These models are essentially binary latent skill models where each subject is endowed with a $K$-dimensional latent skill vector $\aaa\in\{0,1\}^K$ indicating the mastery/deficiency statuses of $K$ skills, and each test item $j$ is designed to measure a certain configuration of skills, summarized by a loading vector $\bo q_j\in\{0,1\}^K$. The matrix $\mathbf Q \in\{0,1\}^{J\times K}$ collecting all of the $J$ skill loading vectors $\bo q_1,\ldots,\bo q_J$ as row vectors is often pre-specified by educational experts.
The observed data for each subject is a $J$-dimensional binary vector $\bo r \in\{0,1\}^J$, indicating the correct/wrong answers to $J$ test questions in the assessment.
Recently, there have been emerging studies on the identifiability and estimation of such cognitive diagnosis models  \citep[e.g.,][]{chen2015qmat, xu2017, fang2019cdm, gu2019learning, gu2020partial, chen2020cdm}.
However, to our knowledge, there have not been works that model multilayer (i.e., deep) latent structure behind the data and investigate identifiability in such scenarios.}

%%%
Bayesian Pyramids are also related to deep belief networks \citep[DBNs,][]{hinton2006dbn}, 
sum-product networks \citep[SPNs,][]{poon2011spn} and latent tree graphical models \citep[LTMs,][]{mourad2013ltm} in the machine learning literature.
{DBNs have undirected edges between the deepest two layers designed based on computational considerations \citep{hinton2009dbn}, while a Bayesian Pyramid is a fully generative directed graphical model (Bayesian network) with all the edges pointing top down.
Such generative modeling is naturally motivated by the identifiability considerations and also provides a full description of the data generating process.
Also, DBNs in their popular form are models for multivariate binary data, feature a fully-connected graph structure, and use logistic link functions between layers.  In contrast, Bayesian Pyramids 
accommodate general multivariate categorical data, and allow flexible forms of layerwise conditional distributions and sparse connections between consecutive layers.} 
An SPN is a rooted directed acyclic graph consisting of sum nodes and product nodes.
\cite{zhao2015relate} show that an SPN can be converted to a bipartite Bayesian network, where the number of discrete latent variables equals the number of sum nodes in the SPN.
Our model is more general in that in addition to modeling a bipartite network between the latent layer and the observed layer, we further model the dependence of the latent variables instead of assuming them independent.
 LTMs are a special case of our model \eqref{eq-2layer} because, while all the variables in an LTM form a tree, we allow for more general DAGs beyond trees.
Although the above models are extremely popular, identifiability has received essentially no attention; an exception is  
\cite{zwiernik2018ltm}, which discussed identifiability for relatively simple LTMs.

Our two-latent-layer Bayesian Pyramid shares a similar structure with the nonparametric latent feature models in \cite{doshi2009correlated}.
Both works consider a mixture of binary latent feature models, with each data point associated with both a deep latent cluster ($z$ in our notation) and a binary vector of latent features ($\aaa^{(1)}$ in our notation). 
One distinction is that we adopt very flexible probabilistic distributions (see Examples 2 and 3) as conditional distributions in the DAG, while \cite{doshi2009correlated} posits that the latent features are a deterministic function of the products of $z$ and $\aaa^{(1)}$. 
In addition, Bayesian Pyramids are directly inspired by identifiability considerations, and in the next section, we provide explicit conditions that guarantee the model is identifiable -- not only in the two-latent-layer architecture similar to \cite{doshi2009correlated} but also in general deep latent hierarchies with $\aaa^{(2)},\ldots, \aaa^{(D-1)}, z$. Interestingly, \cite{doshi2009correlated} conjecture heuristically that their specification ``likely resolves the identifiability issues.''

\section{Identifiability and Constrained Latent Class Structure behind Bayesian Pyramids}\label{sec-id}

\subsection{Identifiability of the Constrained Latent Class Model and Posterior Consistency}
\label{sec-sub-id}

In Section \ref{sec-layers} we proposed a new class of multilayer latent variable models deemed Bayesian Pyramids, and 
showed that these models can be formulated as a type of constrained latent class model (CLCM)
defined in \eqref{eq-lca}-\eqref{eq-equa}.
In this section, we study theoretical properties of model \eqref{eq-lca}-\eqref{eq-equa}.
%%%%%%%%%%%%%%%%%%%%

The classic latent class model in \eqref{eq-lca} was shown by \cite{gyllenberg1994non} to be not strictly identifiable. Strict identifiability generally requires one to establish a parameterization in which the parameters can be expressed as one-to-one functions of observables.  As a weaker notion, generic identifiability requires that the map is one-to-one except on a Lesbesgue measure zero subset of the parameter space.
In a seminal paper, \cite{allman2009} leveraged Kruskal's Theorem \citep{kruskal1977} to show generic identifiability for an {unconstrained} latent class model.
{However, \cite{allman2009}'s approach is not sufficient for establishing identifiability in constrained LCMs or Bayesian Pyramids, due to the complex parameter constraints in these models.
Indeed, \cite{allman2009}'s generic identifiability results imply that in the latent class model with an \textit{unconstrained} parameter space, there exists a measure-zero subset of parameters where identifiability breaks down. 
In constrained LCMs, the equality constraints \eqref{eq-equa} exactly enforce parameters to fall into a measure-zero subset. In Bayesian Pyramids, these equality constraints arise from between-layer potentially sparse graphs.
So without a careful and thorough investigation into the relationship between the parameter constraints and the graphical structure, Kruskal's Theorem is not directly applicable to investigating the identifiability of constrained LCMs or Bayesian Pyramids.}

Below we establish strict identifiability of model 
\eqref{eq-lca}-\eqref{eq-equa}, by carefully examining the algebraic structure imposed by the constraint matrix $\mathbf S$.
We first introduce some notation.
Denote by $\otimes$ the Kronecker product of matrices and by $\odot$ the Khatri-Rao product of matrices. 
In particular, consider matrices $\mathbf A=(a_{i,j})\in\mathbb R^{m\times r}$, $\mathbf B=(b_{i,j})\in\mathbb R^{s\times t}$; 
and matrices $\mathbf C=(c_{i,j})=(\bo c_{\bcolon,1}\mid\cdots\mid\bo c_{\bcolon,k})\in\mathbb R^{n\times k}$,
$\mathbf D=(d_{i,j})=(\bo d_{\bcolon,1}\mid\cdots\mid\bo d_{\bcolon,k})\in\mathbb R^{\ell\times k}$, then there are $\mathbf A\otimes \mathbf B \in\mathbb R^{ms\times rt}$ and $\mathbf C\odot \mathbf D \in\mathbb R^{n \ell\times k}$ with
\begin{align*}
	\mathbf A\otimes \mathbf B
	=
	\begin{pmatrix}
		a_{1,1}\mathbf B & \cdots & a_{1,r}\mathbf B\\
		\vdots & \vdots & \vdots \\
		a_{m,1}\mathbf B & \cdots & a_{m,r}\mathbf B
	\end{pmatrix},
	\qquad
	\mathbf C\odot \mathbf D
	=
	\begin{pmatrix}
		\bo c_{\bcolon,1}\otimes\bo d_{\bcolon,1}
		\mid \cdots \mid
		\bo c_{\bcolon,k}\otimes\bo d_{\bcolon,k}
	\end{pmatrix}.
\end{align*}
The above definitions show the Khatri-Rao product is a column-wise Kronecker product; see more in \cite{koldabader2009}. 
We first establish the following technical proposition, which is useful for the later theorem on identifiability.

\begin{proposition}\label{prop-kr-gen}
	For the constrained latent class model, define the following $p$ parameter matrices subject to constraints \eqref{eq-equa} and \eqref{eq-neq} with some constraint matrix $\mathbf S$,
	\begin{align*}
		\bo\Lambda^{(j)} 
		&= 
		\begin{pmatrix}
			\lambda^{(j)}_{1,1} & \lambda^{(j)}_{1,2} & \cdots & \lambda^{(j)}_{1,H} \\
			\vdots & \vdots & \vdots & \vdots \\
			\lambda^{(j)}_{d_j,1} & \lambda^{(j)}_{d_j,2} & \cdots & \lambda^{(j)}_{d_j,H}
		\end{pmatrix},
		\quad
		j=1,\ldots,p.
	\end{align*}
Denote the Khatri-Rao product of the above $p$ matrices by $\mathbf K = \odot_{j=1}^p \bo\Lambda^{(j)}$, which has size $\prod_{j=1}^p d_j \times H$. The following two conclusions hold.
\begin{itemize}
	\item[(a)] If the $H$ column vectors of the constraint matrix $\mathbf S$ are distinct, then $\mathbf K$ must have full column rank $H$.
	\item[(b)] If $\mathbf S$ contains identical column vectors, then  $\mathbf K$ can be rank-deficient.
	
\end{itemize}

\end{proposition}

\begin{remark}\label{rmk-prop1}
Proposition \ref{prop-kr-gen} implies that $\mathbf S$ having distinct columns is sufficient (in part (a)) and almost necessary (in part (b)) for the Khatri-Rao product $\mathbf K = \odot_{j=1}^p \bo\Lambda^{(j)}$ to be full rank.
To see the ``almost necessary'' part, consider a special case where besides constraint \eqref{eq-equa} that $\lambda^{(j)}_{1:d_j,\,h_1}=\lambda^{(j)}_{1:d_j,\,h_2}$ if $S_{j,h_1}=S_{j,h_2}=0$, the parameters also satisfy $\lambda^{(j)}_{1:d_j,\,h_1}=\lambda^{(j)}_{1:d_j,\,h_2}$ if $S_{j,h_1}=S_{j,h_2}=1$. 
In this case, our proof shows that whenever the binary matrix $\mathbf S$ contains identical column vectors in columns $h_1$ and $h_2$, the matrix $\mathbf K$ also contains identical column vectors in columns $h_1$ and $h_2$ and hence is surely rank-deficient. 
\end{remark}

In the Khatri-Rao product matrix $\mathbf K$ defined in Proposition \ref{prop-kr-gen}, each column characterizes the conditional distribution of vector $\bo y$ given a particular latent class.
Therefore, Proposition \ref{prop-kr-gen} reveals an  nontrivial algebraic property: whether $\mathbf S\in\{0,1\}^{p\times H}$ has distinct column vectors is linked to whether the $H$ conditional distributions of $\bo y$ given each latent class are linearly independent.
The matrix $\mathbf S$ does not need to have full column rank in order to have distinct column vectors. For example, a $3\times 3$ matrix $\mathbf S$ with three columns being $(1,0,0)^\top$, $(0,1,1)^\top$, and $(1,1,1)^\top$ is rank-deficient but has distinct column vectors. Indeed, it is not hard to see that a binary matrix $\mathbf S$ with $m$ rows can have as many as $2^m$ distinct column vectors.

Proposition \ref{prop-relation} reformulates a Bayesian Pyramid into a \textit{constrained LCM} with a constraint matrix $\mathbf S$, and then Proposition \ref{prop-kr-gen}  establishes a nontrivial algebraic property of such constrained LCMs. The Propositions \ref{prop-relation} and \ref{prop-kr-gen} together pave the way for the development of the identifiability theory of Bayesian Pyramids.
In particular, the sufficiency part of Proposition \ref{prop-kr-gen} uncovers a nontrivial inherent algebraic structure of the considered models. To prove Proposition \ref{prop-kr-gen}, we leveraged a novel proof technique, the marginal probability matrix, in order to find a sufficient and almost necessary condition for the Khatri-Rao product of conditional probability matrices to have full rank.
We anticipate that the conclusion in Proposition \ref{prop-kr-gen} can be useful in other graphical models with discrete latent variables, even beyond the Bayesian Pyramids considered in this paper. This is because graphical models involving discrete latent structure can often be formulated as a latent class model with equality constraints determined by the graph. Therefore the conclusion of Proposition \ref{prop-kr-gen} might be of independent interest.

Although the linear independence of $\mathbf K$'s columns itself does not lead to identifiability,
it provides a basis for investigating strict identifiability of our model.
We introduce the definition of strict identifiability under the current setup and then give the strict identifiability result.

\begin{definition}[Strict Identifiability]
	The constrained latent class model with \eqref{eq-equa} and \eqref{eq-neq} is said to be strictly identifiable if for any valid parameters $(\LLambda,\mathbf S,\nnu)$, the following equality holds if and only if $(\overline\LLambda, \overline{\mathbf S}, \overline\nnu)$ and $(\LLambda,\mathbf S,\nnu)$ are identical up to a latent class permutation:
	\begin{align}\label{eq-id}
		\mathbb P(\bo y = \bo c\mid \LLambda, \mathbf S, \nnu)
		=
		\mathbb P(\bo y = \bo c\mid \overline\LLambda, \overline{\mathbf S}, \overline\nnu),
		\quad
		\forall \bo c\in \times_{j=1}^p [d_j].
	\end{align}
\end{definition}

\begin{remark}
When the constraint matrix $\mb S$ is unknown and needs to be identified together with unknown continuous parameters $\bo\Lambda$, there is a trivial nonidentifiability issue that needs to be resolved.
To see this, continue to consider the special case mentioned in Remark \ref{rmk-prop1} where $\lambda^{(j)}_{1:d_j,\,h_1}=\lambda^{(j)}_{1:d_j,\,h_2}$ whenever $S_{j,h_1}=S_{j,h_2}$, then given a matrix $\mb S$ we can generally denote $\lambda^{(j)}_{1:d_j,\,h} =: \lambda^{(j)}_{1:d_j,+}$ if $S_{j,h} = 1$ and $\lambda^{(j)}_{1:d_j,\,h} =: \lambda^{(j)}_{1:d_j,-}$ if $S_{j,h} = 0$.
Then without further restrictions, the following alternative $(\tilde {\mb S}, \tilde{\bo\Lambda})$ will be indistinguishable from the true $({\mb S},\bo\Lambda)$, where 
$\tilde {\mb S} = \one_{p\times H} - {\mb S}$, $\tilde\lambda_{1:d_j,+} = \lambda^{(j)}_{1:d_j,-}$ and $\tilde\lambda_{1:d_j,-} = \lambda^{(j)}_{1:d_j,+}$.
One straightforward way to resolve such trivial nonidentifiability of $\mb S$ is to simply enforce that whenever $s_{j,h_1}>s_{j,h_2}$ the order of $\lambda_{j,c,h_1}$ and $\lambda_{j,c,h_2}$ is fixed for every possible category $c\in[d_j]$. In the following studies of identifiability, we always assume such orderings of $\bo\Lambda$ with respect to $\mb S$ has been fixed.
\end{remark}

For an arbitrary subset $\mathcal A\subseteq[p]$, denote by $\mathbf S_{\mathcal A,\bcolon}$ the submatrix of $\mathbf S$ that consists of those rows indexed by variables belonging to $\mathcal A$. The $\mathbf S_{\mathcal A,\bcolon}$ has size $|\mathcal A| \times H$.

\begin{theorem}[Strict Identifiability]\label{thm-suff}
	Consider the proposed constrained latent class model under \eqref{eq-equa} and \eqref{eq-neq} with true parameters $\bo\nu=\{\nu_h\}_{h\in[H]}$, $\bo\Lambda=\{\bo\Lambda^{(j)}\}_{j\in[p]}$, and $\mathbf S$. 
	Suppose there exists a partition of the $p$ variables $[p]=\mathcal A_1\cup \mathcal A_2\cup \mathcal A_3$ such that 
	\begin{itemize}
		\item[(a)] the submatrix $\mathbf S_{\mathcal A_i,\bcolon}$ has distinct column vectors for $i=1$ and $2$; and
		\item[(b)] for any $h_1\neq h_2\in[H]$, there is $\lambda^{(j)}_{c,h_1}\neq\lambda^{(j)}_{c,h_2}$ for some $j\in\mathcal A_3$ and some $c\in[d_j]$.
	\end{itemize}   
	Also suppose $\nu_h>0$ for each $h\in[H]$.
	Then $(\LLambda,\mathbf S,\nnu)$ are strictly identifiable up to a latent class permutation.
\end{theorem}

In the above theorem, the constraint matrix $\mathbf S$ is not assumed to be fixed and known. This implies that both the matrix 
$\mathbf S$ and the parameters can be uniquely identified from data.  
We next give a corollary of Theorem \ref{thm-suff}, which replaces the condition on parameters $\LLambda$ in part (b) with a slightly more stringent but also more transparent condition on the binary matrix $\mathbf S$.

\begin{corollary}\label{cor-three}
Consider the proposed constrained latent class model under \eqref{eq-equa} and \eqref{eq-neq}.
Suppose $\nu_h>0$ for each $h\in[H]$.
If there is a partition of the $p$ variables $[p]=\mathcal A_1\cup \mathcal A_2\cup \mathcal A_3$ such that each submatrix $\mathbf S_{\mathcal A_i,\bcolon}$ has distinct column vectors for $i=1,2,3$,
then parameters  $(\bo\nu,\bo\Lambda,\mathbf S)$ are strictly identifiable up to a latent class permutation.
\end{corollary}

The conditions of Corollary \ref{cor-three} are more easily checkable than those in Theorem \ref{thm-suff} because they only depend on the structure of the constraint matrix $\mathbf S$. 
It requires that after some column rearrangement, matrix $\mathbf S$ should vertically stack three submatrices, each of which has distinct column vectors.

The conclusions of Theorem \ref{thm-suff} and Corollary \ref{cor-three} both regard {strict identifiability}, which is the strongest possible conclusion on parameter identifiability up to label permutation.
If we consider the slightly weaker notion of {generic identifiability} as proposed in \cite{allman2009}, the conditions in Theorem \ref{thm-suff} and Corollary \ref{cor-three} can be relaxed.
Given a constraint matrix $\mathbf S$, denote the constrained parameter space for $(\nnu,\LLambda)$ by
\begin{align}\label{eq-ts}
\mathcal T^{\mathbf S} = \{(\nnu,\LLambda):\, \LLambda~\text{satisfies the constraints specified by}~\mathbf S\};	
\end{align}
and define the following subset of $\mathcal T_{\mathbf S}$ as
\begin{align}\label{eq-ns}
\mathcal N^{\mathbf S} = \{(\nnu,\LLambda)\in \mathcal T^{\mathbf S}:
&~ \exists~ (\overline\nnu,\overline\LLambda)~\text{satisfying the constraints specified by some}~\overline{\mathbf S}
\\ \notag
&~ \text{such that}~ \mathbb P(\yy\mid\nnu,\LLambda) = \mathbb P(\yy\mid\overline\nnu,\overline\LLambda)\}.	
\end{align}

With the above notation, the generic identifiability of the proposed constrained latent class model is defined as follows.

\begin{definition}[Generic Identifiability]\label{def-genid}
	Parameters $(\LLambda,\mathbf S, \nnu)$ are said to be generically identifiable if $\mathcal N_{\mathbf S}$ defined in \eqref{eq-ns} has measure zero with respect to the Lebesgue measure on $\mathcal T_{\mathbf S}$ defined in \eqref{eq-ts}.
\end{definition}

\begin{theorem}[Generic Identifiability]\label{thm-genid}
Consider the constrained latent class model under \eqref{eq-equa} and \eqref{eq-neq}.
Suppose for $i=1,2$, changing some entries of $S_{\mathcal A_i,\bcolon}$ from ``$\,1$'' to ``$\,0$'' yields an $\widetilde {\mathbf S}_{\mathcal A_i,\bcolon}$ having distinct columns. Also suppose for any $h_1\neq h_2\in[H]$, there is $\lambda^{(j)}_{h_1,c}\neq\lambda^{(j)}_{h_2,c}$ for some $j\in\mathcal A_3$ and some $c\in[d_j]$. 
Then $(\mathbf \Lambda, \mathbf S,\nnu)$ are generically identifiable.
\end{theorem}

\begin{remark}
	Note that altering some $S_{j,h}$ from one to zero corresponds to adding one more equality constraint that the distribution of the $j$-th variable given the $h$-th latent class is set to the baseline through \eqref{eq-equa}. 
	Therefore, Theorem \ref{thm-genid} intuitively implies that if enforcing more parameters in $\mc T$ to be equal can give rise to a strictly identifiable model, then the parameters that make the original model unidentifiable only occupy a negligible set in $\mc T$.
\end{remark}

Theorem \ref{thm-genid} relaxes the conditions on $\mathbf S$ for strict identifiability presented earlier. In particular, here the submatrices $\mathbf S_{\mc A_i,\bcolon}$ need not have distinct column vectors; rather, it would suffice if altering some entries of $\mathbf S_{\mc A_i,\bcolon}$ from one to zero yield distinct column vectors.
As pointed out by \cite{allman2009}, generic identifiability is often sufficient for real data analyses.

So far, we have focused on discussing model identifiability.
Next, we show that our identifiability results guarantee Bayesian posterior consistency under suitable priors.
Given a sample of size $n$, denote the observations by $\bo y_1,\ldots,\bo y_n$, which are $n$ vectors each of dimension $p$.
Recall that under \eqref{eq-lca}, the distribution of the vector $\bo y$ under the considered model can be denoted by a $p$-way probability tensor $\bo\Pi=(\pi_{c_1\cdots c_p})$.
When adopting a Bayesian approach, one can specify prior distributions for the parameters $\LLambda$, $\mb S$, and $\nnu$,
which induce a prior distribution for the probability tensor $\bo\Pi$. Within this context, we are now ready to state the following theorem.

\begin{theorem}[Posterior Consistency]
\label{thm-pos}
Denote the collection of model parameters by $\Theta=(\LLambda,\mb S,\nnu)$.
Suppose the prior distributions for the parameters $\LLambda$, $\mb S$, and $\nnu$ all have full support around the true values.
If the true latent structure $\mathbf S^0$ and model parameters $\LLambda^0$ satisfy the proposed strict identifiability conditions in Theorem \ref{thm-suff} or Corollary \ref{cor-three}, we have 
	\begin{align*}
		\mathbb P(\Theta\in  \mathcal N^c_{\epsilon}(\Theta^0) \mid \bo y_1,\ldots,\bo y_n)
		\to 0~\text{almost surely},
	\end{align*}
	where $\mathcal N^c_{\epsilon}(\Theta^0)$ is the complement of an $\epsilon$-neighborhood of the true parameters $\Theta^0$ in the parameter space.
\end{theorem}

Theorem \ref{thm-pos} implies that under an identifiable model and with appropriately specified priors, the posterior distribution places increasing probability in arbitrarily small neighborhoods of the true parameters of the constrained latent class model as sample size increases.  These parameters include the mixture proportions and the class specific conditional probabilities.

%%%%%%%%%%%%%%%%%%%%%%%%%%%%%%%%%%%%%%%%%%%

\subsection{Identifiability of Multilayer Bayesian Pyramids}\label{subsec-dpid}
According to Proposition \ref{prop-relation}, for a multilayer Bayesian Pyramid with a $p\times K_1$ binary graphical matrix $\GG^{(1)}$ between the two bottom layers, one can follow \eqref{eq-relation} to construct a $p\times 2^{K_1}$ constraint matrix $\mathbf S^{(1)}$ as illustrated in Example \ref{exp-toy}. 
We next provide transparent identifiability conditions that directly depend on the binary graphical matrices $\GG^{(m)}$s.
With the next theorem, one only needs to examine the structure of the between layer connecting graphs to establish identifiability.

\begin{theorem}\label{thm-stack}
	Consider the multilayer latent variable model specified in \eqref{eq-model}--\eqref{eq-modelfac}. Suppose the numbers $K_1,\ldots,K_{D-1}$ are known.
	Suppose each binary graphical matrix $\GG^{(m)}$ of size $K_{m-1} \times K_{m}$ (size $p\times K_1$ if $m=1$) takes the following form after some row permutation,
	\begin{align}\label{eq-ggm}
		\GG^{(m)} = 
		\begin{pmatrix}
			\mb I_{K_{m}}\\
			\mb I_{K_{m}}\\
			\mb I_{K_{m}}\\
			\GG^{(m),\star}
		\end{pmatrix},\quad m=1,\ldots,D-1,
	\end{align}
	where $\GG^{(m),\star}$ generally denotes a submatrix of $\GG^{(m)}$ that can take an arbitrary form.
	Further suppose that the conditional distributions of  variables satisfy the inequality constraint in \eqref{eq-neq}.
	Then the following parameters are identifiable up to a latent variable permutation within each layer: the probability distribution tensor $\bo\tau$ for the deepest latent variable $z^{(D)}$, the {conditional probability table} of each variable (including observed and latent) given its parents, and also the binary graphical matrices $\{\GG^{(m)};\, m=1,\ldots,D-1\}$. 
\end{theorem}

\begin{remark}
	The proof of Theorem \ref{thm-stack} provides a nice layer-wise argument on identifiability, that is, one can examine the structure of the Bayesian Pyramid in the bottom-up direction. As long as for some $\ell$ there are $\GG^{(1)},\ldots,\GG^{(\ell)}$ taking the form of \eqref{eq-ggm}, then the parameters associated with the conditional distribution of $\bo y$ and $\aaa^{(1)},\ldots,\aaa^{(\ell-1)}$ are identifiable and the marginal distributions of $\aaa^{(\ell)}$ are also identifiable.
\end{remark}

Theorem \ref{thm-stack} implies a requirement that $p\geq 3K_1$ and that $K_{m-1}\geq 3K_{m}$ for every $m=2,\ldots,D-1$, through the form of $\GG^{(m)}$ in \eqref{eq-ggm}. 
That is, the number of latent variables per layer  decreases as the layer goes deeper up the pyramid.
The condition \eqref{eq-ggm} in Theorem \ref{thm-stack} requires that each latent variable $\alpha^{(m)}_{k}$ in the $m$-th latent layer has at least three children in the $(m+1)$-th layer that do not have any other parents.
Our identifiability conditions hold regardless of
the specific models chosen for the  conditional distributions, $y_j\mid \alpha^{(1)}_k$ and each $\alpha^{(m)}_{k} \mid \alpha^{(m+1)}_{k'}$, as long as the graphical structure is enforced and these component models are not over-parameterized in a naive manner.
We next give two concrete examples which differently model the distribution of $\aaa^{(m)}\mid \aaa^{(m+1)}$ but both respect the graph given by $\GG^{(m+1)}$.

\begin{example}\label{exp-maineff}
We first consider modeling the effects of the parent variables of each $\alpha^{(m)}_{k}$ as
\begin{align}\label{eq-maineff}
    \mathbb P(\alpha^{(m)}_{k} = 1 \mid \aaa^{(m+1)}, \bo\beta^{(m+1)}, \GG^{(m+1)})
    &= f \Big( \beta^{(m+1)}_{k,0} + \sum_{k':\, g^{(m+1)}_{k,k'}  = 1} \beta^{(m+1)}_{k,k'} \alpha^{(m+1)}_{k'} \Big),
\end{align}
where $f:~\mathbb R\to (0,1)$ is a link function. The number of $\beta$-parameters 
in \eqref{eq-maineff}
equals $\sum_{k'=1}^{K_{m+1}} g^{(m+1)}_{k,k'}$, which is the number of edges pointing to  $\alpha^{(m)}_{k}$.
Choosing $f(x)=1/(1+\exp(-x))$ leads to a model similar to sparse deep belief networks
 \citep{hinton2006dbn,lee2007sdpn}.
\end{example}

\begin{example}\label{exp-or}
To obtain a more parsimonious alternative to 
\eqref{eq-maineff}, let  
\begin{align}
    \label{eq-or}
%    &
\lefteqn{
    \mathbb P(\alpha^{(m)}_{k} = 1 \mid \aaa^{(m+1)}, \bo\theta^{(m+1)}, \GG^{(m+1)})} \\ \notag
    &  
    &=\begin{cases}
        \theta^{(m+1)}_{k,1}, 
        & \text{\normalfont{if}} ~~ \mathbbm{1}(\alpha^{(m+1)}_{k'} = g^{(m+1)}_{k,k'} = 1 ~\text{for at least one}~k'\in[K_{m+1}]);\\[3mm]
        \theta^{(m+1)}_{k,0}, 
        & \text{\normalfont{otherwise}}.\\
    \end{cases}
\end{align}
Model \eqref{eq-or} satisfies the conditional independence encoded by $\GG^{(m+1)}$, since  $I(\alpha^{(m+1)}_{k'} = g^{(m+1)}_{k,k'} = 1 ~\text{for at least one}~k'\in[K_{m+1}]) \equiv \prod_{k':\, g^{(m+1)}_{k,k'} = 1} (1 -  \alpha^{(m+1)}_{k'} )$,
implying that the distribution of $\alpha^{(m)}_k$ only depends on its parents in the $(m+1)$th latent layer.
This model provides a 
probabilistic version of Boolean matrix factorization \citep{miettinen2014bmf}.
The binary indicator  equals the Boolean product of two binary vectors $\GG^{(m+1)}_{k,\bcolon}$ and $\aaa^{(m+1)}$. The $1-\theta^{(m+1)}_{k,1}$ and $\theta^{(m+1)}_{k,0}$ quantify the two probabilities that the entry $\alpha^{(m)}_k$ does not equal the Boolean product.
\end{example}

Since Examples \ref{exp-maineff} and \ref{exp-or} satisfy the conditional independence constraints encoded by graphical matrices $\GG^{(m+1)}$s, they satisfy the equality constraint in \eqref{eq-equa} with the constraint matrix $\mb S^{(m+1)}$.
Therefore, our identifiability conclusion in Theorem \ref{thm-stack} applies to both examples with appropriate inequality constraints on the $\beta$--parameters or the $\theta$-parameters; for example, see Proposition \ref{prop-2layer} in the next Section \ref{sec-method}.

Besides Examples \ref{exp-maineff} and \ref{exp-or}, there are many other models that respect the graphical structure. For example, \eqref{eq-maineff} can be extended to include interaction effects of the parents of $\alpha^{(m)}_k$ as follows, 
\begin{align}\label{eq-alleff}
    &~\mathbb P(\alpha^{(m)}_{k} = 1 \mid \aaa^{(m+1)}, \bo\theta^{(m+1)}, \GG^{(m+1)})
    = f \Big( \beta^{(m+1)}_{k,0} 
    + \sum_{k':\, g^{(m+1)}_{k,k'}  = 1} \beta^{(m+1)}_{k,\,k'} \alpha^{(m+1)}_{k'} 
    \\ \notag
   &~\qquad + \sum_{k_1\neq k_2:\atop g^{(m+1)}_{k,k_1} = g^{(m+1)}_{k,k_2} = 1} 
       \beta^{(m+1)}_{k,\,k_1 k_2} \alpha^{(m+1)}_{k_1} \alpha^{(m+1)}_{k_2}
    + \cdots 
    + \beta^{(m+1)}_{k,\,\text{all}} \prod_{\ell:\, g^{(m+1)}_{k,\ell}=1} \alpha^{(m+1)}_{\ell} 
    \Big).
\end{align}
In \eqref{eq-alleff} if $\alpha^{(m)}_k$ has $M := \sum_{k'=1}^{K_{m+1}} g^{(m+1)}_{k,k'}$ parents, then the number of $\beta$-parameters equals $2^{M}$.

\cite{anandkumar2013dag} considered the identifiability of linear Bayesian networks. Although both \cite{anandkumar2013dag} and this work address identifiability issues of Bayesian networks, their  results are not applicable to our highly nonlinear models. 
The nonlinearity requires techniques that look into the inherent tensor decomposition structures (in particular, a \textit{constrained} CP decomposition) caused by the discrete latent variables and graphical constraints imposed on the discrete latent distribution.
Such inherent constrained tensor structures are specific to discrete and graphical latent structures and are not present in the settings considered in \cite{anandkumar2013dag}.

\section{Bayesian Inference for Two-Layer Bayesian Pyramids}\label{sec-method}
As discussed earlier, the proposed multilayer Bayesian Pyramid in Section \ref{sec-layers} has universal identifiability arguments for many different model structures.
In this section, as an important special case, we focus on a two-latent-layer model with depth $D=2$ and use a Bayesian approach to infer the latent structure and model parameters.
Recall our two-latent-layer Bayesian Pyramid specified earlier in \eqref{eq-2layer} takes the form
\begin{align}\notag
\mathbb P(y_{i,j}=c\mid \aaa^{(1)}_i = \aaa) 
    &=\frac{ \exp\left(\beta_{j,c,0} + \sum_{k=1}^{K_1} \beta_{j,c,k} g^{(1)}_{j,k}\alpha_k \right)}{ \sum_{m=1}^{d_j} \exp\left(\beta_{j,m,0} + \sum_{k=1}^{K_1} \beta_{j,m,k} g^{(1)}_{j,k}\alpha_k \right)},
	\quad j\in[p],~ c \in [d_j];
	\\ \notag %\label{eq-2layer}
    \mathbb P(\aaa^{(1)}_i = \aaa) &= \sum_{b=1}^B \tau_b \prod_{k=1}^{K_1} \eta_{k,b}^{\alpha_k} (1-\eta_{k,b})^{1-\alpha_k},\quad 
    \aaa\in\{0,1\}^{K_1}.
\end{align}
Here we assume $\beta_{j,d_j,0} \equiv \beta_{j,d_j,k} \equiv  0$ for all $k\in[K_1]$, as conventionally done in logistic models.
The parameter $\beta_{j,c,k}$ can be viewed as the weight associated with the potential directed edge from latent variable $\alpha^{(1)}_k$ to observed variable $y_j$, for response category $c$. 
The $\beta_{j,c,k}$ only impacts the likelihood if there is an edge from $\alpha^{(1)}_k$ to $y_j$ with $g^{(1)}_{j,k}=1$.
Denote the collection of $\beta$ parameters as $\bo\beta$. 
\eqref{eq-2layer} specifies a usual latent class model in the second latent layer with $B$ latent classes for the $K_1$-dimensional latent vector $\aaa^{(1)}$.
This layer of the model has latent class proportion parameters $\bo\tau = (\tau_1,\ldots,\tau_B)^\top$ and conditional probability parameters $\bo\eta = (\eta_{k,b})_{K_1\times B}$.  The full set of 
model parameters across layers is $(\GG^{(1)}, \bo\beta, \bo\tau, \bo\eta)$, and 
the model structure is shown in Fig.~\ref{fig-2layer}.
We can denote the conditional probability $\mathbb P(y_j=c\mid \aaa^{(1)}=\aaa)$ in \eqref{eq-2layer} by $\lambda_{j,c,\aaa}$.

\subsection{Identifiability Theory adapted to Two-layer Bayesian Pyramids defined in \eqref{eq-2layer}}
\medskip

For the two-latent-layer model in \eqref{eq-2layer}, the following proposition presents strict identifiability conditions in terms of explicit inequality constraints for the $\bo\beta$ parameters.

\begin{proposition}\label{prop-2layer}
	Consider model \eqref{eq-2layer} with true parameters $(\GG^{(1)}, \bo\beta, \bo\tau, \bo\eta)$.
	\begin{itemize}
	\item[(a)] Suppose  $\GG^{(1)} = (\mb I_{K_1};~\mb I_{K_1};~\mb I_{K_1};~(\GG^{\star})^\top)^\top$,  $\beta_{j,\,j,\,c} \neq 0$, $\beta_{j,\,j+K_1,\,c} \neq 0$, and $\beta_{j,\,j+2K_1,\,c} \neq 0$ for $j\in[K_1]$, $c\in[d_j-1]$.
	%Also suppose for any $\aaa_1\neq\aaa_2\in\{0,1\}^{K_1}$, there exists some $j>2K_1$ and $c$ such that $\lambda^{(j)}_{c,\,\aaa_1} \neq \lambda^{(j)}_{c,\,\aaa_2}$.
	Then $\GG^{(1)}$, $\bo \beta$, and  probability tensor $\nnu^{(1)}$ of $\aaa\in\{0,1\}^K$ are strictly identifiable.
	
	\item[(b)] Under the conditions of part (a), if further there is $K_1 \geq 2\ceil{\log_2 B}+1$, then the  parameters $\bo\tau$ and $\bo\eta$ are generically identifiable from $\nnu^{(1)}$.
	\end{itemize}
\end{proposition}

As mentioned in the last paragraph of Section \ref{sec-layers}, the identifiability results in that section apply to general distributional assumptions for variables organized in a sparse multilayer Bayesian Pyramid.
When considering specific models, properties of the model can be leveraged to weaken the identifiability conditions. 
The next proposition illustrates this, establishing generic identifiability for the model in \eqref{eq-2layer}.
Before stating the identifiability result, we formally define the allowable constrained parameter space for $\bo\beta$ under a graphical matrix $\GG^{(1)}$ as
\begin{equation}\label{eq-omegab}
    \Omega(\bo\beta;\, \GG^{(1)}) = \{\beta_{1:p,\, 1:K_1,\, 1:(d_j-1)};\, \beta_{j,c,k} \neq 0~\text{if}~g^{(1)}_{j,k}=1;~\text{and}~\beta_{j,c,k}=0~\text{if}~g^{(1)}_{j,k}=0\}.
\end{equation}

\begin{proposition}
	\label{prop-gendiag}
Consider model \eqref{eq-2layer} with $\bo\beta$ belonging to $\Omega(\bo\beta;\, \GG^{(1)})$ in \eqref{eq-omegab}.
{Suppose the graphical matrix $\GG^{(1)}$ 
can be rewritten as
$\GG^{(1)}=(\GG_1, \GG_2, \GG_3, (\GG^\star)^\top)^\top$, where each $\GG_m$ has size $K_1\times K_1$ and
\begin{align*}
    \GG_m = \begin{pmatrix}
    1 & * & \cdots & * \\
    * & 1 & \cdots & * \\
    \vdots & \vdots & \ddots & \vdots\\
    * & * & \cdots & 1
    \end{pmatrix}, \quad m=1,2,3;
\end{align*}
that is, each of $\GG_1,\GG_2,\GG_3$ has all the diagonal entries equal to one while any off-diagonal entry is free to be either one or zero. }% end of blue
Also suppose $K_1 \geq 2\ceil{\log_2 B}+1$.
Then $(\GG^{(1)}, \bo\beta, \bo\tau, \bo\eta)$ are generically identifiable.
\end{proposition}

The two different identifiability conditions on the binary graphical matrix $\GG^{(1)}$ stated in Proposition \ref{prop-2layer} and in Proposition \ref{prop-gendiag} correspond to different identifiability notions -- strict and generic identifiability, respectively.
The \textit{generic} identifiability notion is slightly less stringent than  \textit{strict} identifiability, by allowing a Lebesgue-measure-zero subset $\mathcal N$ of the parameter space $\mathcal T$ where identifiability does not hold.
Our sufficient generic identifiability conditions in Proposition \ref{prop-gendiag} are much less stringent than conditions in Proposition \ref{prop-2layer}.

\subsection{Bayesian Inference for the Latent Sparse Graph and Number of Binary Latent Traits}
\medskip

We propose a Bayesian inference procedure for two-latent-layer Bayesian Pyramids.
% For ease of presentation, from now on we assume $d_1=\cdots=d_p=d$.
We apply a Gibbs sampler by employing the Polya-Gamma data augmentation in \cite{polson2013pg} together with the auxiliary variable method for multinomial logit models in \cite{holmes2006}. Such Gibbs sampling steps can handle general multivariate categorical data.

\medskip
\noindent
\textbf{Inference with a Fixed $K_1$.}
First consider the case where the true number of binary latent variables $K_1$ in the middle layer is fixed. 
Inferring the latent sparse graph $\GG^{(1)}$ is equivalent to inferring the sparsity structure of the continuous parameters $\beta_{jkc}$'s.
Let the prior for $\beta_{j,k,0}$ be $N(\mu_0,\sigma_0^2)$, 
where hyperparameters $\mu_0,\sigma_0^2$ can be set to give weakly informative priors for the logistic regression coefficients \citep{gelman2008weakly}.
Specify priors for other parameters as
\begin{flalign}\notag
\text{(i) for fixed $K_1$, 
%$k\in[K_1]$
$k=1,\ldots,K_1$
:}
&&
\beta_{jck}
    \mid(\sigma_{ck}^2,~ g_{j,k}=1) \sim N(0, \sigma_{ck}^2); \qquad  \qquad  \qquad&&
    \\ \notag
    %\qquad
%
&&\beta_{jck}
    \mid(\sigma_{ck}^2,~ g_{j,k}=0) \sim N(0, ~\sigma_0^2);    \qquad  \qquad  \qquad&&
    \\ \notag
&&\sigma_{ck}^2  \sim \text{InvGa}(b_{1\sigma}, b_{2\sigma});\qquad  \qquad  \qquad && \\ \label{eq-kknown}
&&\mathbb P(g_{j,k}=1) = 1- \mathbb P(g_{j,k}=0) = \gamma;  \qquad  \qquad  \qquad &&  
\end{flalign}
Here $\sigma_0$ is a small positive number specifying the ``pseudo''-variance, and we take $\sigma_0=0.1$ in the numerical studies. 
Our adoption of a prior variance for $\beta_{j,c,k}$ when $g_{j,k}=0$ follows a similar spirit as the ``pseudo-prior'' approach in the Bayesian variable selection literature.
In Bayesian variable selection for regression analysis, \cite{dellaportas2002bayesian} first proposed using a pseudo-prior for the variance when one variable is not included in the model to facilitate convenient Gibbs sampling steps. 
Specifically, a binary variable $\rho_j$ encodes whether or not the $j$th predictor is included in the regression, and the regression coefficient is $\beta_j\rho_j$, with $\beta_j\sim \pi_j N(0,\sigma^2) + (1-\pi_j) N(0, \sigma_0^2)$ and $\mathbb P(\rho_j=1) = \pi_j$.
The pseudo-prior variance $\sigma_0^2$ does not affect the posterior but may influence mixing of the MCMC algorithm.

The $\gamma$ in \eqref{eq-kknown} is further given a noninformative prior $\gamma\sim \text{Beta}(1,1)$.
The hyperparameters for the Inverse-Gamma distribution can be set to $b_{1\sigma}=b_{2\sigma}=2$. 
In the data augmentation part, for each subject $i\in[n]$, each observed variable $j\in[p]$, and each non-baseline category $c=1,\ldots,d_j-1$, we introduce Polya-Gamma random variable $w_{ijc} \sim \text{PG} (1, 0)$.
We use the auxiliary variable approach in \cite{holmes2006} for multinomial regression to derive the conditional distribution of each $\beta_{jck}$.
Given data $y_{ij}\in[d_j]$, introduce binary indicators $y_{ijc} = \mathbbm{1}(y_{ij} = c)$. The posteriors of $\bo\beta$ satisfy that 
\begin{align*}
    p({\bo\beta_{j,:,:}} \mid \bo y_{1:n})
    \propto
    {p(\bo\beta_{j,:,:})}
    \prod_{i=1}^n \prod_{c=1}^{d_j-1}
    \frac{\left[\exp\left(\beta_{jc0} + \sum_{k=1}^{K_1} \beta_{jck} g_{j,k} \alpha_{i,k} \right)\right]^{y_{ijc}}}{\sum_{c'=1}^{d_j} \exp\left(\beta_{jc'0} + \sum_{k=1}^{K_1} \beta_{jc'k} g_{j,k} \alpha_{i,k} \right)},
\end{align*}
and we introduce notation
\begin{align}\label{eq-defphi-main}
    \phi_{ijc} & =
    \beta_{jc0} + \sum_{k=1}^{K_1} \beta_{jck} g_{j,k} \alpha_{i,k} - C_{ij(c)};
    \\ \notag
    C_{ij(c)} &= \log\left\{ \sum_{ 1\leq \ell\leq d_j,\,\ell\neq c} \exp\left( \beta_{j\ell 0} + \sum_{k=1}^{K_1} \beta_{j\ell k} g_{j,k} \alpha_{i,k} \right) \right\}.
\end{align}
Next by the property of the Polya-Gamma random variables \citep{polson2013pg}, we have 
\begin{align*}
    \frac{\exp(\phi_{ijc})^{y_{ijc}}}{1 + \exp(\phi_{ijc})}
    =2 \exp\{(y_{ijc}-1/2)\phi_{ijc}\} \int_{0}^{\infty} \exp\{-w_{ijc} \phi_{ijc}^2/2\} p^{\text{PG}}(w_{ijc}\mid 1, 0) \text{d}w_{ijc},
\end{align*}
where $p^{\text{PG}}(w_{ijc}\mid 1, 0)$ denotes the density function of PG$(1,0)$.
Based on the above identity, the conditional posterior of the $\beta_{jc0}$'s and $\beta_{jck}$'s are still Gaussian, and the conditional posterior of each $w_{ijc}$ is still Polya-Gamma with $(w_{ijc}\mid -)
    \sim
    \text{PG}\left(1,~ \beta_{jc0} + \sum_{k=1}^{K_1} \beta_{jck} g_{j,k} \alpha_{i,k} - C_{ij(c)}\right)$; these full conditional distributions are easy to sample from. 
As for the binary entries $g_{j,k}$'s indicating the presence or absence of edges in the Bayesian Pyramid, we sample each $g_{j,k}$ individually from its posterior Bernoulli distribution.
The detailed steps of such a Gibbs sampler with known $K_1$ are presented in the Supplementary Material.

\medskip
\noindent
\textbf{Inferring an Unknown $K_1$.}
On top of the sampling algorithm described above for fixed $K_1$, we propose a method for simultaneously inferring  $K_1$ and other parameters.
In the context of mixture models,
\cite{rousseau2011overfit} defined over-fitted mixtures with more than enough latent components and used shrinkage priors to effectively delete the unnecessary ones. 
In a similar spirit but motivated by Gaussian linear factor models, \cite{legramanti2020} proposed the cumulative shrinkage process (CSP) prior, which has a spike and slab structure.
We use a CSP prior 
on the variances $\{\sigma^2_{ck}\}$ of $\{\beta_{jck}\}$ to infer the number of latent binary variables $K_1$ in a two-layer Bayesian Pyramid.
{The rationale for using such an \textit{increasing shrinkage prior} for latent dimension selection is that, it is natural to expect additional latent dimensions to play a progressively less important role in characterizing the data, so the associated parameters should have a stochastically decreasing effect. Specifically, under the CSP prior, the $k$th latent dimension is controlled by a scalar $\theta_k$ that follows a spike-and-slab distribution. Redundant dimensions will be essentially deleted by progressively shrinking the sequence $\{\theta_1,\theta_2,\ldots\}$ towards an appropriate value $\theta_{\infty}$ (the spike). In particular, \cite{legramanti2020} considered the \textit{continuous} factor model where $\theta_k>0$ denotes the variance of the factor loadings for the $k$th factor and $\theta_\infty$ is a small positive number indicating the variance of redundant latent factors.}

We next describe in detail the prior specifications with an unknown number of binary latent variables $K_1$. 
{Consider an upper bound $K_{\upper}$ for $K_1$, with $K_1 < K_{\upper}$.
Based on the identifiability conditions in Theorem \ref{thm-stack} about the shape of $\mathbf G^{(1)}_{p\times K_1}$, $K_1$ is naturally constrained to be at most $p/3$, therefore $K_{\upper}$ can be set to $\ceil{p/3}$ or smaller in practice.
We adopt a prior that detects redundant binary latent variables by increasingly shrinking the variance of $\beta_{jck}$'s as $k$ grows from 1 to $K_{\upper}$.}
Specifically, letting $\beta_{jck} \sim N(0,\sigma_{ck}^2)$, we put a CSP prior on variances $\{\sigma^2_{c1}, \sigma^2_{c2},\ldots, \sigma^2_{cK_{\upper}}\}$ for each category $c\in[d-1]$, where $\sigma^2_{\infty}$ is a prespecified small positive number indicating the spike variance for redundant binary latent variables:
\begin{align}\notag
\text{(ii) for unknown } & K_1 < K_{\upper},~ k = 1,\ldots, K_{\upper}:
\\[2mm] \label{eq-csp}
&
\sigma^2_{ck}\mid \pi_k \sim
    (1-\pi_k) \text{InvGa}(b_{1\sigma},b_{2\sigma}) + \pi_k \delta_{\sigma^2_{\infty}};
\\
\label{eq-stick}
&
\pi_k = \sum_{\ell=1}^k v_\ell\prod_{m=1}^{\ell-1}(1-v_m),
\end{align}
where $\text{InvGa}(b_{1\sigma}, b_{2\sigma})$ refers to the inverse gamma distribution with shape $b_{1\sigma}$ and scale $b_{2\sigma}$, and $\pi_k$ has a stick-breaking representation as in \eqref{eq-stick} with $v_1,v_2,\ldots, v_{K_{\upper}-1}$ independently following the Beta distribution $\text{Beta}(1,\alpha_0)$. We set $v_{K_{\upper}}\equiv 1$ to truncate the stick-breaking representation at $K_{\upper}$, similarly to \cite{legramanti2020}.
%note that \eqref{eq-stick} implies $\pi_{K_{\upper}} = 1$ when $v_{K_{\upper}}=1$. 
%indicates the stick-breaking representation is truncated at $K_{\upper}$, because $\pi_{K_{\upper}} = 1$ when $v_{K_{\upper}}=1$ by the construction in \eqref{eq-stick}. 
{The $\delta_{\sigma^2_{\infty}}$ represents the Dirac spike distribution with $\sigma^2_\infty$ serving as the variance of redundant latent variables, while $\text{InvGa}(a_{\sigma},b_{\sigma})$ represents the more diffuse slab distribution for the variances corresponding to active latent variables. The ``increasing shrinkage'' comes from the fact that as the latent variable index $k$ increases, the probability of $\alpha_k$ belonging to the spike, $\pi_k$, stochastically increases, because
$$
\mathbb E[\pi_k] 
= \sum_{\ell=1}^k \mathbb E[v_\ell] \prod_{m=1}^{\ell-1}\mathbb E[1- v_m]
= 1 - \frac{1}{(1/\alpha_0 + 1)^k}
$$ 
increases as the index $k$ increases. Therefore, the CSP prior features an increasing amount of shrinkage for larger $k$. We introduce auxiliary variables $\{h_k;\; k=1,\ldots,K_{\upper}\}$ with $h_k\in[K_{\upper}]$ to help with understanding and facilitate posterior computation. 
Specifically, the prior in \eqref{eq-csp} can be obtained by marginalizing out a discrete auxiliary variable $h_k$ with $\mathbb P(h_k = \ell) = v_\ell\prod_{m=1}^{\ell-1}(1-v_m)$, so \eqref{eq-csp}--\eqref{eq-stick} can be reformulated in terms of $h_k$ as
\begin{align*}
    \left(\sigma_{ck}^2\mid h_k\right) &\sim \mathbbm{1}(h_k > k)\cdot \text{InvGa}(b_{1\sigma}, b_{2\sigma}) + \mathbbm{1}(h_k\leq k)\cdot \delta_{\sigma^2_{\infty}};
    \\
    \mathbb P(h_k \leq k) &= \sum_{\ell=1}^k v_\ell\prod_{m=1}^{\ell-1}(1-v_m) = \pi_k,\quad k=1,\ldots,K_{\upper}.
\end{align*}
Therefore, the auxiliary variables $h_k$ determine whether the $k$th latent dimension is in the spike and hence corresponds to a redundant latent variable $\alpha_k$; specifically, if $h_k > k$ then $\sigma^2_{ck}$ follows the slab distribution $\text{InvGa}(b_{1\sigma}, b_{2\sigma})$ and $\alpha_k$ is \textit{active}, otherwise $\sigma^2_{ck}=\sigma^2_\infty$ is in the spike and $\alpha_k$ is \textit{redundant}.
Given \eqref{eq-csp}, the largest possible number of active latent variables is $K_{\upper}-1$, because $\pi_{K_{\upper}} \equiv 1 = \mathbb P(h_{K_{\upper}} \leq K_{\upper})$ and the last latent variable is always redundant.
Since the event $\mathbbm{1}(h_k > k)$ indicates that the $k$th component is in the slab and hence active, we can write the total number of active latent dimensions $K^\star$ as
\begin{equation}\label{eq-kstar}
    K^\star
    = \sum_{k=1}^{K_{\upper}}  \mathbbm{1}(h_k > k).
\end{equation}
}% end of blue
Tracking the posterior samples of all the $h_k$ can give a posterior estimate of $K^\star$.
The above data augmentation leads to Gibbs updating steps. We present the details of our Gibbs sampler in the Supplementary Material.

\begin{remark}\label{rmk-deep}
It is methodologically straightforward to extend our Gibbs sampler to deep Bayesian Pyramids with more than two latent layers. 
%Essentially, our current MCMC sampling steps have all the ingredients needed for estimating deeper models. 
To see this, note that in deep Bayesian Pyramids with $m\geq 2$, for each $m$, the conditional distribution of $\aaa^{(m)}$ given $\aaa^{(m+1)}$ in Example 2 is a special case of the conditional distribution of $\bo y$ given $\aaa^{(1)}$. Indeed, both of these conditionals follow generalized linear models with the (multinomial) logit link, with the parent variables serving as predictors for the child.
Under such a formulation, introducing additional Polya-Gamma auxiliary variables for the $\aaa^{(1)}$-layer similar to those for the $\bo y$-layer would allow Gibbs updates in a three-latent-layer Bayesian Pyramid.
In this work, we focus on two-latent-layer models for computational efficiency.
\end{remark}

\begin{remark}
We also remark that it is not hard to derive and implement a Gibbs sampler for constrained latent class models (CLCMs) mentioned in Section \ref{sec-sub-id}.
Performing Gibbs sampling for CLCMs is a relatively straightforward extension to the current Gibbs sampler for Bayesian Pyramids. 
The reason is that one can similarly use the data augmentation strategies in \cite{holmes2006} and \cite{polson2013pg} to deal with $y_j\in[d_j]$; and one can also adopt similar priors for the binary constraint matrix $\mathbf S\in\{0,1\}^{p\times H}$ as those adopted for the graphical matrix $\mathbf G\in\{0,1\}^{p\times K_1}$ in a Bayesian Pyramid. 
%similarly as those for the binary graphical matrix $\mathbf G\in\{0,1\}^{p\times K_1}$. 
Our preliminary simulations for such a Gibbs sampler under CLCMs showed that the recovery of parameters and the constraint matrix are not as stable and accurate as that for Bayesian Pyramids (shown in the later Figures \ref{fig-betapos}--\ref{fig-rmse}). 
One explanation is that compared to multilayer Bayesian Pyramids, the parametrization of a CLCM is less parsimonious and requires many more mixture proportion parameters in order to describe the same joint distribution of the observed variables. Therefore, we have chosen to focus on the method for Bayesian Pyramids in this work, and treat CLCMs mainly as intermediate tools to help establish identifiability of Bayesian Pyramids.
\end{remark}

\section{Simulation Studies}\label{sec-simu}
We conducted replicated simulation studies to assess the Bayesian  procedure proposed in Section \ref{sec-method} and examine whether the model parameters are indeed estimable as implied by Theorem \ref{thm-pos}.
Consider a two-latent-layer Bayesian Pyramid with $p=20$ observed variables, $d=4$ response categories for each observed variable, and $K_1=4$ binary latent variables in the middle layer and one binary latent variable in the deep layer. 
Let the true $p\times K_1$ binary graphical matrix be $\GG^{(1)} = (\mb I_4;\, \mb I_4;\, \mb I_4;\, 1100;\, 0110;\, 0011;\, 1001;\, 1010;\, 1001;\, 0101;\, 1110;\, 0111)$.
Such a  $\GG^{(1)}$ satisfies the  conditions for strict identifiability in Theorem \ref{thm-stack}, since it contains three copies of the identity matrix $\mb I_4$ as submatrices.
Let the true intercept parameters for categories $c=1,2,3$ be $(\beta_{j,1,0}, \, \beta_{j,2,0}, \, \beta_{j,3,0}) = (-3, -2, -1)$ for each  $j$; and for any $g^{(1)}_{j,k} = 1$, let the corresponding true main-effect parameters of the binary latent variables be $\beta_{j,c,k} = 3$ if variable $y_j$ has a single parent and $\beta_{j,c,k} = 2$ if $y_j$ has multiple parents.
See Fig.~\ref{fig-betapos}(c) for a heatmap of the sparse matrix of the main-effect parameters $(\beta_{j1k};\, j\in[p],\, k\in[K_1])$ for category $c=1$.

We use the method developed in Section \ref{sec-method} with the CSP prior for posterior inference under an unknown $K_1$. We {specify an upper bound for $K_1$ as $K_{\upper} = 7$, because $\ceil{p/20}=7$ is a natural upper bound here based on the identifiability considerations mentioned before}.
As for the hyperparameters in the CSP prior, we mainly {follow the default setting and suggestion in \cite{legramanti2020}. Specifically, we set $\alpha_0$ to the same value in \cite{legramanti2020}, $\alpha_0 = 5$; we also follow the suggestion in \cite{legramanti2020} that the spike variance $\sigma^2_{\infty}$ should be a small positive number but should avoid to take excessively low values by setting $\sigma^2_{\infty}$ = $0.07$. 
% Our simulations adopting these choices turn out to produce accurate estimation results under different settings 
Our simulations adopting these choices turn out to produce accurate estimation results under different settings (see the later Figures \ref{fig-betapos}--\ref{fig-rmse}). In our preliminary simulations, we also varied these hyperparameters around the above values, and did not observe the algorithmic performance to be sensitive to them.}
Throughout the Gibbs sampling iterations, we enforce an identifiability constraint on $\bo\beta$ that $\beta_{jck} > 0$ as long as $g_{j,k} = 1$.
For each sample size $n$ we conduct 50 independent simulation replications. 
Since the model is identifiable up to a permutation of the latent variables in each layer, we post-process the posterior samples to find a column permutation of $\GG^{(1)}$ to best match the simulation truth; then the columns of other parameter matrices are permuted accordingly.

We have used the Gelman-Rubin convergence diagnostic \citep{gelman1992inference, gelman2013bayesian} to assess the convergence of the MCMC output from multiple random initializations. 
In particular, for the simulation setting corresponding to $n=1000$, we randomly initialize the parameters from their prior distributions $J=5$ times and then run five MCMC chains for each simulated dataset. 
For each MCMC chain, we ran the chain for 15000 iterations, discarding the first 10000 iterations as burn-in, and retaining every fifth sample post burn-in to thin the chain.
After collecting the posterior samples, we calculate the \textit{potential scale reduction factor} $R^2$ (i.e., the Gelman-Rubin statistic) of the model parameters. The median Gelman-Rubin statistics for the deep conditional probabilities $\bo\eta=(\eta_{kb})_{K\times B}$ and that for the deep latent class proportions $\bo\tau=(\tau_b)_{B\times 1}$ across the 5 chains are as follows,
$$
\text{median GR}(\bo\eta_{K\times B})
=\begin{pmatrix}
    1.0009  &  1.0013\\
    1.0016  &  1.0020\\
    1.0018  &  1.0015\\
    1.0016  &  1.0014
\end{pmatrix},
\quad
\text{median GR}(\bo\tau_{B\times 1})
=\begin{pmatrix}
1.0021\\
1.0021
\end{pmatrix}.
$$
The Gelman-Rubin statistics for other parameters are similarly well controlled and are omitted.
We also inspected the traceplots of the MCMC outputs and observed fast mixing of the MCMC chain after convergence. These observations justify running MCMC for 15000 iterations and discarding the first 10000 as burn-in, therefore we adopt these settings in all the numerical experiments. When applying the proposed method to other datasets, we also recommend calculating the Gelman-Rubin statistics and inspecting the traceplots to determine the appropriate number of overall MCMC iterations and burn-in iterations.

{Our Bayesian modeling of $\GG$ adopts rather uninformative priors with each entry $g_{j,k}\sim \text{Bernoulli}(\gamma)$ and 
we do not force $\GG$ to take a specific form (such as to include any identity submatrix as described in Proposition \ref{prop-2layer} for strict identifiability) but rather let the sampler freely explore the space of all binary matrices to estimate $\GG$. Our theory on identifiability and posterior consistency implies that the posteriors of both $\GG$ and other parameters concentrate around their true values as sample size grows. This is empirically verified in our simulation studies; Fig. \ref{fig-betapos}, \ref{fig-Q}, \ref{fig-rmse} show that as sample size $n$ grows, both the discrete $\GG$ and the continuous parameters are estimated consistently. We next elaborate on these findings from simulations.}

Fig.~\ref{fig-betapos} presents posterior means of the main-effect parameters $(\beta_{j1k};\, j\in[p],\, k\in[K_1])$ for category $c=1$, averaged across the 50 independent replications. The leftmost plot of Fig.~\ref{fig-betapos} shows that for a relatively small sample size $n=500$, the posterior means of $(\beta_{j1k})$ already exhibit similar structure as the ground truth in the rightmost plot.
Also, the true number of $K_1=4$ binary latent variables in the middle latent layer are revealed in  Fig.~\ref{fig-betapos}(a)-(b).
For the sample size $n=2000$, Fig.~\ref{fig-betapos}(b) shows the posterior means of $(\beta_{j1k})$ are very close to the ground truth.
The posterior means for other categories $c=2,3,4$ show similar patterns to those for $c=1$.
The estimated $(\beta_{j1k})$ are slightly biased toward zero for the smaller sample size
($n=500$) with bias less for the larger sample size ($n=2000$). 
Our results suggest that the binary graphical matrix that underlies $\{\beta_{jck}\}$ (that is, the sparsity structure of the continuous parameters) is easier to estimate than the nonzero $\{\beta_{jck}\}$ themselves.
That is, with finite samples, the existence of a link between each observed variable $y_j$ and each binary latent variable $\alpha_k$ is easier to estimate than the strength of such a link (if it exists).

\begin{figure}[h!]

	\centering
	\resizebox{.98\textwidth}{!}{%
    \includegraphics[height=3cm]{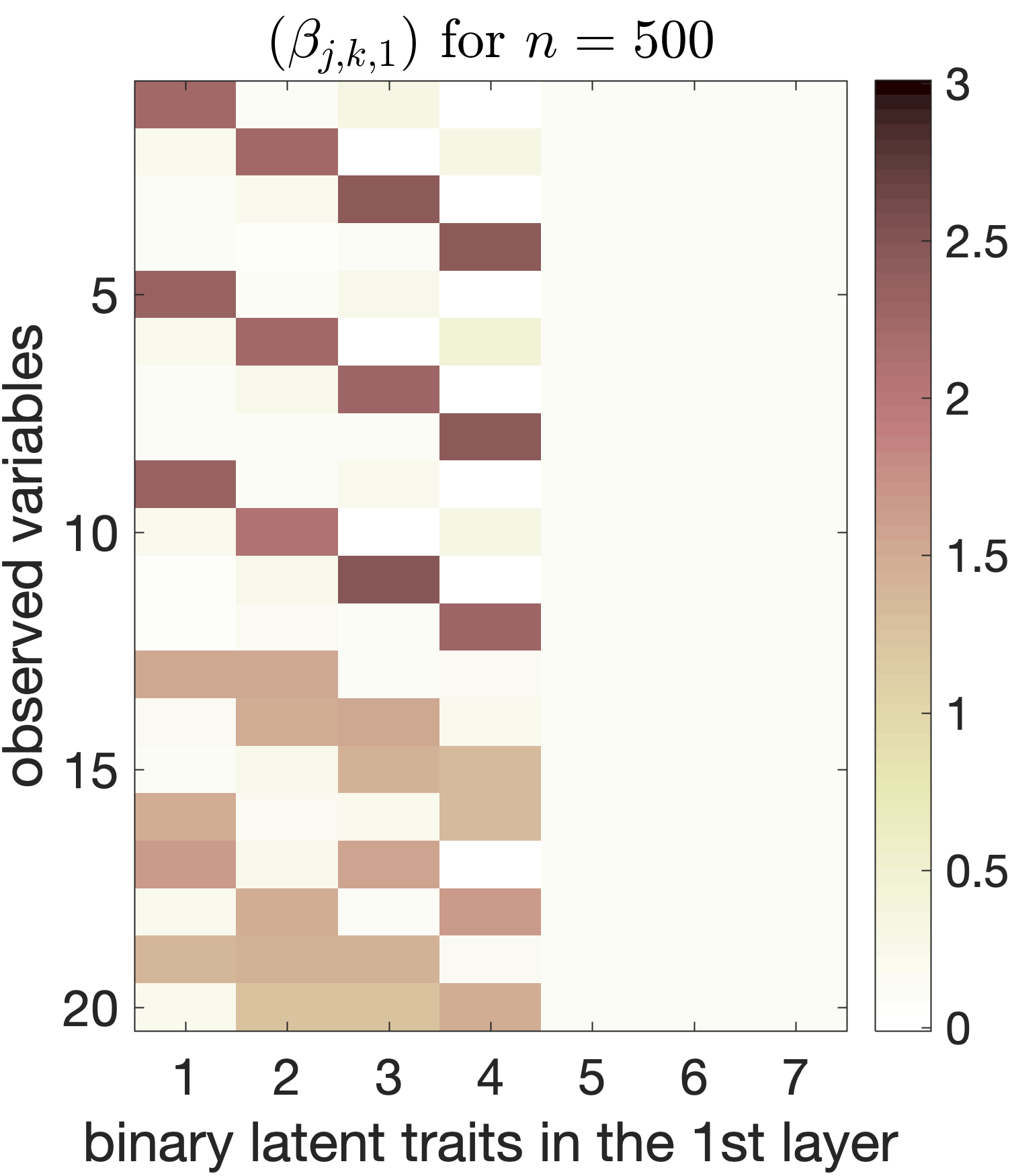}%
    \quad
    \includegraphics[height=3cm]{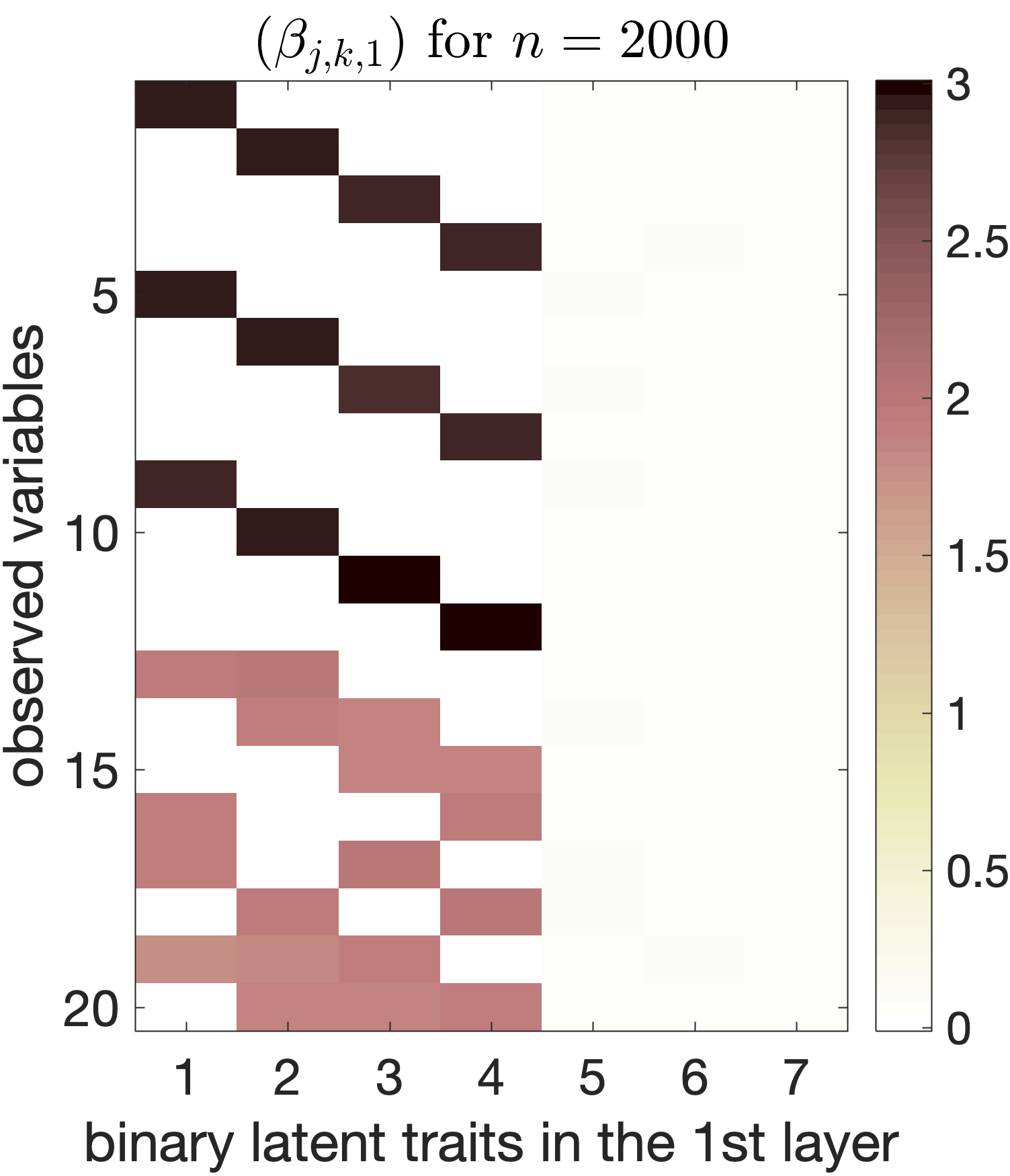}%
    \quad
    \includegraphics[height=3cm]{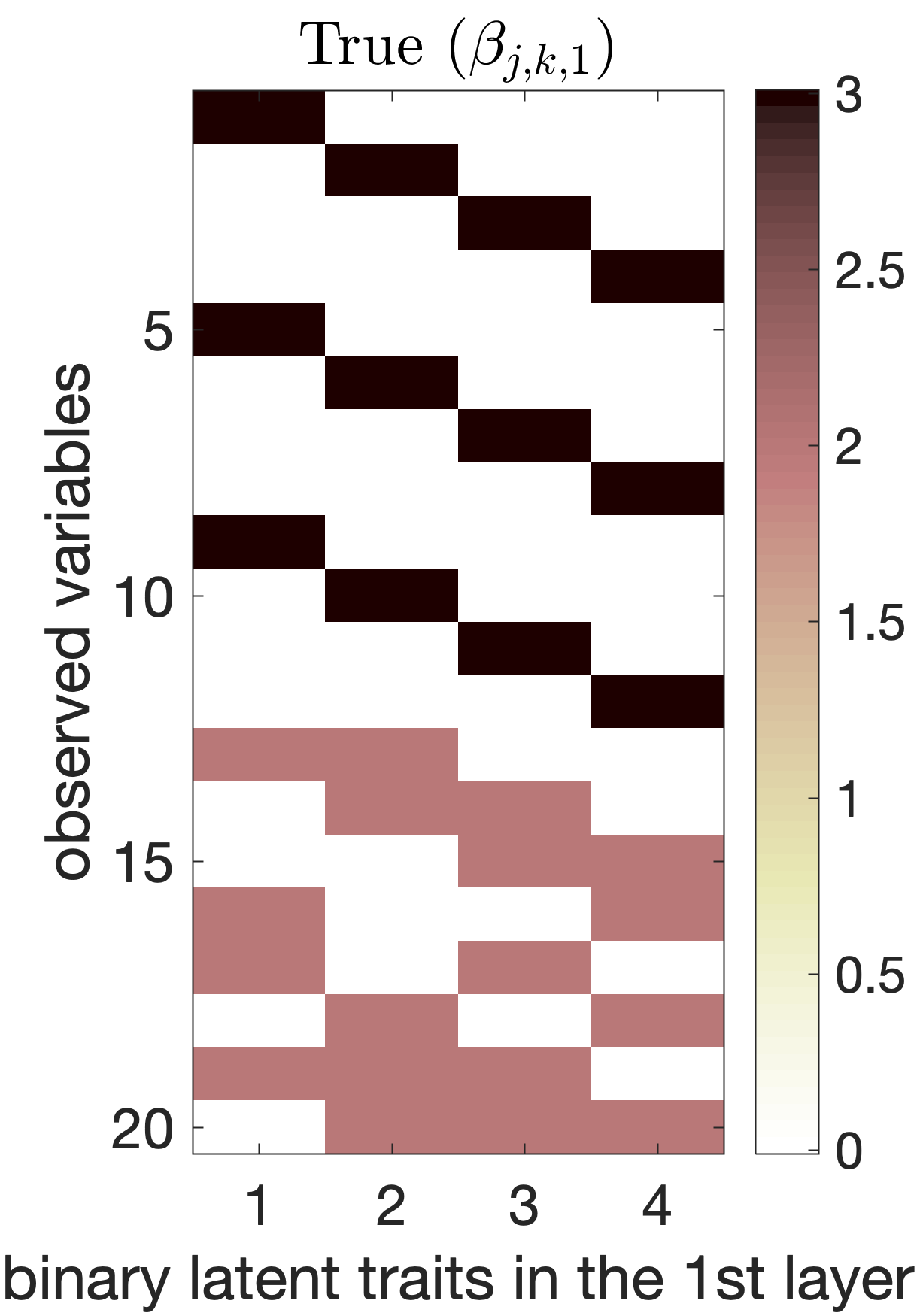}%
    }

	\caption{Posterior means of the main-effect parameters $\{\beta_{j1k}\}$ for response category $c=1$ averaged across 50 independent replications, for $n=500$ or $2000$. The $K_{\upper} = 7$ is taken in the CSP prior for posterior inference, while the true $K_1$ is 4 as shown in the rightmost plot.}
	\label{fig-betapos}
\end{figure}

To assess how the approach performs with an increasing sample size, we  consider eight sample sizes $n=250\cdot i$ for $i=1,2,\ldots,8$ under the same settings as above.  To compare the estimated structures to the simulation truth with $K_1=4$, we retain exactly four binary latent variables $k\in[K_{\upper}]$; choosing those having the 
largest average posterior variance $1/d \sum_{c=1}^d \sigma^2_{ck}$.
For the discrete structure of the binary graphical matrix $\GG^{(1)}$, we present mean estimation errors in Fig.~\ref{fig-Q}. Specifically, Fig.~\ref{fig-Q}(a) plots the errors of recovering the entire matrix $\GG^{(1)}$, Fig.~\ref{fig-Q}(b) plots the errors of recovering the row vectors of $\GG^{(1)}$, and Fig.~\ref{fig-Q}(b) plots the errors  of recovering the individual entries of $\GG^{(1)}$.
For sample size as small as $n=750$,  estimation of $\GG^{(1)}$ is very accurate across the simulation replications.
Notably, $n=750$ is much smaller than the total number of cells $d^p = 4^{20} \approx 1.1\times 10^{12}$ in the contingency table for the observed variables $\bo y$.
For continuous parameters $\bo\beta_0 = \{\beta_{jc0}\}$, $\bo\beta=\{\beta_{jck}\}$, and $\bo\eta = \{\eta_{kb}\}$ respectively, in Fig.~\ref{fig-rmse} we plot average root-mean-square errors (RMSE) of the posterior means versus sample size $n$. For each  $\bo\beta_0 = \{\beta_{jc0}\}$, $\bo\beta=\{\beta_{jck}\}$, and $\bo\eta = \{\eta_{kb}\}$, Fig.~\ref{fig-rmse} shows RMSE decreases with sample size, which is as expected given our posterior consistency result in Theorem \ref{thm-pos}.

{In simulations (and also the later real data analysis), we have checked the posterior samples collected from the MCMC algorithm and obtained the traceplots of various parameters in the model. 
By examining these, we have not observed label-switching issues in our numerical studies. 
In general, we would suggest one uses the traceplots of those mixture-component-specific parameters to examine whether there is a label switching issue. If there exists such issues, we recommend using the \texttt{R} package \texttt{label.switching}  \citep{papastamoulis2016label} for MCMC outputs to address the issue.}

\begin{figure}[h!]
    \centering
    %\begin{subfigure}[c]{0.325\textwidth}\centering
	\includegraphics[width=0.325\textwidth]{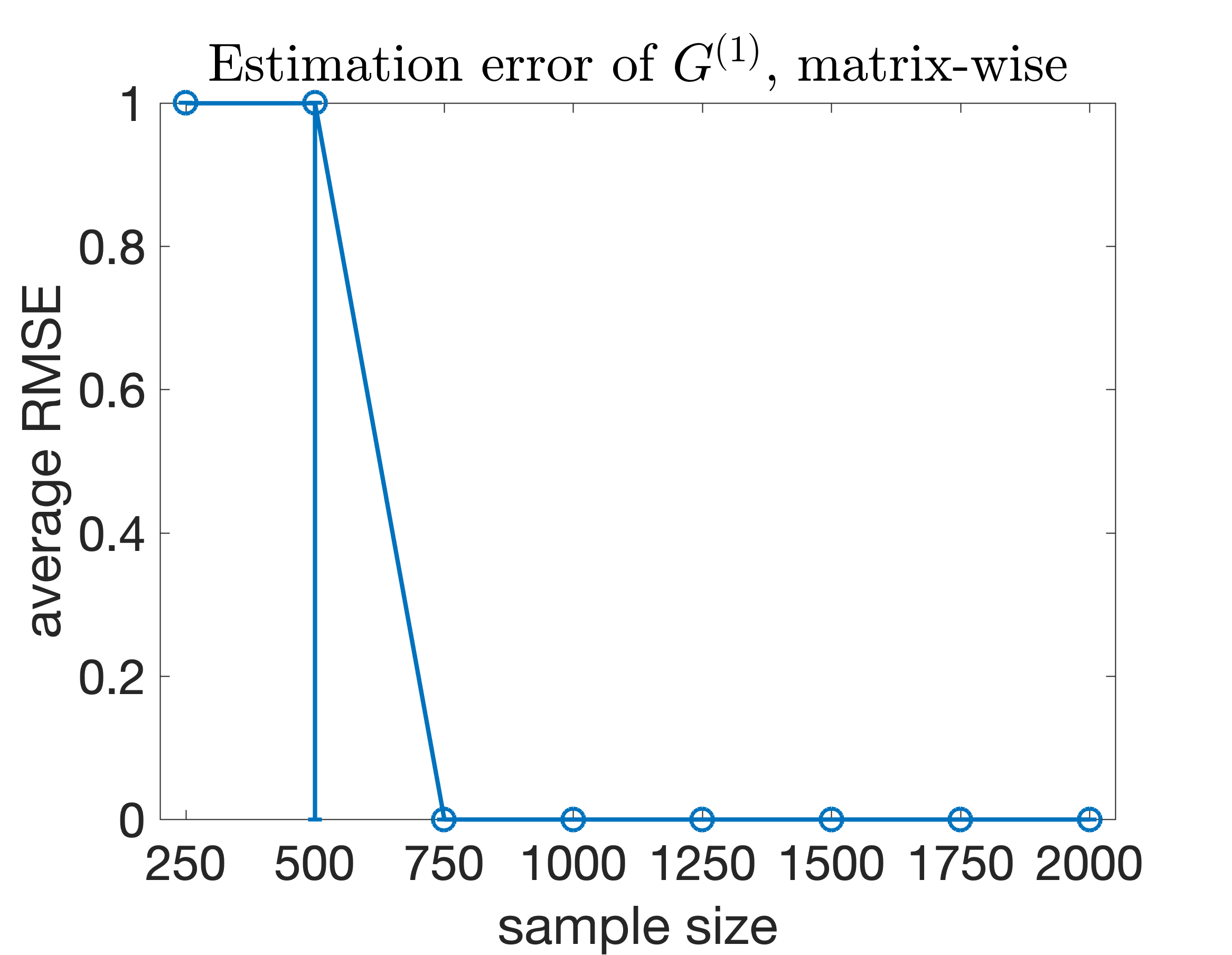}
% 	\caption{Mean errors of recovering the entire matrix $\GG^{(1)}$}
% 	\end{subfigure}
%     %
% 	\begin{subfigure}[c]{0.325\textwidth}\centering
	\includegraphics[width=0.325\textwidth]{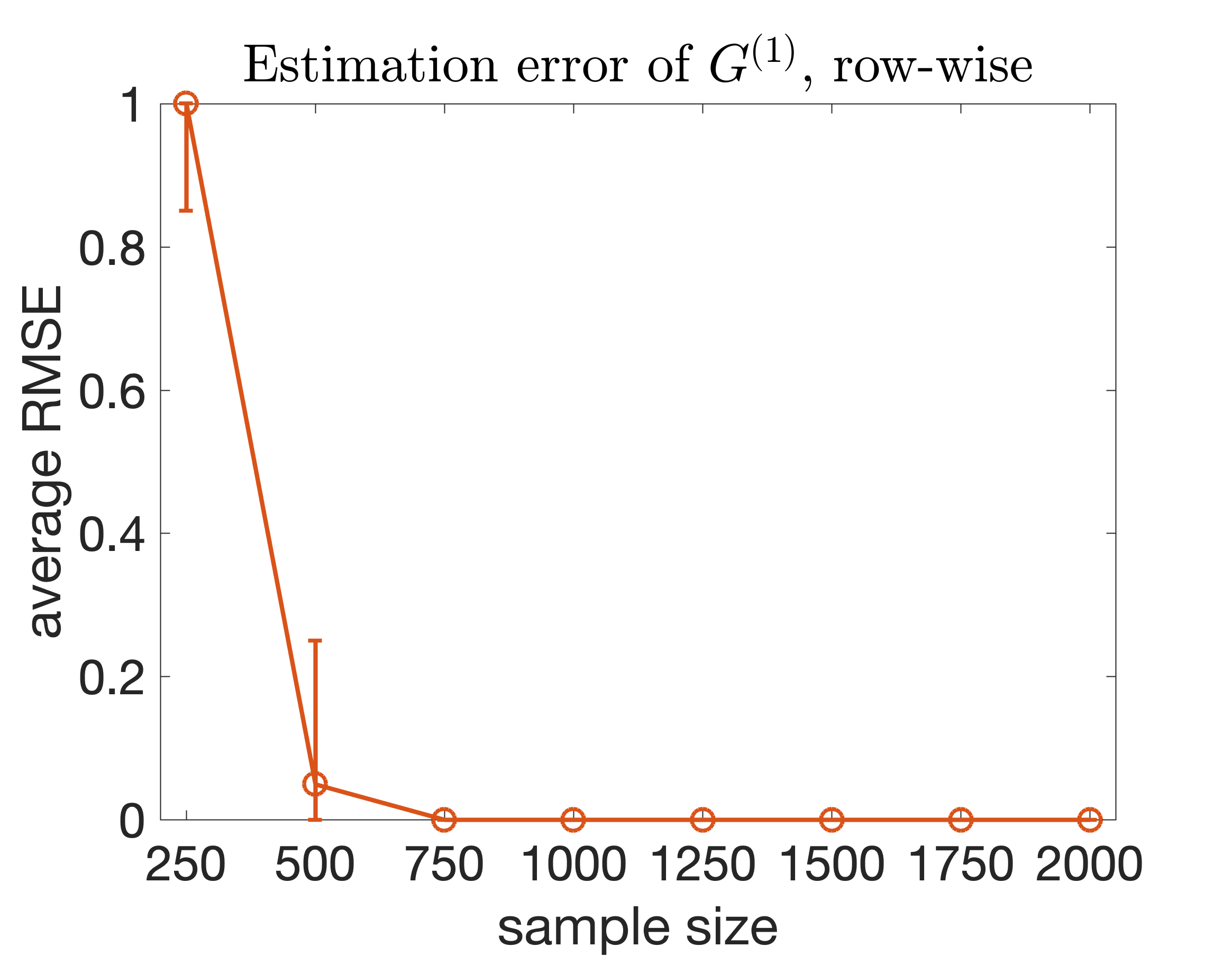}
% 	\caption{Mean errors of recovering the row vectors of $\GG^{(1)}$}
% 	\end{subfigure}
% 	%
% 	\begin{subfigure}[c]{0.325\textwidth}\centering
	\includegraphics[width=0.325\textwidth]{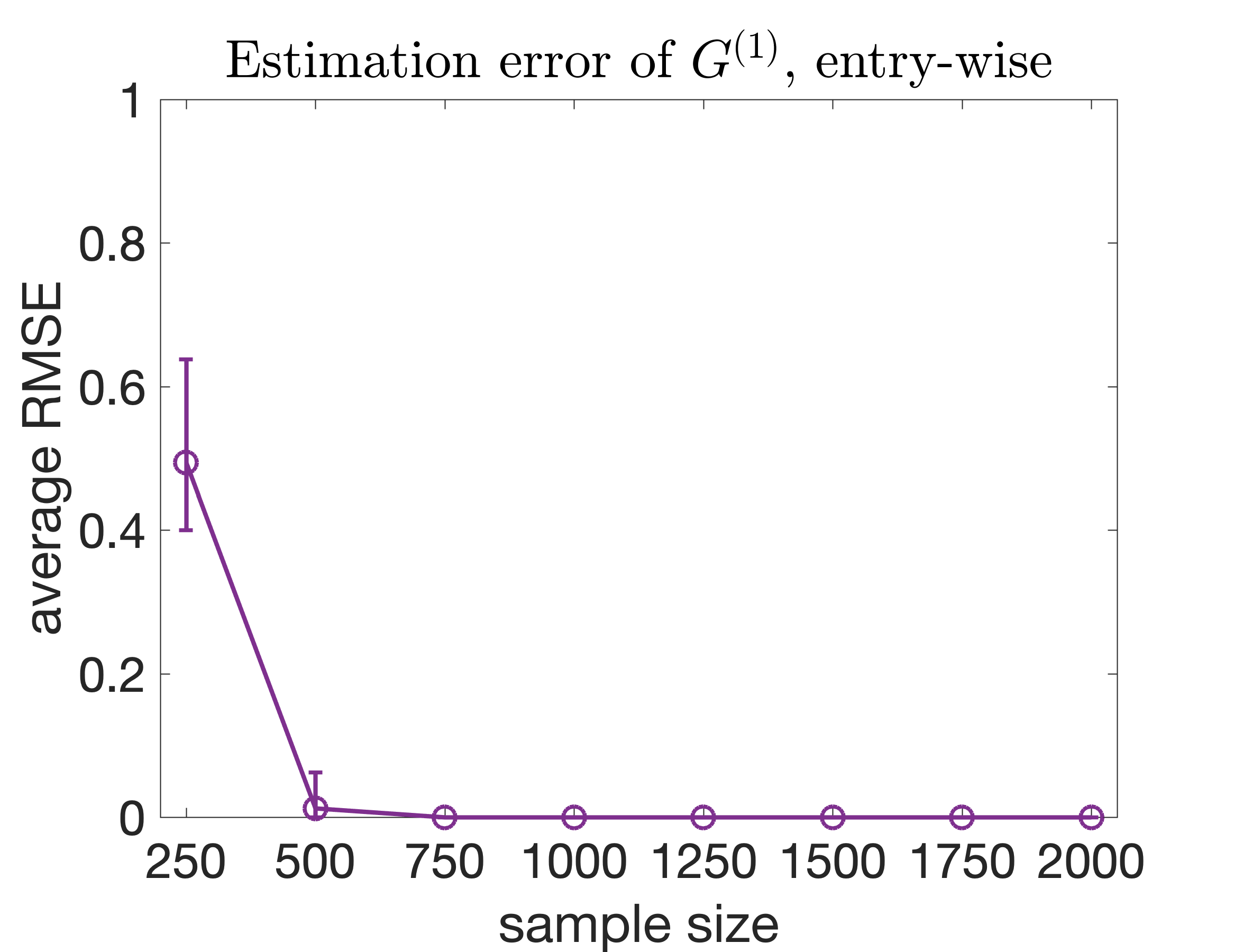}
% 	\caption{Mean errors of recovering the individual entries of $\GG^{(1)}$}
% 	\end{subfigure}
	\caption{Mean estimation errors of the binary graphical matrix $\GG^{(1)}$ at the matrix level (in (a)), row level (in (b)), and entry level (in (c)). The median, 25\% quantile, and 75\% quantile based on the 50 independent simulation replications are shown by the circle, lower bar, and upper bar, respectively.}
	\label{fig-Q}
\end{figure}

\begin{figure}[h!]
    \centering
    % \begin{subfigure}[c]{0.325\textwidth}\centering
	\includegraphics[width=0.325\textwidth]{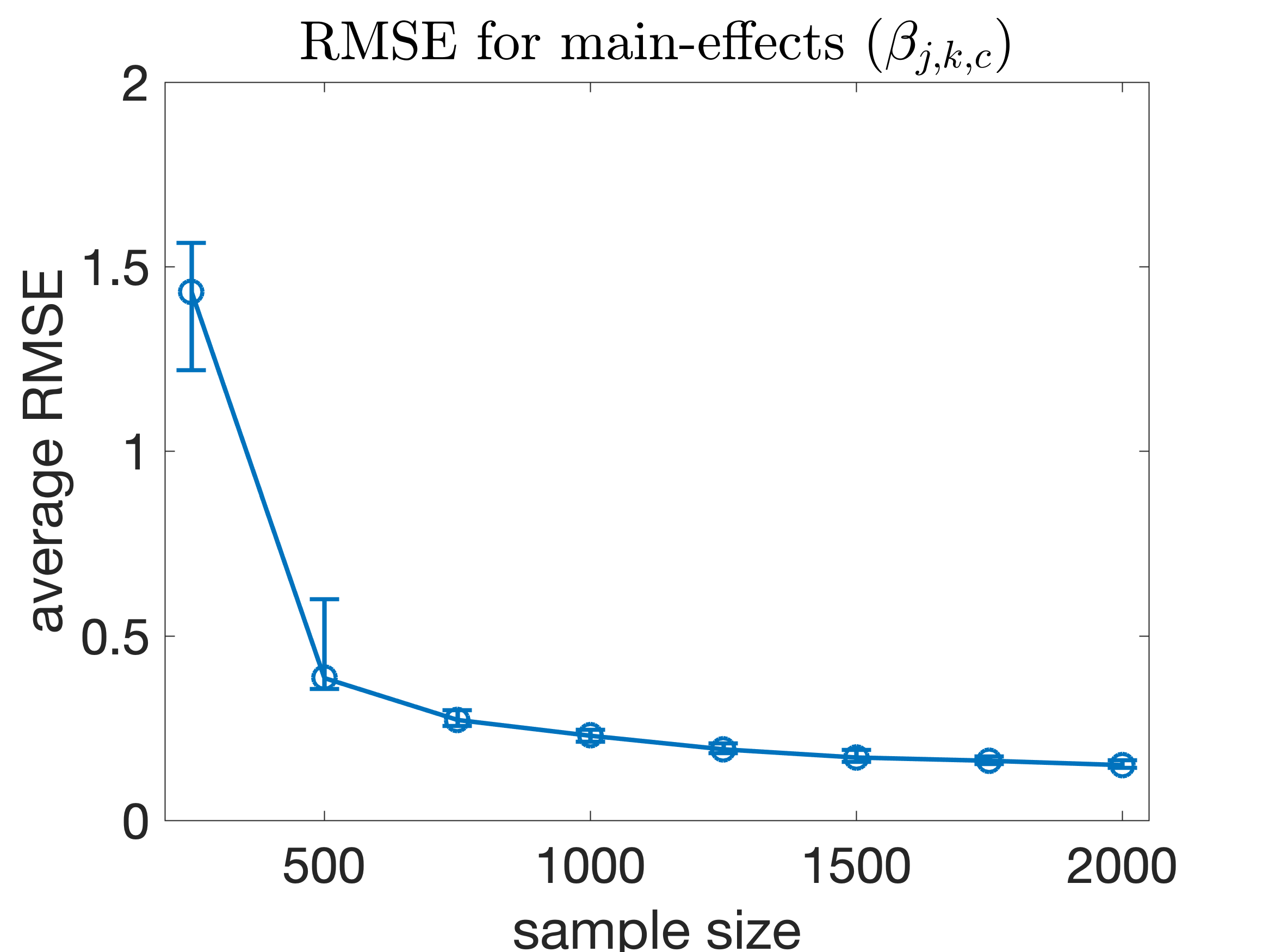}
% 	\caption{$\bo\beta=\{\beta_{jck}\}$}
% 	\end{subfigure}
% 	%
% 	\begin{subfigure}[c]{0.325\textwidth}\centering
	\includegraphics[width=0.325\textwidth]{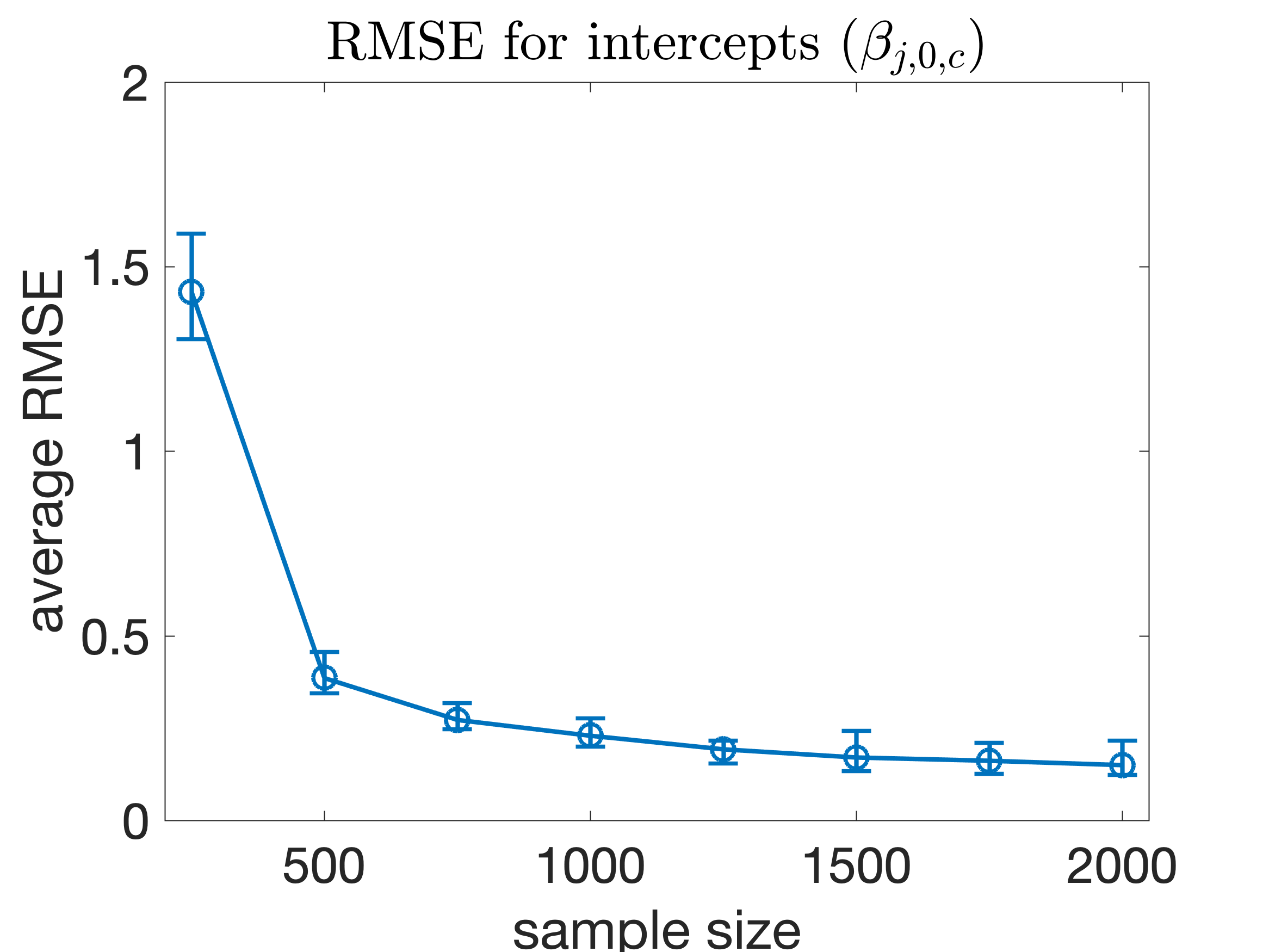}
	%\caption{$\bo\beta_0 = \{\beta_{jc0}\}$}
% 	\end{subfigure}
% 	\begin{subfigure}[c]{0.325\textwidth}\centering
	\includegraphics[width=0.325\textwidth]{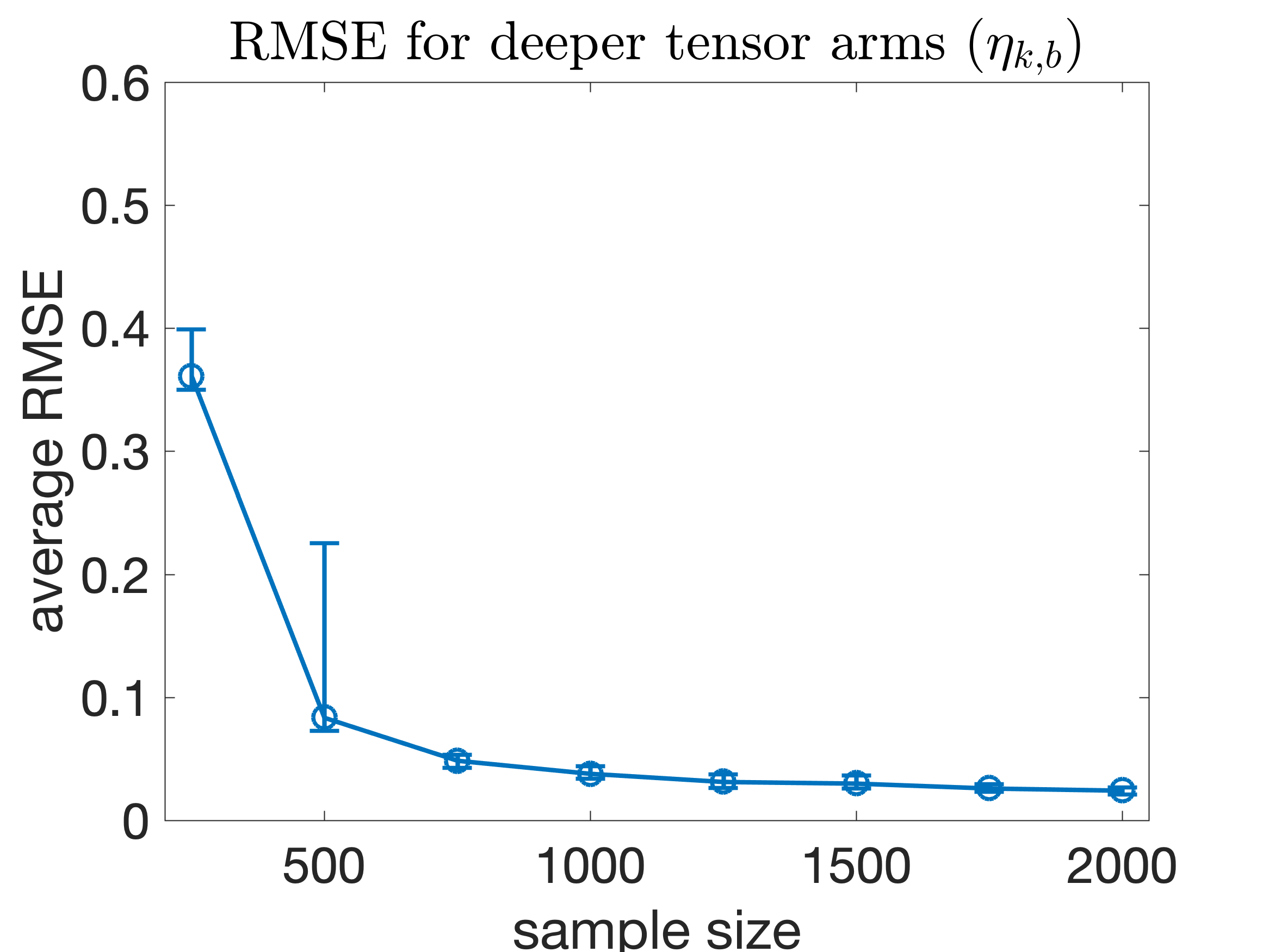}
% 	\caption{$\bo\eta = \{\eta_{kb}\}$}
% 	\end{subfigure}
    \caption{Average root-mean-square errors (RMSE) of the posterior means of the main-effects $\bo\beta$, intercepts $\bo\beta_0$, and deeper tensor arms $\bo\eta$ versus sample size $n=250\cdot i$ for $i\in\{1,2,\ldots,8\}$. The median, 25\% quantile, and 75\% quantile based on the 50 independent simulation replications are shown by the circle, lower bar, and upper bar, respectively.}
    \label{fig-rmse}
\end{figure}

\section{Application to DNA Nucleotide Sequence Data}\label{sec-data}
We apply our proposed two-latent-layer Bayesian Pyramid with the CSP prior to the Splice Junction dataset of DNA nucleotide sequences; the data are available from the UCI machine learning repository.
We also analyze another dataset of nucleotide sequences, the Promoter data, and present the results in the Supplementary Material.
The Splice Junction data consist of A, C, G, T nucleotides ($d=4$) at $p = 60$ positions for $n = 3175$ DNA sequences.
There are two types of genomic regions called  introns and exons; junctions between these two are called splice junctions \citep{nguyen2016dna}.
Splice junctions can be further categorized as 
(a) exon-intron junction; and (b) intron-exon junction.
The $n = 3175$ samples in the Splice dataset each belong to one of three types:  Exon-Intron junction (``EI'', 762 samples);  Intron-Exon junction (``IE'', 765 samples); and  Neither EI or IE (``N'', 1648 samples).
Previous studies have used supervised learning methods for predicting sequence type
\citep[e.g.,][]{li2003dna, nguyen2016dna}.
Here we fit the proposed two-latent-layer Bayesian Pyramid to the data in a completely unsupervised manner, with the sequence type information held out. 
We use the nucleotide sequences to learn discrete latent representations of each sequence, and then investigate whether the learned lower dimensional discrete latent features are interpretable and predictive of the sequence type.

We let the variable $z$ in the deepest latent layer have $B=3$ categories, inspired by the fact that there are three types of sequences: EI, IE, and N; but we do not use any information of which sequence belongs to which type when fitting the model to data.
% In the CSP prior, we take $K_{\upper}=7$. 
{As mentioned earlier, the upper bound for the binary latent variables $K_{\upper}$ can be set to $\ceil{p/3}$ or smaller in order to yield an identifiable Bayesian Pyramid model. In practice, when $p$ is large, we recommend starting with a relatively small $K_{\upper}$, inspecting the estimated active/redundant latent dimensions, and only increasing $K_{\upper}$ if all the latent dimensions are estimated to be active {\em a posteriori} (that is, if the posterior model of $K^\star$ defined in \eqref{eq-kstar} equals $K_{\upper}-1$). For this splice junction dataset with $p=60$, we start with $K_{\upper}=7$ for better computational efficiency; this $K_{\upper}$ is the same as that used in the simulations and already allows for $2^{K_{\upper}-1} = 64$ distinct latent binary profiles. We find that the posteriors select only $K^\star=5$ active latent dimensions; this suggests that it is not necessary to increase $K_{\upper}$ to a larger number for this dataset.}

We still run the Gibbs sampler for 15000 iterations, discarding the first 10000 as burn-in, and
retaining every fifth sample post burn-in to thin the chain.
Based on examining traceplots, the sampler has good mixing behavior.
As mentioned in the last paragraph, our method selects $K^\star = 5$ binary latent variables.
We index the saved posterior samples 
by $r\in\{1,\ldots,R=2000\}$, for each $r$ denote the samples of $\GG$ by $(g^{(r)}_{j,k};\, j\in[60], k\in[5])$. 
Similarly for each $r$, denote the posterior samples of the nucleotides sequences' latent binary traits by $(a^{(r)}_{i,k};\, i\in[3175], k\in[5])$, and denote those of the nucleotide sequences' deep latent category by $(z_i^{(r)};\, i\in[3175])$ where $z_i^{(r)}\in\{1,2,3\}$.
Define our final estimators 
$\widehat\GG = \left(\widehat g_{j,k}\right)$,  $\widehat{\mb A} = \left(\widehat \alpha_{i,k}\right)$, and $\widehat{\mb Z} = (\widehat z_{i,b})$ to be
\begin{align*}
    &\widehat g_{j,k} = \mathbbm{1}\left(\frac{1}{R} \sum_{r=1}^R g^{(r)}_{j,k} > \frac{1}{2}\right),\quad
    \widehat \alpha_{i,k} = \mathbbm{1}\left(\frac{1}{R} \sum_{r=1}^R a^{(r)}_{i,k} > \frac{1}{2}\right),\\[3mm]
    &
    \widehat z_{i,b} =
    \begin{dcases}
        1, & \text{if}~ b = \argmax_{b\in[B]} \frac{1}{R} \sum_{r=1}^R \mathbbm{1}\left(z^{(r)}_{i}=b\right);\\
        0, & \text{otherwise}.
    \end{dcases}
\end{align*}
The $\widehat g_{j,k}$, $ \widehat \alpha_{i,k}$, and $\widehat z_{i,b}$ summarize information of the element-wise posterior modes of the discrete latent  structures in our model.
The $60\times 5$  matrix $\widehat\GG$ depicts how the loci load on the binary latent traits, the $3175\times 5$ matrix $\widehat{\mb A}$ depicts the presence or absence of each binary latent trait in each  nucleotide sequence, and the $3175\times 3$ matrix $\widehat{\mb Z}$ depicts which deep latent group each nucleotide sequence belongs to. 
The $\widehat{\mb G}$, $\widehat{\mb A}$, and $\widehat{\mb Z}$ are all binary matrices, but the first two are binary feature matrices while the last one $\widehat{\mb Z}$ has each row having exactly one entry of ``1'' indicating group membership.
In Fig.~\ref{fig-splice}, the first three plots display the three estimated matrices $\widehat\GG$, $\widehat{\mb A}$, and $\widehat{\mb Z}$, respectively; and the last plot shows the held-out gene type information for reference.
As for the estimated loci loading matrix $\widehat\GG$, Fig.~\ref{fig-splice}(a) provides information on how the $p=60$ loci depend on the five binary latent traits. Specifically, we found that the first 27 loci show somewhat similar loading patterns and mainly load on the first four binary traits. Also, the middle 10 loci (from locus 28 to locus 37) are similar in loading on all five traits, and the last 23 loci (from locus 38 to locus 60) are similar in exhibiting sparser loading structures. 
Fig.~\ref{fig-splice}(b)--(d) show that the two matrices $\widehat{\mb A}$ and $\widehat{\mb Z}$ corresponding to the $n=3175$ nucleotide sequences exhibit a clear pattern of clustering, which aligns well with the known but held-out junction types EI, IE, and N.

\begin{figure}[h!]
	\centering
	\bigskip
	\resizebox{\textwidth}{!}{%
    \includegraphics[height=3cm]{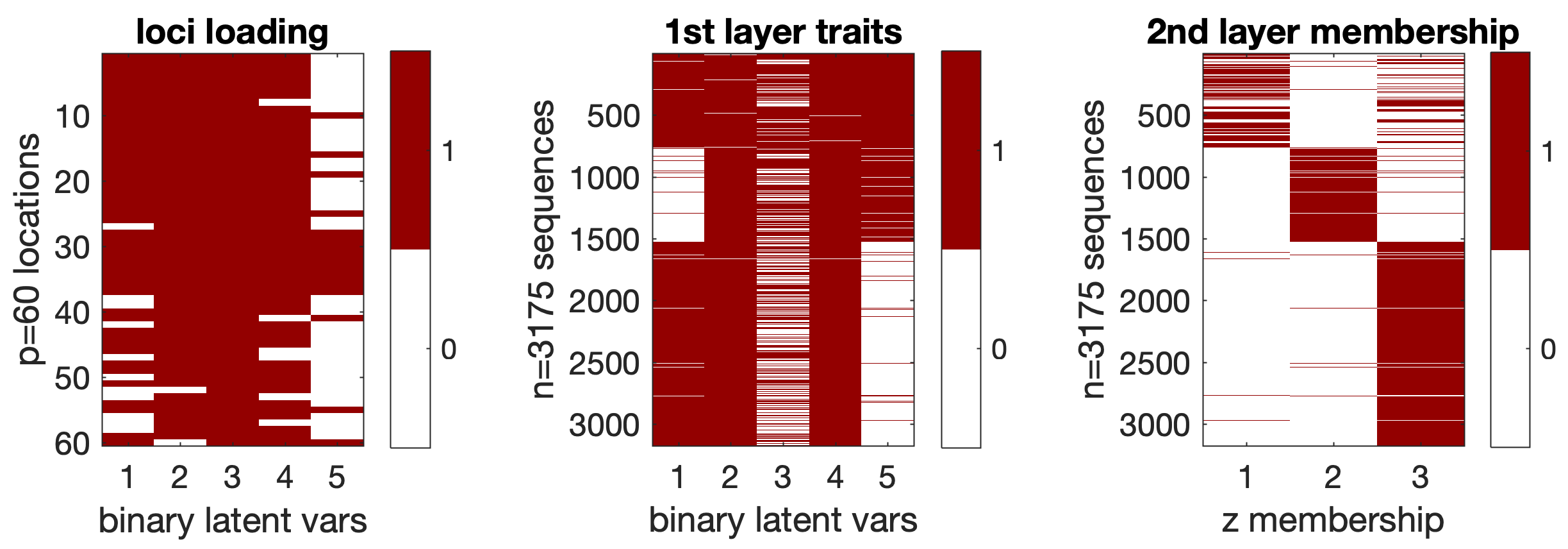}%
    \quad
    \includegraphics[height=3cm]{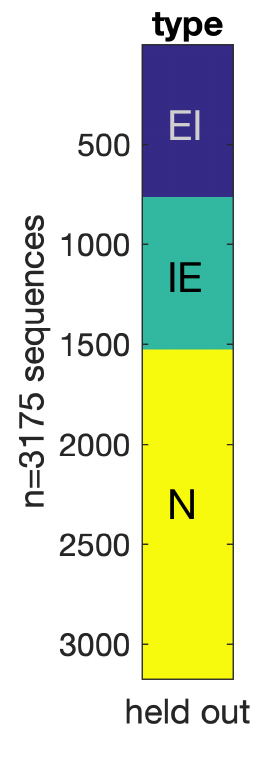}%
    }
    
    \begin{minipage}[c]{0.26\textwidth}\centering
    (a) $\widehat\GG_{60\times 5}$
    \end{minipage}
    ~
    \begin{minipage}[c]{0.26\textwidth}\centering
    (b) $\widehat{\mb A}_{3175\times 5}$
    \end{minipage}
    ~
    \begin{minipage}[c]{0.26\textwidth}\centering
    (c) $\widehat{\mb Z}_{3175 \times 3}$
    \end{minipage}
    ~
    \begin{minipage}[c]{0.15\textwidth}\centering
    (d) held-out
    \end{minipage}
    
%	\begin{minipage}[c]{0.48\textwidth}
%	\includegraphics[width=\linewidth]{figures/DIC_splice_K}
%	\end{minipage}
%	\hfill
%	\begin{minipage}[c]{0.48\textwidth}
% 	\includegraphics[width=0.7\textwidth]{figures/splice_K4_QA}
% \includegraphics[width=\textwidth]{figures/splice_csp_Kupper_4.png}
%	\end{minipage}
	\caption{Splice junction data analysis under the CSP prior with $K_{\upper} = 7$. 
	Plots are presented with the $K^\star = 5$ binary latent traits selected by our proposed method a posteriori. 
	After applying the rule-lists approach to deterministically match the latent features to the gene types as in \eqref{eq-rulelist}, the accuracies for predicting the gene types EI, IE, N are all above 95\%.}
	\label{fig-splice}
	%0.9716, 0.9505, 0.9480.}
\end{figure}

To formally assess how the latent discrete features learned by the proposed method perform in downstream prediction, we apply the ``rule lists'' classification approach in \cite{angelino2017rule} to the estimated latent features in $\widehat{\mb A}$ and $\widehat{\mb Z}$ for $n=3175$ nucleotide sequences.
The rule-lists approach is an interpretable classification method based on a categorical feature space, and it finds simple and deterministic {rules} of the categorical features in predicting a (binary) class label. 
For each instance $i\in\{1,\ldots,n=3175\}$, we define the categorical features to be the $8$-dimensional  vector $ \widehat \xx_i = (\widehat z_{i,1},\,\widehat z_{i,2},\,\widehat z_{i,3},\,\widehat a_{i,1},\,\ldots,\widehat a_{i,5})$.
The $\widehat \xx_i$ is concatenated from $\widehat{\bo z}_i$ and $\widehat {\bo a}_i$ and is a feature vector of binary entries.
Denote the ground-truth nucleotide sequence types by $\bo t=(t_i;\; 1\leq i \leq 3175)$, then $t_i = \text{EI}$ for $i=1,\ldots,762$, 
$t_i = \text{IE}$ for $i=763,\ldots,1527$, and $t_i = \text{N}$ for $i=1528,\ldots,3175$.
Recall that the information of $t_i$'s are not used in fitting our Bayesian Pyramid to obtain the latent features $\widehat{\bo x}_i$'s.
We use the Python package \verb|CORELS| for the rule-lists method with $\bo x_i$'s and $t_i$'s as input, and find rules that match $\widehat\xx_i$'s to $t_i$'s. %97.35\%, 97.64\%, 97.13\% for the EI, IE, N, respectively.
Specifically, these deterministic rules given by \verb|CORELS| are:
\begin{eqnarray}\label{eq-rulelist}
    \widehat t_i^{\;\text{EI}}
    &=&
    \mathbbm{1}(\widehat a_{i,1} = 1 ~\text{and}~ \widehat a_{i,5} = 1);\\ \notag
    \widehat t_i^{\;\text{IE}}
    &=&
    \mathbbm{1}(\widehat z_{i,2} = 1);\\ \notag
    \widehat t_i^{\;\text{N}}
    &=& 
    \mathbbm{1}(\widehat z_{i,2} = 0 ~\text{and}~ \widehat a_{i,5} = 0)
    % = 
    % \begin{dcases}
    %     \text{EI}, & \text{if}~ \widehat a_{i,1} = 1 ~\text{and}~ \widehat a_{i,5} = 1;\\
    %     \text{IE}, & \text{if}~ \widehat z_{i,2} = 1;\\
    %     \text{N}, & \text{if}~ \widehat z_{i,2} = 0 ~\text{and}~ \widehat a_{i,5} = 0,\\
    % \end{dcases}
\end{eqnarray}
Eq.~\eqref{eq-rulelist} gives a very simple and explicit rule depending on $\widehat {z}_{i,2}$, $\widehat a_{i,1}$, and $\widehat a_{i,5}$ for each nucleotide sequence $i$. 
This simple rule can be empirically verified by comparing the middle two plots to the rightmost plot in Fig.~\ref{fig-splice}. 
In \eqref{eq-rulelist}, the rules for EI and IE are not mutually exclusive, but those for EI and N are mutually exclusive and the same holds for IE and N.
Therefore, to obtain mutually exclusive rules for the three class labels EI, IE, and N based on \eqref{eq-rulelist}, we can simply define the following two types of labels $\widehat {\bo t}^{\;\star} = (\widehat t_i^{\;\star};\; i\in[n])$ and $\widehat {\bo t}^{\;\dagger} = (\widehat t_i^{\;\dagger};\; i\in[n])$,
\begin{align}\label{eq-stardagger}
    \widehat t_i^{\;\star} = 
    \begin{dcases}
    \text{EI}, & \text{if}~
    \widehat t_i^{\;\text{EI}}=1~\text{and}~
    \widehat t_i^{\;\text{IE}}=0;\\
    \text{EI}, & \text{if}~
    \widehat t_i^{\;\text{IE}}=1;\\
    \text{N}, & \text{if}~
    \widehat t_i^{\;\text{N}}=1;
    \end{dcases}
    \quad\text{or}\quad
    \widehat t^{\;\dagger}_i = 
    \begin{dcases}
    \text{EI}, & \text{if}~
    \widehat t_i^{\;\text{EI}}=1;\\
    \text{EI}, & \text{if}~
    \widehat t_i^{\;\text{IE}}=1~\text{and}~
    \widehat t_i^{\;\text{EI}}=0;\\
    \text{N}, & \text{if}~
    \widehat t_i^{\;\text{N}}=1.
    \end{dcases}
\end{align}
\begin{figure}[h!]
    \centering
    % \includegraphics[width=\textwidth]{figures/rmse.png}
    % \begin{subfigure}[c]{0.45\textwidth}\centering
	\includegraphics[width=0.37\textwidth]{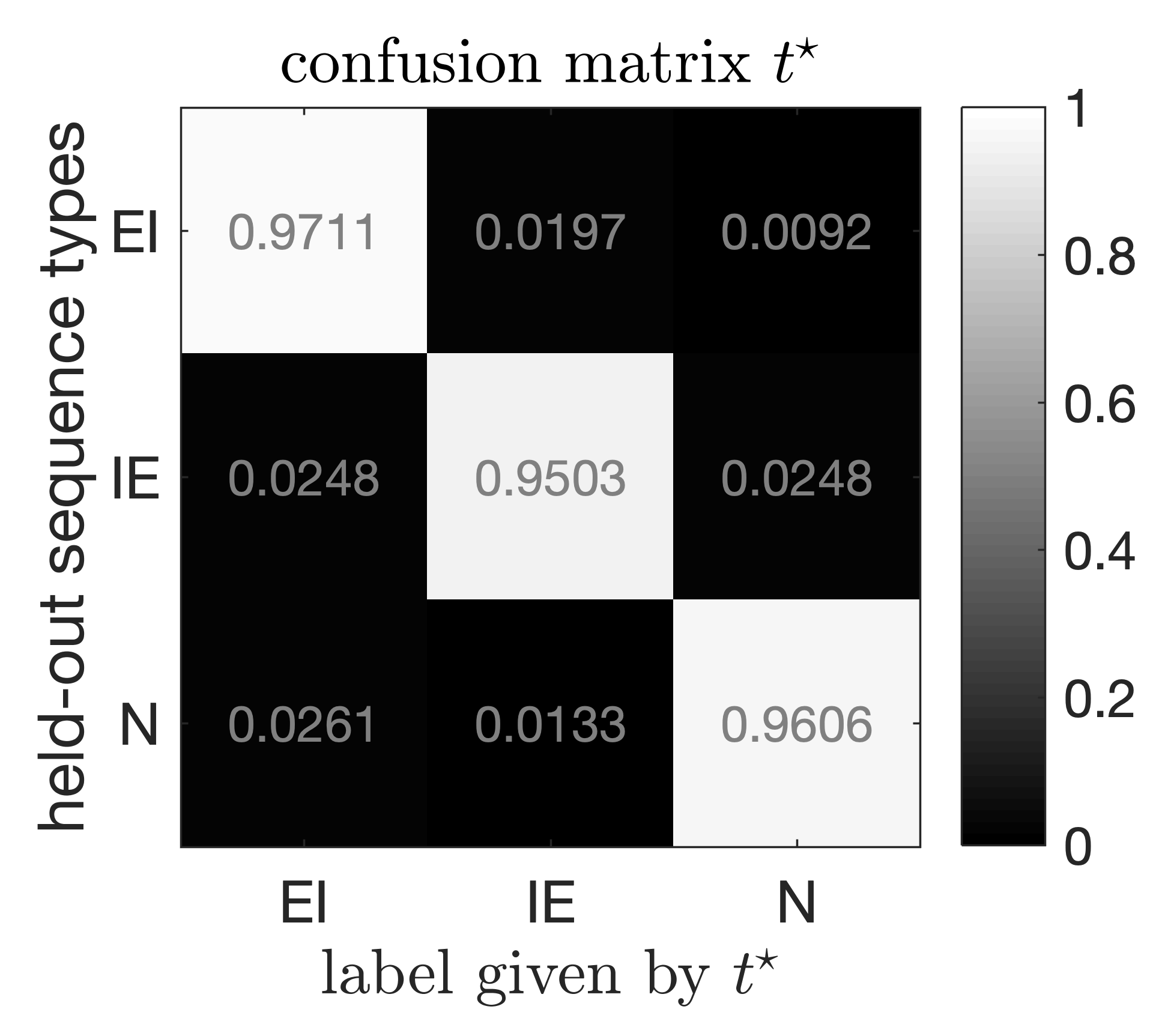}
	\qquad
% 	\caption{confusion matrix for the labels given by $\widehat {\bo t}^{\;\star}$}
% 	\end{subfigure}
% 	%
% 	\quad
% 	\begin{subfigure}[c]{0.45\textwidth}\centering
	\includegraphics[width=0.37\textwidth]{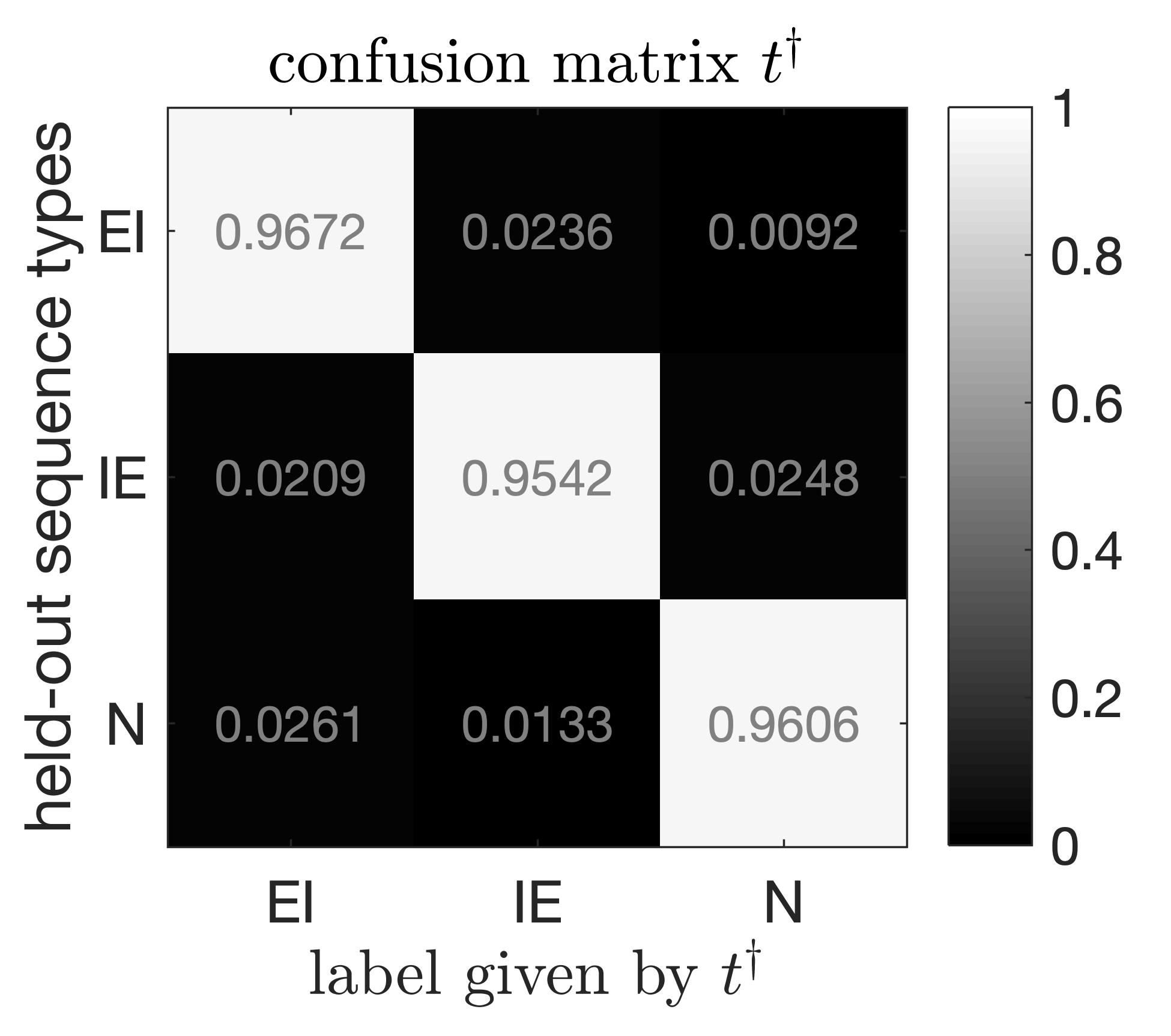}
% 	\caption{confusion matrix for the labels given by $\widehat {\bo t}^{\;\dagger}$}
% 	\end{subfigure}
	\caption{Downstream interpretable classification results for splice junction data following fitting the two-latent-layer Bayesian Pyramid.
	Prediction performance of the learned lower-dimensional discrete latent features are summarized in two confusion matrices, corresponding to $\widehat {\bo t}^{\;\star}$ and $\widehat {\bo t}^{\;\dagger}$ defined in \eqref{eq-stardagger}.
	}
	\label{fig-confmat}
\end{figure}

For the two types of labels $\widehat {\bo t}^{\;\star}$ and $\widehat {\bo t}^{\;\dagger}$ in \eqref{eq-stardagger}, we provide their corresponding normalized confusion matrices in Fig.~\ref{fig-confmat} along with their heatmaps. Fig.~\ref{fig-confmat} shows that both $\widehat {\bo t}^{\;\star}$ and $\widehat {\bo t}^{\;\dagger}$ have very high prediction accuracies for the sequence types $\bo t$, with all the diagonal entries above $95\%$.
The superb prediction performance as shown in Fig.~\ref{fig-confmat} implies that the learned discrete latent features of nucleotide sequences are interpretable and very useful in the downstream task of classification.
For the Splice Juntion dataset, such accuracies are even comparable to the state-of-the-art prediction performance given by fine-tuned convolutional neural networks given in \cite{nguyen2016dna}, which are also between 95\% and 97\%.

In the above data analysis, we have fixed $B=3$ inspired by the fact that there are three types of splice junctions. 
Alternatively, we can use the \textit{overfitted mixture framework} of \cite{rousseau2011overfit} to accommodate an unknown $B$.
Overfitted mixtures are finite mixture models with an unknown number of mixture components $B$, where only an upper bound $B_{\text{upper}}$ of $B$ is known. 
In such scenarios, \cite{rousseau2011overfit} proposed to use shrinkage priors for the mixture proportion parameters, such as Dirichlet priors with small Dirichlet hyperparameters.
More specifically, denote the mixture proportion parameters corresponding to the $B_{\text{upper}}$ ``overfitted'' mixture components by $\bo\pi := (\pi_1,\ldots,\pi_{B_{\text{upper}}})$, then $\bo\pi$ lives on the $(B_{\text{upper}}-1)$-dimensional probability simplex. The overfitted mixture framework in \cite{rousseau2011overfit} guarantees that with the prior $\bo\pi\sim \text{Dirichlet}(a_1, \ldots,  a_{B_{\text{upper}}})$ where the Dirichlet parameters are small enough with respect to the dimension of the observations, the redundant mixture components will be ``emptied out'' in the posterior asymptotically. Hence, the $B$ true mixture components underlying the data will be identified from the posterior distribution. In the Supplementary Material, we reanalyze the real data specifying $B=5$ and simulation studies with an overfitted $B$, and achieve favorable results.

For comparison, we also applied/implemented two alternative discrete latent structure models, including: 
% \begin{enumerate}
    % \item[(a)]  
(a) the latent class model, with a single univariate discrete latent variable behind the observed data layer; and
    % \item[(b)]  
(b) a single-latent-layer model, with one layer of \textit{independent} binary latent variables behind the data layer.
% \end{enumerate}
For model (a), we applied the nonparametric Bayesian method in \cite{dunson2009} to the splice junction dataset and extracted $k_{\text{DX}}=10$ latent classes (default in \cite{dunson2009}; increasing $k_{\text{DX}}$ beyond 10 is not considered because in the current 10 latent classes there are already empty classes not occupied by any nucleotide sequence).
For model (b), we implemented a new Gibbs sampler  (see description in S.9.3 in the Supplementary Material), still employing the cumulative shrinkage process prior as used for Bayesian Pyramids.
Then based on the learned lower-dimensional latent representations, we still apply the ``rule-lists'' classification approach in \cite{angelino2017rule} (via the Python package \texttt{CORELS}) for each of the three labels EI, IE, and N via binary classification. 
{In addition to the above two competitors (abbreviated as ``DX2009'' and ``IndepBinary'', respectively), we also consider a benchmark approach of directly applying the rule-lists classifier to raw DNA nucleotide sequences. Since the rule-lists classifier in the \texttt{CORELS} package  only applies to binary predictors, we first convert the DNA nucleotide sequences of A, G, C, T into longer sequences containing binary indicators of whether each loci is ``A'', or ``G'', or ``C''; after this, each original $60$-dimensional nucleotide sequence is converted to a binary vector of length $180$. We then apply the rule-lists classifier using the same setting of learning at most three ``rules'' (i.e., \texttt{max\_card}$\;=3$ in the \texttt{CORELS} package) as all the other three latent variable methods: DX2009, IndepBinary, and BayesPyramid.
We obtain the classification accuracy in Table \ref{tab-acc}.}

{The left three columns in Table \ref{tab-acc} list accuracies on the training set containing 80\% of the samples, and the right three columns on the test set containing the remaining 20\% of the samples.
The test accuracies are very close to the training ones, indicating satisfactory generalizability.
The reason is that the maximum number of rules learned for each model is limited to be at most three, i.e., \texttt{max\_card}$\;=3$ in the \texttt{CORELS} package.
Such a small number of rules in a classifier do not overfit the training data and hence generalize well to the test data. 
% We remark that increasing \texttt{max\_card} to three yield the same accuracies for BayesPyramid and Raw DNA sequences, while increasing \texttt{max\_card} to four causes the \texttt{CORELS} package unable to produce results for the raw DNA sequences, likely due to the high-dimensionality of the search space of rules. Therefore, we focus on \texttt{max\_card}$\;=3$ and present the corresponding results in Table \ref{tab-acc}.
We remark that when increasing \texttt{max\_card} beyond three, the downstream classification accuracy of BayesPyramid remains the same as those in Table \ref{tab-acc}; however, for the raw nucleotide sequences, 
the \texttt{CORELS} package cannot complete the execution of the classifier with \texttt{max\_card}$\;=4$, likely due to the high-dimensionality of the search space of rules. Therefore, we present the results of \texttt{CORELS} classification with \texttt{max\_card}$\;=3$ in Table \ref{tab-acc}.}

{Remarkably, Table \ref{tab-acc} shows that all the three unsupervised learning methods learn latent features that are more predictive of the sequence type than the raw sequences themselves, among which the Bayesian Pyramid is the apparent top performer. This observation indicates that the latent variable approaches, especially our multilayer Bayesian Pyramids, can ``denoise'' the raw sequence data and return more useful features for downstream tasks.}
The model in \cite{dunson2009} (``DX2009'') and the single-layer-independent-latent model (``IndepBinary'') are special cases of Bayesian Pyramids, and both have excellent performance on the splice data.  However, Bayesian Pyramids significantly decreased the test set missclassification rate of IndepBinary by 31--60\%, of DX2009 by 51--68\%, {and of raw nucleotide sequences by 81--87\%.}

\begin{table}
\caption{\label{tab-acc}Downstream classification accuracy for Splice data.
}
\centering
\fbox{%
\resizebox{0.95\textwidth}{!}{
\begin{tabular}{lcccccc}
 & EI (train) & IE (train) & N (train) & EI (test) & IE (test) & N (test) \\
\hline
Raw nucleotide nucleotides & 0.909 & 0.854 & 0.840 &  0.916 & 0.866 & 0.839 \\
DX2009 & 0.946 & 0.955 & 0.918 & 0.950 &  0.965 & 0.929  \\
IndepBinary & 0.964 & 0.955 & 0.956 & 0.965 & 0.957 & 0.965  \\
BayesPyramid & 0.971 & 0.975 & 0.970 &  0.984 & 0.983 & 0.976  \\
\end{tabular}}
}
\end{table}

%%% Discussion %%%
% \newpage
\section{Discussion}\label{sec-disc}
In this article, we have proposed Bayesian Pyramids, a general family of discrete multilayer latent variable models for multivariate categorical data. Bayesian Pyramids cover multiple existing statistical and machine learning models as special cases, while allowing for various different assumptions on the conditional distributions. Our identifiability theory is key in providing reassurance that one can reliably and reproducibly learn the latent structure, guarantees that are lacking from almost all of the related literature.
%%%
The proposed Bayesian approach has excellent performance in the simulations and data analyses, and shows good computational performance and a surprising ability to accurately infer meaningful and useful latent structure.  
There are immediate applications of the proposed Bayesian Pyramid approach in many other disciplines.
For example, in ecology and biodiversity research, joint species distribution modeling is a very active area \citep[see a recent book][]{ovaskainen2020eco}. The data consist of high-dimensional binary indicators of presence or absence of different species at each sampling location.  In this contest, Bayesian Pyramids provide a useful approach for inferring latent community structure of biologic and ecological interest.

This work focuses on the {unsupervised} setting and uses {latent variable approaches}.
In unsupervised cases, there is not a specific outcome/response associated with each data point, and the general goal is to discover ``interesting'' patterns in data \citep{murphy2012machine}.
Such unsupervised learning problems are generally considered more challenging and less well studied than supervised ones.  Latent variables can capture semantically interpretable constructs, and marginalizing out latents induces great flexibility in the resulting marginal distribution of observables.
Bayesian Pyramids provide unsupervised feature discovery tools for learning latent architectures underlying multivariate data, yielding insight into data generating processes, performing nonlinear dimension reduction, and extracting useful latent representations.

{In order to realize the great potential of latent variable approaches to unsupervised learning, identifiability issues must be addressed so that reliable latent constructs can be reproducibly extracted from data.
%%%%
Especially, if one wishes to {interpret the latent} representations and parameters estimated from a latent variable model and use them in downstream analyses, then identifiability is a prerequisite for making such interpretation meaningful and reproducible.
In this sense, identifiability is necessary for interpretability of the latent structures.
We remark that such interpretability considerations differ from the goals of learning interpretable treatment rules for some outcome in complex supervised learning settings \citep[see, e.g.][]{kitagawa2018ewm, semenova2021debiased}.
In modern machine learning, although powerful deep latent variable models have been proposed, identifiability issues have rarely been considered and the design of the latent architecture is often guided by heuristics.
In this work, our identifiability theory directly inspires how to specify the deep latent architecture in a Bayesian Pyramid: the identifiability conditions on matrices $\GG^{(m)}$ (in Theorem 4)  inspire us to specify a ``pyramid'' structure featuring fewer and fewer latent variables deeper up the hierarchy ($J\geq 3K_1;\; K_1\geq 3 K_2,\; \ldots$) and sparse graphical connections between layers.
}

% There are future research directions that can unlock the applied potential of the proposed framework. 
% In the future, it is worth further advancing and refining scalable algorithms and inferential frameworks for implementing Bayesian Pyramids beyond the two-latent-layer case.
{In the future, it is worth further advancing and refining scalable algorithms for Bayesian Pyramids beyond the two-latent-layer case.
Methodologically, our current Gibbs sampling procedure can be readily extended to multiple latent layers, because non-adjacent layers in a Bayesian Pyramid are conditionally independent and the current Gibbs updates for shallower layers can be similarly applied to deeper ones.
In terms of the computational performance, however, efficiently sampling discrete latent variables via MCMC is generally more challenging than sampling continuous ones. In a Bayesian Pyramid, binary entries of the graphical matrices $\mathbf G^{(1)}, \mathbf G^{(2)}, \ldots$ act similarly as covariate selection indicators in (logistic) regression problems, yet are more difficult to estimate because the ``covariates'' $\aaa^{(1)},\aaa^{(2)},\ldots$ themselves are latent and discrete.
This fact can cause unsatisfactory mixing behavior of naive extensions of the Gibbs sampler to deeper models. 
For example, sampling to update an entry in a deeper graphical matrix $\mathbf G^{(d)}$ where $d\geq 2$ can cause a substantial shift of the posterior space for the associated continuous parameters in downward layers.
Future research is warranted to advance the MCMC methodology or explore complementary methods for estimating the discrete latent structures in deeper models.
}

\section*{Acknowledgements}
The authors thank the Editor, Associate Editor, and three referees for helpful and constructive comments.

\section*{Funding Statement}
This work was partially supported by National Science Foundation grant DMS-2210796.
This work was also partially supported by grants R01ES027498 and R01ES028804 from National Institutes of Environmental Health Sciences of the National Institutes of Health, and received funding from European Research Council under the European Union's Horizon 2020 research and innovation program (grant agreement No 856506).

\section*{Data Availability Statement}
Matlab code implementing the proposed method and data are available at this GitHub repository: \texttt{https://github.com/yuqigu/BayesianPyramids}.

\section*{Conflict of Interest Statement}
The authors declare no potential conflicts of interest.

%\newpage
\singlespacing
\bibliographystyle{apalike}
\bibliography{ref.bib}
\spacingset{1.45}

\newpage

\input{supplement}

\end{document}

%% file: supplement.tex
\begin{center}
 {\LARGE{Supplement to ``Bayesian Pyramids: Identifiable Multilayer Discrete Latent Structure Models for Discrete Data''}}
\end{center}

\bigskip
\bigskip
\bigskip

% \beginsupplement
% \renewcommand{\thepage}{S\arabic{page}} 
\renewcommand{\thesection}{S.\arabic{section}}  
\renewcommand{\thetable}{S.\arabic{table}}  
\renewcommand{\thefigure}{S.\arabic{figure}}
\renewcommand{\theequation}{S.\arabic{equation}}

% set counters to zero
\setcounter{section}{0}
\setcounter{equation}{0}
\setcounter{figure}{0}

In this Supplementary Material, Section \ref{sec-pf} presents  technical proofs of all the theoretical results included in the main text, Section \ref{sec-poscomp} provides the Gibbs sampling details of the two-latent-layer Bayesian Pyramids and those of an alternative model, and Section \ref{sec-nume} presents additional numerical analyses, including overfitted mixture model simulations and real data analysis, and the analysis of the promoter dataset of nucleotide sequences.

%In this appendix, we provide technical proofs of the main results. The proofs of additional results (including those of some technical lemmas), posterior computation details, and additional data analysis are deferred to the Supplementary Material.

% \section*{Appendix: Technical Proofs}
\section{Technical Proofs of Theoretical Results}
\label{sec-pf}
\subsection{Proof of Proposition \ref{prop-relation}}
%\begin{proof}[Proof of Proposition \ref{prop-relation}]
We first prove part (a). 
Recall the model \eqref{eq-model} gives the following joint distribution of the observed $\bo y$ and all the latent variables, 
\begin{align*}
	\MP(\bo y,\aaa^{(1)},\ldots,\aaa^{(D-1)},z^{(D)}) 
	= \MP(\bo y\mid \aaa^{(1)})\cdot
	\prod_{m=1}^{D-2} \MP(\aaa^{(m)}\mid \aaa^{(m+1)})\cdot
	\MP(\aaa^{(D-1)}\mid z^{(D)})
	\cdot \mathbb P(z^{(D)}).
\end{align*}
In the above display, if marginalizing out all the latent variables except the $\aaa^{(1)}$ in the first latent layer, then we obtain
\begin{align*}
	\MP(\bo y,\aaa^{(1)}) 
	&= \sum_{\aaa^{(2)}}\cdots \sum_{\aaa^{(D-1)}}\sum_{z^{(D)}}
	\MP(\bo y,\aaa^{(1)},\ldots,\aaa^{(D-1)},z^{(D)})\\
	&= \MP(\bo y\mid \aaa^{(1)}) \sum_{\aaa^{(2)}}\cdots \sum_{\aaa^{(D-1)}}\sum_{z^{(D)}} \mathbb P(\aaa^{(1)}\mid\aaa^{(2)}) \cdots \mathbb P(\aaa^{(D-1)}\mid z^{(D)}) \mathbb P(z^{(D)})
	\\
	&= \MP(\bo y\mid \aaa^{(1)}) \cdot \mathbb P(\aaa^{(1)}),
\end{align*}
where the $\mathbb P(\aaa^{(1)})$ can be understood as the probability mass function for the $K_1$-dimensional binary vector $\aaa^{(1)}$.
We introduce a notation $\nnu^{(1)} = \left(\nu^{(1)}_{\bo a};~\bo a \in \{0,1\}^{K_1}\right)$ to denote the parameters of this categorical distribution for $\aaa^{(1)}$. Then $\nnu^{(1)}$ lives in the $(2^{K_1}-1)$-dimensional simplex.
Then based solely on $\aaa^{(1)} \in \{0,1\}^{K_1}$, the vector $\bo y$ has the following probability mass function for each $\bo c\in\times_{j=1}^p [d_j]$,
\begin{align}\label{eq-ya1}
	\MP(\bo y = \bo c\mid \nnu^{(1)}, \, \mathfrak L^{(1)},\,\GG^{(1)})
	= \sum_{\bo a \in \{0,1\}^{K_1} } \nnu^{(1)}_{\bo a} \prod_{j=1}^p \mathbb P(y_j = c_j \mid \aaa^{(1)} = \bo a,\,\mathfrak L^{(1)},\,\GG^{(1)}),
\end{align}
where the $\mathfrak L^{(1)} = \{\Lambda^{(j)};~j\in[p]\}$ denotes the collection of conditional probability tables of $y_j$ given $\aaa^{(1)}$.
First, note that the above display gives a latent class model for $\bo y$ with $2^{K_1}$ latent classes.
We next show that the above constrained latent class model satisfies the equality constraint \eqref{eq-equa} under the constraint matrix $\mb S^{(1)}$ constructed in \eqref{eq-relation}.
To prove this conclusion, consider two latent classes characterized by $\bo a\neq\bo a^\star\in\{0,1\}^{K_1}$ with $\bo a=(a_1,\ldots,a_{K_1})$ and $\aaa^\star=(a_1^\star,\ldots,a_{K_1}^\star)$. 
For some variable $j\in[p]$, if $S^{(1)}_{j,\bo a} = S^{(1)}_{j,\bo a^\star} = 0$, then by construction \eqref{eq-relation} we have $\mathbbm{1}(\bo a \succeq \GG^{(1)}_{j,\bcolon}) = \mathbbm{1}(\bo a^\star \succeq \GG^{(1)}_{j,\bcolon}) = 1$. This means there are $\bo a \succeq \GG^{(1)}_{j,\bcolon}$ and
	$\bo a^\star \succeq \GG^{(1)}_{j,\bcolon}$.
Then by the definition of the graphical matrix $\GG^{(1)}$, the binary vector $\GG^{(1)}_{j,\bcolon}$ specifies the set of parent variables of variable $y_j$. According to the factorization of the distribution of $\bo y$, the conditional distribution of $y_j$ given $\bo a^{(1)}$ only depends on the parent variables of variable $y_j$ in the first latent layer. In other words, the conditional probability distribution vector $\lambda^{(j)}_{1:d_j,\bo a}$ only depends on those $k\in[K_1]$ such that $g^{(1)}_{j,k}=1$.
Note that for all such $k\in[K_1]$ where $g^{(1)}_{j,k}=1$, we have $a_k=a^\star_k = 1$.
This implies $\bo a$ and $\bo a^\star$ must have the same conditional probability table for this variable $y_j$, which means $\lambda^{(j)}_{1:d_j,\bo a} = \lambda^{(j)}_{1:d_j,\bo a^\star}$. This is exactly the equality constraint \eqref{eq-equa} and thus proves the conclusion of part (a).

We next prove the conclusion of part (b) of the proposition. For an arbitrary $m=1,\ldots,D-2$, the Bayesian network structure implies that the joint distribution of $\aaa^{(m)}$ and $\aaa^{(m+1)}$ is
\begin{align*}
	&\MP(\aaa^{(m)}, \aaa^{(m+1)}) \\
	= &~ \sum_{\aaa^{(m+2)}}\cdots \sum_{\aaa^{(D-1)}}\sum_{z^{(D)}}
	\MP(\aaa^{(m)} \mid \aaa^{(m+1)}) \mathbb P(\aaa^{(m+1)}\mid\aaa^{(m+2)}) \cdots \mathbb P(\aaa^{(D-1)}\mid z^{(D)}) \mathbb P(z^{(D)}) \\
	= &~ \MP(\aaa^{(m)} \mid \aaa^{(m+1)}) \sum_{\aaa^{(m+2)}}\cdots \sum_{\aaa^{(D-1)}}\sum_{z^{(D)}} \mathbb P(\aaa^{(m+1)}\mid\aaa^{(m+2)}) \cdots \mathbb P(\aaa^{(D-1)}\mid z^{(D)}) \mathbb P(z^{(D)})
	\\
	= &~ \MP(\aaa^{(m)} \mid \aaa^{(m+1)})\cdot \mathbb P(\aaa^{(m+1)}),
\end{align*}
where $\mathbb P(\aaa^{(m+1)})$ is the probability mass function for the $K_{m+1}$-dimensional binary vector $\aaa^{(m+1)}$.
Similar to the proof of part (a), we use $\nnu^{(m+1)} = \left(\nu^{(m+1)}_{\aaa};~\aaa \in \{0,1\}^{K_{m+1}}\right)$ to denote the parameters of this categorical distribution for $\aaa^{(m+1)}$. Then $\nnu^{(m+1)}$ lives in the $(2^{K_{m+1}}-1)$-dimensional simplex.
The marginal distribution of $\aaa^{(m)}$ can be written as follows,
\begin{align}\label{eq-a1a2}
	&~\mathbb P(\aaa^{(m)} = \bo a^{(m)}\mid \nnu^{(m+1)}, \,\LLambda^{(m+1)},\,\GG^{(m+1)})\\ \notag
	=
	&~
	\sum_{\bo a \in \{0,1\}^{K_{m+1}} } \nnu^{(m+1)}_{\bo a} \prod_{k=1}^{K_m} \mathbb P(\alpha^{(m)}_{k} = a^{(m)}_k \mid \aaa^{(m+1)} = \bo a,\,\mathfrak L^{(m+1)},\,\GG^{(m+1)}),
\end{align}
where $\bo a^{(m)} = (a^{(m)}_1,\ldots,a^{(m)}_{K_m})$ is an arbitrary binary vector of length $K_m$.
The above distribution of $\aaa^{(m)}$ takes a similar form as that in \eqref{eq-ya1}. 
Therefore, we can use a similar argument as that in the proof of part (a) of the proposition to proceed with the proof. 
Specifically, \eqref{eq-a1a2} shows that after marginalizing out the distribution of all the latent variables deeper than $\aaa^{(m)}$ other than $\aaa^{(m+1)}$, the distribution of $\aaa^{(m)}$ becomes a constrained latent class model with $2^{K_{m+1}}$ latent classes.
Next it suffices to show such a constrained latent class model satisfies the constraint \eqref{eq-equa} with $\mb S^{(m+1)}$ defined in \eqref{eq-relation2}.
Consider two latent classes $\bo a\neq\bo a^\star\in\{0,1\}^{K_{m+1}}$ with $\bo a=(a_1,\ldots,a_{K_{m+1}})$ and $\aaa^\star=(a_1^\star,\ldots,a_{K_{m+1}}^\star)$. 
For some $k\in[K_m]$, if $S^{(m+1)}_{k,\bo a} = S^{(m+1)}_{k,\bo a^\star} = 0$, then by construction \eqref{eq-relation2} we have 
%$\bo a \succeq \GG^{(m+1)}_{k,\bcolon}$ and $\bo a^\star \succeq \GG^{(m+1)}_{k,\bcolon}$.
\begin{equation}\label{eq-gm1}
	\bo a \succeq \GG^{(m+1)}_{k,\bcolon},\quad
	\bo a^\star \succeq \GG^{(m+1)}_{k,\bcolon}
\end{equation}
Then by definition of $\GG^{(m+1)}$, the vector $\GG^{(m+1)}_{k,\bcolon}$ specifies the set of parent variables of variable $\alpha^{(m)}_k$ in the $(m+1)$-the latent layer. According to the factorization of distribution of $\aaa^{(m)}$, the conditional distribution of $\alpha^{(m)}_k$ given $\aaa^{(m+1)}$ only depends on its parent variables. Equivalently, the conditional probability of $\mathbb P(\aaa^{(m)}_k = 1 \mid \aaa^{(m+1)})$ only depends on those $\aaa^{(m+1)}_{k'}$ for which $g^{(m+1)}_{k,k'}=1$.
Note that for all such $k'\in[K_{m+1}]$ where $g^{(m+1)}_{k,k'}=1$, we have $a_k=a^\star_k = 1$ by \eqref{eq-gm1}.
This implies $\bo a$ and $\bo a^\star$ must have the same conditional probability of $\mathbb P(\aaa^{(m)}_k = 1 \mid \aaa^{(m+1)}=\bo a) = \mathbb P(\aaa^{(m)}_k = 1 \mid \aaa^{(m+1)}=\bo a^\star)$. Therefore there also is $\mathbb P(\aaa^{(m)}_k = 0 \mid \aaa^{(m+1)}=\bo a) = \mathbb P(\aaa^{(m)}_k = 0 \mid \aaa^{(m+1)}=\bo a^\star)$.
The above two equalities are exactly the equality constraint \eqref{eq-equa} and this proves the conclusion of part (b).
The proof of Proposition \ref{prop-relation} is now complete.
%\end{proof}

%%%

\subsection{Proof of Proposition \ref{prop-kr-gen}}
We first prove part (a) of the proposition.
Denote the 	Khatri-Rao product of the $p$ matrices $\{\bo\Lambda^{(j)},\,j\in[p]\}$ by $\KK^{0} = \odot_{j=1}^p \bo\Lambda^{(j)}$, then $\KK^{0}$ is of size $\prod_{j=1}^p d_j \times H$ with rows indexed by the variable pattern $\yy\in \times_{j=1}^p [d_j]$ and columns indexed by the latent category $h\in[H]$. 
We first introduce a useful lemma about the Khatri-Rao product. 

\begin{lemma}\label{lem-krprod}
	For matrices $\mathbf A_j$, $\mathbf B_j$ for $j\in\{1,\ldots,p\}$, there is 
	\begin{align}\label{eq-krprod}
		\odot_{j=1}^p \left\{
		\mathbf A_j  \mathbf B_j
		\right\}
		=\left\{\otimes_{j=1}^p \mathbf A_j \right\} 
		\bcdot 
		\left\{\odot_{j=1}^p \mathbf B_j\right\},
	\end{align}
	where $\mathbf A_j$, $\mathbf B_j$ have compatible dimensions such that the left hand side is well defined.
\end{lemma}

By Lemma \ref{lem-krprod}, %there exits an invertible $2^p\times 2^p$ matrix $\mathbf B$ such that
 $\KK^0$ can be written as 
\begin{align*}
	\KK^0
	&=\bigotimes_{j=1}^p  \left\{\begin{pmatrix}
		1 & 0 & \cdots & 0 & 0 \\
		0 & 1 & \cdots & 0 & 0 \\
		\vdots & \vdots & \ddots & \vdots & \vdots \\
		0 & 0 & \cdots & 1 & 0\\
		-1 & -1 & \cdots & -1 & 1
	\end{pmatrix}
	\boldsymbol{\cdot}
	\begin{pmatrix}
			\lambda^{(j)}_{1,1} & \lambda^{(j)}_{2,1} & \cdots & \lambda^{(j)}_{H,1} \\
			\lambda^{(j)}_{1,2} & \lambda^{(j)}_{2,2} & \cdots & \lambda^{(j)}_{H,2} \\
			\vdots & \vdots & \vdots & \vdots \\
			\lambda^{(j)}_{1,d_j-1} & \lambda^{(j)}_{2,d_j-1} & \cdots & \lambda^{(j)}_{H,d_j-1} \\
			1 & 1 & \cdots & 1 \\
		\end{pmatrix}\right\} \\
	&=\left\{\bigotimes_{j=1}^p  \begin{pmatrix}
		1 & 0 & \cdots & 0 & 0 \\
		0 & 1 & \cdots & 0 & 0 \\
		\vdots & \vdots & \ddots & \vdots & \vdots \\
		0 & 0 & \cdots & 1 & 0\\
		-1 & -1 & \cdots & -1 & 1
	\end{pmatrix}
	\right\}
	\boldsymbol{\cdot}
	\left\{
	\bigodot_{j=1}^p 
	\begin{pmatrix}
			\lambda^{(j)}_{1,1} & \lambda^{(j)}_{2,1} & \cdots & \lambda^{(j)}_{H,1} \\
			\lambda^{(j)}_{1,2} & \lambda^{(j)}_{2,2} & \cdots & \lambda^{(j)}_{H,2} \\
			\vdots & \vdots & \vdots & \vdots \\
			\lambda^{(j)}_{1,d_j-1} & \lambda^{(j)}_{2,d_j-1} & \cdots & \lambda^{(j)}_{H,d_j-1} \\
			1 & 1 & \cdots & 1 \\
		\end{pmatrix}\right\} \\
	&=: \bigotimes_{j=1}^p \mathbf B^j \bcdot\KK=:
	\mathbf B^0 \bcdot \KK,
\end{align*}
%from which it is not hard to see that $\mathbf B^0$ is a $\prod_{j=1}^p d_j \times \prod_{j=1}^p d_j$ lower-triangular matrix with all the diagonal entries equal to one. 
where $\mathbf B^0$ is the Kronecker product of $p$ lower-triangular matrices $\mathbf B^0_j$'s whose diagonal entries all equal to one. So by the property of the Kronecker product, $\rank(\mathbf B^0) = \prod_{j=1}^p \rank(\mathbf B^0_j)= \prod_{j=1}^p d_j$.
Therefore $\mathbf B^0$ is invertible, and $\KK^0$ has full column rank if and only if $\KK$ has full column rank.
We next constructively find $k$ rows of this $\prod_{j=1}^p d_j \times H$ matrix to form a $H\times H$ invertible submatrix, which will establish the full-rankness of matrix $\KK$ and hence the full-rankness of $\KK^0$.

We next introduce the \textit{lexicographic order} between binary vectors of the same length to facilitate the proof. For two binary vectors $\boldsymbol{x}=(x_1,\ldots,x_p),\boldsymbol{y}=(y_1,\ldots,y_p)$ both of length $p$, we say $\boldsymbol{x}$ is of greater lexicographic order than $\boldsymbol{y}$ and denote by $\boldsymbol{x}\succ_{\text{lex}}\boldsymbol{y}$ if either $x_1>y_1$, or $x_\ell > y_\ell$ for some $\ell\in\{2,\ldots,p\}$ and $x_m=y_m$ for all $m=1,\ldots,\ell-1$.
Under the assumption of the lemma, the $k$ columns of the constraint matrix $\mathbf S$ are distinct. An important consequence of this is that the $k$ columns of $\mathbf S$ can be arranged in a decreasing lexicographic order. Without loss of generality, we assume
\begin{equation}\label{eq-lexk}
S_{\bcolon,1}\succ_{\text{lex}} S_{\bcolon,2}\succ_{\text{lex}}\ldots \succ_{\text{lex}} S_{\bcolon,k}. 
\end{equation}

Next we construct the $H\times H$ invertible submatrix $\mathbf E$ of $\KK$ as follows. Since the rows of $\KK$ are indexed by $\yy\in[d_j]^p$, we focus on constructing $k$ different patterns $\yy^1,\ldots,\yy^k\in[d_j]^p$ to form $\mathbf E$.
Denote by $\ee_j$ a $J$-dimensional standard basis vector with the $j$th entry being one and the rest being zeros.
We define 
\begin{align}\notag
\yy^1 &= \sum_{j:\, S_{j,1}=1} y^1_j\ee_j +
        \sum_{j:\, S_{j,1}=0} d_j\ee_j;\\
\notag
\yy^2 &= \sum_{j:\, S_{j,2}=1} y^2_j\ee_j +
        \sum_{j:\, S_{j,2}=0} d_j\ee_j;\\ \notag 
 & \vdots \\ \label{eq-y1k}
\yy^k &= \sum_{j:\, S_{j,k}=1} y^k_j\ee_j +
        \sum_{j:\, S_{j,k}=0} d_j\ee_j,   
\end{align}
where $y^h_j\in\{1,\ldots,d_j-1\}$ for each $h\in[H]$ can be arbitrary; as long as $y^h_j \neq d_j$.
We claim that $\EE:= \KK_{(\yy^1,\yy^2,\ldots,\yy^k),\bcolon}$ is invertible.
We next prove this claim. 
First, there is 
\begin{align}\label{eq-elh}
	\EE_{\ell,h}=
	\KK_{\yy^\ell,h} = 
	\prod_{j:\,S_{j,\ell}=1}
	  \lambda^{(j)}_{y^{\ell}_j, h}
	\prod_{j:\, S_{j,\ell}=0}
	  1
	  =\prod_{j:\,S_{j,\ell}=1}
	  \lambda^{(j)}_{y^{\ell}_j, h},\quad
	\ell\in[H],~h\in[H].
\end{align}

The matrix $\KK$ can be viewed as a map taking the $k$ matrices 
$\LLambda=\{\LLambda^{(j)}:\,j\in[p]\}
%=\{\lambda^{(j)}_{h,c};\, j\in[p], c\in[d_j],  h\in[H]\}
$ as the input. 
Consider an arbitrary collection of $p$ vectors $\bo\Delta := \{\bo\delta^{(j)}:\,j\in[p]\}$, where $\bo\delta^{(j)}=\{\delta^{(j)}_{c}:\, c\in[d_j]~\text{and}~\delta^{(j)}_{d_j}=0\}$ is $d_j$-dimensional vector with the last element equal to zero for each $j$.
%define  $\widetilde{\bo\Delta}=\bo\Delta\bcdot\one^\top_k$ with entries $\{\widetilde\delta^{(j)}_{h,c};\, j\in[p], c\in[d_j],  h\in[H]\}$, and $\widetilde\delta^{(j)}_{h,c}=\delta^{(j)}_{c}$ for all $h\in[H]$.
Then $\Lambda^{(j)}$ and $\bo\delta^{(j)}\bcdot\one^\top_k$ have the same size $d_j\times H$.
We introduce the following useful lemma.
% introduce the mult-lemma
\begin{lemma}\label{lem-mult}
There exists a $\prod_{j=1}^p d_j \times \prod_{j=1}^p d_j$ invertible square matrix $\mathbf B:=\mathbf B(\{\bo\delta^{(j)}\}_{j\in[p]})$ depending only on $\{\bo\delta^{(j)}\}_{j\in[p]}$ such that 
\begin{align}\label{eq-algebra}
	\KK(\{\bo\Lambda^{(j)}-\bo\delta^{(j)}\bcdot\one^\top_k \}_{j\in[p]}) 
	&= \mathbf B(\{\bo\delta^{(j)}\}_{j\in[p]}) \bcdot \KK(\{\bo\Lambda^{(j)}\}_{j\in[p]}).
\end{align}
\end{lemma}

% continue with the proof of the lemma
Lemma \ref{lem-mult} implies that we can show $\KK(\{\bo\Lambda^{(j)} \}_{j\in[p]})$ has full column rank by showing that matrix $\KK(\{\bo\Lambda^{(j)}-\bo\delta^{(j)}\bcdot\one^\top_k \}_{j\in[p]})$ has full column rank for some $\bo\delta^{(j)}$'s.
 %$\EE$ is invertible if and only if the $\EE^\star$ with the following definition is invertible,
In particular, we construct $\bo\delta^{(j)}$ to be 
\begin{align*}
	\delta^{(j)}_c = \lambda^{(j)}_{c, 0}~\text{for}~ c\in[d_j-1];\quad \delta^{(j)}_{d_j}=0.
\end{align*}
Recall the $\yy^1,\ldots,\yy^k$ defined in \eqref{eq-y1k}, then the $(\yy^\ell,h)$th entry of $\KK(\{\bo\Lambda^{(j)}-\bo\delta^{(j)}\bcdot\one^\top_k \}_{j\in[p]})$ takes the form
\begin{align*}
	\EE^\star_{\ell,h}
	=\prod_{j:\,S_{j,\ell}=1}\left(\lambda^{(j)}_{y^\ell_j, h} - \delta^{(j)}_{y^\ell_j}\right)
	=\prod_{j:\,S_{j,\ell}=1}\left(\lambda^{(j)}_{y^\ell_j, h} - \lambda^{(j)}_{y^\ell_j, 0}\right),\quad
	\ell\in[H],~h\in[H].
\end{align*} 
Now consider if $\ell < h$, then because $S_{\bcolon,\ell}\succ_{\text{lex}} S_{\bcolon,h}$ under our assumption in \eqref{eq-lexk}, there must be some variable $j\in[p]$ such that $S_{j,\ell}=1$ and $S_{j,h}=0$, so $\lambda^{(j)}_{y^\ell_j,h}=\lambda^{(j)}_{y^\ell_j, 0}$ for this $j$. This means for $\ell<h$ the $\EE^{\star}_{\ell,h}$ must contain a factor of $\left(\lambda^{(j)}_{y^\ell_j, 0}-\lambda^{(j)}_{y^\ell_j, 0}\right)=0$, hence $\EE^{\star}_{\ell,h}=0$ for any $\ell<h$. So far we have obtained that $\EE^\star$ is a lower-triangular matrix. We further look at its diagonal entries. For any $h\in[H]$ there is
%\begin{align*}
	$$\EE^\star_{h,h} = 
	\prod_{j:\,S_{j,h}=1}
	\left(\lambda^{(j)}_{y^{h}_j, h} - \lambda^{(j)}_{y^{h}_j,h}\right) \neq 0
	$$
%\end{align*}
due to the inequality constraint \eqref{eq-neq} that $\lambda^{(j)}_{y_j,h} \neq \lambda^{(j)}_{y_j,0}$ if $S_{j,h}=1$. Now  we have shown $\EE^\star$ is a lower-triangular matrix with all the diagonal entries being nonzero, therefore $\EE^\star$ is invertible.
Since $\EE^\star$ is a submatrix of $\KK(\{\bo\Lambda^{(j)}-\bo\delta^{(j)}\bcdot\one^\top_k \}_{j\in[p]})$ containing $k$ of its rows, the latter must have full column rank. So by Lemma \ref{lem-mult}, $\KK(\{\bo\Lambda^{(j)}\}_{j\in[p]})$ also has full column rank. This proves part (a) of the proposition.

\bigskip
We next prove part (b) of the proposition. It suffices to show that in the special case described in Remark \ref{rmk-prop1} the Khatri-Rao product matrix $\mathbf K=\odot_{j=1}^p \LLambda^{(j)}$ does not have full column rank.
Specifically, suppose that besides constraint \eqref{eq-equa} that $\lambda^{(j)}_{1:d_j,\,h_1}=\lambda^{(j)}_{1:d_j,\,h_2}$ if $S_{j,h_1}=S_{j,h_2}=0$, the parameters also satisfy $\lambda^{(j)}_{1:d_j,\,h_1}=\lambda^{(j)}_{1:d_j,\,h_2}$ if $S_{j,h_1}=S_{j,h_2}=1$. 
In this case, we claim that if $\mb S_{\bcolon,h_1} = \mb S_{\bcolon,h_2}$ then there must also be $\KK_{\bcolon,h_1} = \KK_{\bcolon,h_2}$.
To see this, first note that now the constraint on the matrix $\mb S$ becomes $\lambda^{(j)}_{1:d_j,\,h_1}=\lambda^{(j)}_{1:d_j,\,h_2}$ as long as $S_{j,h_1} = S_{j,h_2}$.
So for each $\bo c=(c_1,\ldots,c_p)^\top \in \times_{j=1}^p [d_j]$, there is
\begin{align}\label{eq-kkeq}
	\KK_{\cc,h_1} 
	= \prod_{j=1}^p \lambda^{(j)}_{c_j,h_1}
	= \prod_{j=1}^p \lambda^{(j)}_{c_j,h_2}
	= \KK_{\cc,h_2}, 
\end{align}
where the second equality above follows from $S_{j,h_1} = S_{j,h_2}$ for each $j\in[p]$.
Now that \eqref{eq-kkeq} holds for all possible $\cc \in \times_{j=1}^p [d_j]$, we have $\KK_{\bcolon,h_1} = \KK_{\bcolon,h_2}$.
Therefore, the matrix $\KK$ is necessarily rank-deficient since it contains identical column vectors.
This completes the proof of Proposition \ref{prop-kr-gen}.
%\end{proof}

%%%%%% proof of strict identifiability %%%%%%

\subsection{Proof of Theorem \ref{thm-suff}}
%\begin{proof}[Proof of Theorem \ref{thm-suff}]
%%%
For a matrix $\mathbf M$, Kruskal's rank is the maximal number $r$ such that every $r$ columns of $\mathbf M$ are linear independent. Denote the Kruskal's rank of $\mathbf M$ by $\rank_{K}(\mathbf M)$. We next restate a useful version of the Kruskal's theorem on three-way tensor decomposition here to facilitate the proof; see more discussion on how this theorem can be invoked to show identifiability for a variety of latent variable models in \cite{allman2009}.

\begin{lemma}[\cite{kruskal1977}]\label{lem-kruskal}
	Suppose $M_1, M_2, M_3$ are three matrices of dimension $a_i\times H$ for $i=1,2,3$, $N_1, N_2, N_3$ are three matrices each with $k$ columns, and $\left(\odot_{i=1}^3 M_i \right) \bcdot \one_{H\times 1} = \left(\odot_{i=1}^3 N_i \right) \bcdot \one_{H\times 1}$. If $\rank_{K}( M_1) + \rank_{K}( M_2) + \rank_{K}( M_3) \geq 2k + 2$, then there exists a permutation matrix $P$ and three invertible diagonal matrices $D_i$ with $D_1D_2D_3 =\mb I_k$ and $N_i = M_i D_i P$.
\end{lemma}

%For a subset $\mathcal A\subseteq[p]$, denote by $\mathbf S_{\mathcal A,\bcolon}$ the $|\mathcal A|\times H$ submatrix of $\mathbf S$ which contains the rows with indices in $\mathcal A$.
%Based on Proposition \ref{prop-kr-gen} and Lemma \ref{lem-kruskal}, we establish the following  identifiability result for the model in \eqref{eq-equa}. Note that {neither} the baseline probabilities $\lambda^{j}_{0,c_j}$'s  {nor} the constraint matrix $\mathbf S$ are assumed known \textit{a priori}.
%%%
	Under assumption (a) in the theorem, we apply Lemma \ref{lem-kruskal}
	%Theorem \ref{thm-suff} 
	to establish that $\KK^i:=\KK(\{\bo\Lambda^{(j)}\}_{j\in\mathcal A_i})$ has full column rank for $i=1$ and $2$.
	Now consider $\KK^3:=\KK(\{\bo\Lambda^{(j)}\}_{j\in\mathcal A_3})$ for the set of variables $\mathcal A_3$. 
    We first claim that for each column $h$ of 
    $\KK^3$, the sum of all the entries of $\KK^3$ in this column equals one. To see this, denote the indices of variables in $\mathcal A_3$ by $\{j_1,\ldots,j_m\}$ with $m=|\mathcal A_3|$, then
	\begin{align*}
		\sum_{\bo c\in \times_{j\in\mathcal A_3} [d_j]} \KK^3_{\bo c, h}=
		 \sum_{(c_{j_1},\ldots,c_{j_m})\in\atop [d_{j_1}]\times\cdots\times[d_{j_m}]} \left(\prod_{\ell=1}^m \lambda^{(j_\ell)}_{h,c_{j_\ell}}\right)
		 =\prod_{\ell=1}^m \left(\sum_{c_{j_\ell}\in[d_{j_\ell}]}  \lambda^{(j_\ell)}_{h,c_{j_\ell}}\right)
		 =\prod_{\ell=1}^m 1 = 1.
	\end{align*}
	From the above it is not hard to see that  $\KK^3_{\bo c, h} = \mathbb P(y_{j_1}=c_{j_1},\ldots,y_{j_m}=c_{j_m}\mid z=h)$ where $z$ is the latent class indicator variable. So $\KK^3_{\bcolon,h}$ is a vector living in the $(\prod_{j\in\mathcal A_3} d_j-1)$-dimensional simplex for each $h$, characterizing the conditional joint distribution of $\{y_j\}_{j\in\mathcal A_3}$ given $z=h$.
	
	Under assumption (b) in the theorem, we claim that the matrix $\KK^3$ must have Kruskal rank at least two.
	We show this  using proof by contradiction. Suppose there exist two different columns of $\KK^3$ indexed by $h_1, h_2$ that are linear dependent, then there are some $a,b\in\mathbb R$ such that
	\begin{align}\label{eq-k3linear}
		a\cdot\KK^3_{\bcolon,h_1} + b\cdot\KK^3_{\bcolon,h_2}=0,
	\end{align}
	with $a\neq 0$ or $b\neq 0$.
	Since in the last paragraph we established $\sum_{\bo c}\KK^3_{\bo c,h}=1$ for each $h$, we can take the sum over all the $\cc \in \times_{j\in\mathcal A_3} [d_j]$ and obtain $a\cdot 1+ b\cdot 1=0$. If $a\neq 0$, then $\KK^3_{\bcolon,h_1}=-b/a\KK^3_{\bcolon,h_2}=\KK^3_{\bcolon,h_2}$. Now we fix an arbitrary $\ell\in\mathcal A_3$ and $c_\ell\in[d_\ell]$ in \eqref{eq-k3linear} and sum over all the remaining indices,
	\begin{align*}
	\sum_{m\in[p]\atop m\neq j}\sum_{c_m\in[d_m]} \KK^3_{\bo c,h_1}
	    =
		\sum_{m\in[p]\atop m\neq j}\sum_{c_m\in[d_m]} \prod_{j=1}^p \lambda^{(j)}_{h_1,c_j}
		=\lambda^{(\ell)}_{h_1,c_\ell}\prod_{m\in[p]\atop m\neq \ell}
		 \left(\sum_{c_m\in[d_m]}\lambda^{(m)}_{h_1,c_m}\right)
		=\lambda^{(\ell)}_{h_1,c_\ell},
	\end{align*}
	so $\KK^3_{\bcolon,h_1}=\KK^3_{\bcolon,h_2}$ indeed indicates that $\lambda^{(\ell)}_{h_1,c_\ell}=\lambda^{(\ell)}_{h_2,c_\ell}$ for any $\ell\in\mathcal A_3$ and $c_\ell\in[d_\ell]$. This contradicts the assumption (b) in the theorem that for any $h_1\neq h_2$ there exists $j\in\mathcal A_3$ and $c\in[d_j]$ for which $\lambda^{(j)}_{h_1,c}\neq \lambda^{(j)}_{h_2,c}$.
	This contradiction implies that $\KK^3$ must have Kruskal rank at least two.
	
	Finally, since the distribution  for the observed $\bo y=(y_1,\ldots,y_j)^\top$ can be written as the tensor $\mathbf\Pi=(\pi_{c_1\cdots c_p})$ and further
	\begin{align}\label{eq-three}
		\mathbf\Pi = \left\{\odot_{j=1}^p \bo\Lambda^{j} \right\}\bcdot \bo\nu
		= \left\{\odot_{i=1}^3 \KK^{i} \right\}\bcdot \bo\nu
		= \KK^3 \odot \KK^2 \odot \left\{ \KK^1 \bcdot \text{diag}(\bo\nu) \right\},
	\end{align}
	where $\bo\nu=(\nu_1,\ldots,\nu_k)^\top$ belongs to the $(k-1)$-dimensional simplex and $\text{diag}(\bo\nu)$ denotes a $H\times H$ diagonal matrix with diagonal entries being $\nu_h$'s.
	Recall that we have established $\KK^1$, $\KK^2$ both have full column rank and $\KK^3$ has Kruskal rank at least two.
	Under the assumption that $\nu_h>0$ for each $h\in[H]$ in the theorem, it is not hard to see that matrix $\KK^1 \bcdot \text{diag}(\bo\nu)$ also has full column rank.
	Now suppose an alternative set of parameters $(\overline{\bo\nu},\overline{\bo\Lambda})$ and the true parameters $({\bo\nu},{\bo\Lambda})$ lead to the same distribution, that is, lead to the same $\mathbf\Pi$. We similarly denote $\overline\KK^i:=\KK(\{\overline{\bo\Lambda}^{(j)}\}_{j\in\mathcal A_i})$ for $i=1,2,3$.
	Then
	\begin{align*}
	\KK^3 \odot \KK^2 \odot \left\{ \KK^1 \bcdot \text{diag}(\bo\nu) \right\}
	=\overline\KK^3 \odot \overline\KK^2 \odot \left\{ \overline\KK^1 \bcdot \text{diag}(\overline{\bo\nu}) \right\}.
	\end{align*}
	Now we apply Lemma \ref{lem-kruskal} to this three-way decomposition and obtain 
	\begin{align}\label{eq-decomp}
		\overline\KK^i = \KK^i D_i P~~\text{for}~~ i=2,\,3;
		\qquad
		\overline\KK^1\bcdot\text{diag}(\overline{\bo\nu}) = \KK^1\bcdot \text{diag}({\bo\nu}) D_1 P
	\end{align}
	for a permutation matrix $P$ and three invertible diagonal matrices $D_i$. Now note that both $\overline\KK^i$ and $\KK^i$ have each column characterizing the conditional joint distribution of $\{y_j\}_{j\in\mathcal A_i}$ given the $h$th latent class $z=h$. Therefore the sum of each column of $\overline\KK^i$ or $\KK^i$ equals one, which implies the diagonal matrix $D_i$ is an identity matrix for $i=2$ or 3. Since Lemma \ref{lem-kruskal} ensures $D_1D_2D_3=I_k$, we also obtain $D_1=I_k$.  
	So far we have obtained $\overline\KK^i = \KK^i P$ for $i=2,3$ and $\overline\KK^1\bcdot\text{diag}(\overline{\bo\nu}) = \KK^1\bcdot \text{diag}({\bo\nu}) P$. 
	Now proceeding in the same way as the argument after Eq.~\eqref{eq-k3linear}, the $\overline\KK^i = \KK^i P$ for $i=2,3$ implies $\overline{\bo\Lambda}^{j}={\bo\Lambda}^{j}\cdot P$ for all $j\in\mathcal A_2\cup\mathcal A_3$.
    For $i=3$, consider $h_1\in[H]$ and assume without loss of generality that the $(h_2,h_1)$th entry of matrix $P$ is $P_{h_2,h_1}=1$. Then the $h_1$th column of the equality $\overline\KK^1\bcdot\text{diag}(\overline{\bo\nu}) = \KK^1\bcdot \text{diag}({\bo\nu}) P$ gives $\overline\KK^1_{\cc,h_1}\cdot\bar\nu_{h_1}=\overline\KK^1_{\cc,h_1}\cdot\nu_{h_2}$; summing over the index $\cc\in\times_{j\in\mathcal A_1}[d_j]$ gives $\bar\nu_{h_1}=\nu_{h_2}$. Note we have generally established $\bar\nu_{h_1}=\nu_{h_2}$ whenever $P_{h_2,h_1}=1$, which means $\overline\nnu^\top=\nnu^\top\cdot P$. Combining this with \eqref{eq-decomp}, we obtain $\overline\KK^1 = \KK^1 \cdot P$ and $\overline{\bo\Lambda}^{(j)}={\bo\Lambda}^{(j)}$ for $j\in\mathcal A_1$ as a last step. 
    This establishes the strict identifiability of all the parameters $\bo\nu$ and $\mathbf \Lambda$ up to a relabeling of the latent classes. 
    %The proof of Theorem \ref{thm-suff complete.
%\end{proof}

%%%%%%%%% proof of corollary %%%%%%%%%%%

\subsection{Proof of Corollary \ref{cor-three}}
%\begin{proof}[Proof of Corollary \ref{cor-three}]
Following a similar argument as the proof of Theorem \ref{thm-suff}, we can establish that each $\KK^i=\KK(\{\bo\Lambda^{(j)}\}_{j\in\mathcal A_i})$ has full column rank for $i=1,2,3$. Since if a matrix with at least two columns has full column rank, it must also have Kruskal rank at least two, the assumptions in Theorem \ref{thm-suff} are satisfied and the conclusion of the corollary directly follows.
%\end{proof}

\subsection{Proof of Theorem \ref{thm-genid}}
%\begin{proof}[Proof of Theorem \ref{thm-genid}]
We next introduce a useful concept in algebraic geometry, the \textit{algebraic variety}, to facilitate the proof.
An algebraic variety $\mathcal V$ is defined to be the simultaneous zero-set of a finite set of multivariate polynomials $\{f_i\}_{i=1}^n\subseteq \mathbb R[x_1,x_2,\ldots,x_d]$, 
$\mathcal V =\mathcal V(f_1,\ldots,f_n) = \{\boldsymbol x\in\mathbb R^d\mid f_i(\boldsymbol x) = 0,~ 1\leq i\leq n\}$.
An algebraic variety ${\mathcal V}$ equals the entire space $\mathbb R^d$ if all of the polynomials $f_i$ defining it are zero polynomials.
Otherwise, ${\mathcal V}$ is called a \textit{proper subvariety} of dimension less than $d$, in which case it must have Lebesgue measure zero in $\mathbb R^d$. 
The same argument still holds if the above $\mathbb R^d$ is replaced by a parameter space $\mc T\subseteq\mathbb R^d$ that contains an open ball of full dimension in $\mathbb R^d$.
For the constrained latent class model under constraints \eqref{eq-equa} and \eqref{eq-neq}, we have the following parameter space for $\LLambda$ subject to the constraint of $\mathbf S$ (denoted by $\LLambda^{\mathbf S}$) and mixture proportions $\nnu$,
\begin{align*}
	\Omega^{\mb S} = 
    \Big\{(\LLambda^{\mathbf S},\bo\nu):\,
    & \forall~j\in[p],~\forall~c_j\in[d_j],~ \max_{h:\,S_{j,h}=0} \lambda^{(j)}_{c_j, h} = \min_{h:\,S_{j,h}=0} \lambda^{(j)}_{c_j, h}
    \\
    & \forall~j\in[p],~\forall~c_j\in[d_j],~ 
   \lambda^{(j)}_{c_j, h_1} \neq \lambda^{(j)}_{c_j, h_2}
     ~\text{if}~S_{j,h_1}\neq S_{j,h_2};
    \\
    & \sum_{h=1}^k \nu_h = 1, ~ \nu_h>0~~\forall~ h\in[H]
      \Big\}.
\end{align*}
On $\Omega^{\mb S}$, altering some entries in $\mb S$ from one to zero is equivalent to imposing more equality constraints on $\LLambda$ and forces the resulting parameters to be in $\Omega^{\widetilde{\mb S}}$ instead of $\Omega^{\mb S}$. 
Under the condition that the altered submatrix $\widetilde{\mb S}_{\mc A_i,\bcolon}$ has distinct column vectors, 
 Proposition \ref{prop-kr-gen} gives that for any $\LLambda \in \Omega^{\widetilde{\mb S}}$ the matrix $ \odot_{j\in\mc A_i} \LLambda^{(j)}$ has full column rank $k$ for $i=1,2$.
Note that the statement that the $\prod_{j\in\mc A_i}d_j \times H$ matrix $\odot_{j\in\mc A_i} \LLambda^{(j)}$ has full column rank $k$ is equivalent to the following statement,
\begin{itemize}
	\item[(M)] the maps sending the matrix $\odot_{j\in\mc A_i} \LLambda^{(j)}$ to all of its $C_i := (\prod_{j\in\mc A_i}d_j)! /(k! (\prod_{j\in\mc A_i}d_j - k)!)$ possible $H\times H$ minors $M_1^i,M_2^i,\ldots, M^i_{C_i}$ yields at least one nonzero minor.
\end{itemize}  
Here each matrix determinant $M^i_\ell$ is a polynomial of the  parameters $\{\LLambda^{(j)}\}_{j\in\mc A_i}$, so we can simply write each of them as $M^i_\ell (\{\LLambda^{(j)}\}_{j\in\mc A_i})$.
We now define 
$$
\mathcal N = \bigcup_{i=1,2} \Big\{\bigcap_{\ell=1}^{C_i} \{(\LLambda,\nnu)\in \Omega^{\mb S}: M^i_\ell (\{\LLambda^{(j)}\}_{j\in\mc A_i}) = 0\}\Big\},
$$
then $\mathcal N \subseteq \Omega^{\mb S}$ and $\mathcal N$ is an algebraic variety defined by polynomials of the model parameters. 
An important observation is that this $\mathcal V$ is a \textit{proper} subvariety of $\Omega^{\mb S}$.
We next prove this claim.
First, note that for any $\LLambda \in \Omega^{\widetilde{\mb S}}$, there exists some $\LLambda^{\star} \in \Omega^{\mb S}$ that is arbitrarily close to $\LLambda$.
This is because the altered parameter space $\Omega^{\widetilde{\mb S}}$ is obtained by changing some inequality constraints in the original $\Omega^{\mb S}$ to equality constraints.
Therefore, fixing a $\LLambda \in \Omega^{\widetilde{\mb S}}$ and considering a corresponding $\LLambda^{\star} \in \Omega^{\mb S}$ close enough to $\LLambda$,
since the statement (M) ensures $M^i_{\ell_i} (\{\LLambda^{(j)}\}_{j\in\mc A_i}) \neq 0$ for some $\ell_i\in[C_i]$, $i=1,2$, 
we also have  $M^i_{\ell_i} (\{\LLambda^{\star, (j)}\}_{j\in\mc A_i}) \neq 0$  for  $\ell_i\in[C_i]$, $i=1,2$.
The above argument holds because the polynomials $M^i_{\ell_i}(\bcdot)$'s are continuous functions of the $\LLambda$-parameters and  $\LLambda^{\star} \in \Omega^{\mb S}$ and $\LLambda$ are close enough. 
%This is because statement (M) ensures that there exists at least one set of model parameters under the original constraint matrix $\mb S$ that give nonzero values for the polynomials defining $\mathcal V$. 
This indicates that the polynomials defining $\mathcal V$ contain at least one nonzero polynomial on $\Omega^{\mb S}$. 
This proves the earlier claim that $\mathcal V$ is a {proper} subvariety of $\Omega^{\mb S}$, so $\mathcal V$ has measure zero with respect to the Lebesgue measure on $\Omega^{\mb S}$.
Then we can simply consider parameters falling in $\Omega^{\mb S}\setminus \mc N$ and proceed in the same way as in the proof of Theorem \ref{thm-suff} to establish identifiability. The argument above establishes generic identifiability of the parameters by Definition \ref{def-genid}. This completes the proof of Theorem \ref{thm-genid}.
%\end{proof}

\subsection{Proof of Theorem \ref{thm-pos}}
%\begin{proof}[Proof of Theorem \ref{thm-pos}]
First, based on Theorem 2 in \cite{dunson2009} we obtain that when the priors for $\LLambda^{\mb S}$ (denoting the $\LLambda$ subject to the constraints imposed by the constraint matrix $\mb S$) and $\nnu$ both have full support, then the induced prior on the probability tensor $\bo\Pi$ for $\bo y_i$ also has full support on the $(\prod_{j=1}^p d_j - 1)$-dimensional simplex.
	We need to establish that for any $\epsilon^\star>0$ there exists some positive integer $N$ such that for all $n>N$ there is
	\begin{align*}
		\mathbb P(\Theta\in \mathcal N^c_{\epsilon}(\Theta^0) \mid \bo y_1,\ldots,\bo y_n) < \epsilon^\star~\text{a.s.}~\mathbb P^0.
	\end{align*}
	Define $$\delta= \inf_{\bo\Theta\in\mathcal N^c_{\epsilon}(\Theta^0) } \norm{\bo\Pi_{\Theta^0} - \bo\Pi_{\Theta}}_1,$$ 
	where $\bo\Pi_{\Theta^0}$ denotes the probability tensor characterizing the distribution of $\bo y_i$ under parameters $\Theta^0$. According to identifiability, there is $\mathbb P_{\Theta^0} \neq \mathbb  P_{\Theta}$ for $\Theta\in\bo{\Theta}\setminus\mathcal N_{\epsilon}(\Theta^0)$. Since $\mathcal N_{\epsilon}$ is an open set and the entire parameter space $\bo{\Theta}$ is a compact set, the complement $\mathcal N^c_{\epsilon}(\Theta^0)$ is also compact. Therefore  the previously defined $\delta$ is greater than zero. and there exists some $N$ such that for $n>N$,
	\begin{align*}
		\mathbb P(\Theta\in \mathcal N^c_{\epsilon}(\Theta^0) \mid \bo y_1,\ldots,\bo y_n) 
		\leq \mathbb P(\norm{\bo\Pi_{\Theta^0} - \bo\Pi_{\Theta}}_1 >\delta/2 \mid \bo y_1,\ldots,\bo y_n) 
		< \epsilon^\star,
	\end{align*}
	where the last inequality above results from the assumption the prior has full support on the space of the probability tensor $\bo\Pi$. This proves the conclusion of the theorem.
%\end{proof}

\subsection{Proof of Theorem \ref{thm-stack}}
We build on the conclusions of Proposition \ref{prop-relation} and Corollary \ref{cor-three} and prove the identifiability of parameters layer by layer in the ``bottom-up'' direction.
Recall that the values of $K_1,\ldots,K_{D-1}$ are known.
We rewrite the form of $\GG^{(m)}$ as in \eqref{eq-ggm} in the theorem,
\begin{align*}
		\GG^{(m)} = 
		\begin{pmatrix}
			\mb I_{K_m}\\
			\mb I_{K_m}\\
			\mb I_{K_m}\\
			\GG^{(m),\star}
		\end{pmatrix}.
\end{align*}
First consider the $K_1\times K_1$ submatrix of $\GG^{(1)}$ consisting of its first $K_1$ rows (which is $\mb I_{K_1}$) and denote it by $\GG^{(1)}_{1:K_1,\bcolon}$.
The $\GG^{(1)}_{1:K_1,\bcolon}$ will have a corresponding $K_1\times 2^{K_1}$ constraint matrix $\mb S^{(1)}_{1:K_1,\bcolon}$, which consists of the first $K_1$ rows of $\mb S^{(1)}$.
We next prove that this $\mb S^{(1)}_{1:K_1,\bcolon}$ has $2^{K_1}$ distinct column vectors.
Note that each column of $\mb S^{(1)}_{1:K_1,\bcolon}$ is indexed by a latent class characterized by a $K_1$-dimensional binary pattern $\aaa^{(1)}\in\{0,1\}^{K_1}$. 
So we only need to show that if $\aaa = (\alpha_1,\ldots,\alpha_{K_1}) \neq \aaa^\star = (\alpha_1^\star,\ldots,\alpha_{K_1}^\star)$, then $\mb S^{(1)}_{1:K_1,\aaa} \neq \mb S^{(1)}_{1:K_1,\aaa^\star}$.
Without loss of generality, suppose $\alpha_k\neq\alpha^\star_k$ for certain $k\in[K_1]$, with $\alpha_k=1$ and $\alpha^\star_k=0$.
Since the $k$th row of $\GG^{(1)}_{1:K_1,\bcolon}$ is the standard basis vector $\ee_k$ which equals one in the $k$th coordinate and zero in other coordinates, according to Proposition \ref{prop-relation} we have
\begin{align*}
	S^{(1)}_{k,\aaa} &= 1-\mathbbm{1}(\aaa\succeq \ee_k) = 1-1 = 0;\\
	S^{(1)}_{k,\aaa^\star} &= 1-\mathbbm{1}(\aaa^\star\succeq \ee_k) = 1-0 = 1.
\end{align*}
So $S^{(1)}_{k,\aaa} \neq S^{(1)}_{k,\aaa^\star}$, and hence $\mb S^{(1)}_{1:K_1,\aaa} \neq \mb S^{(1)}_{1:K_1,\aaa^\star}$.
Since $\aaa$ and $\aaa^\star$ are arbitrary, we have shown that $\mb S^{(1)}_{1:K_1,\bcolon}$ has $2^{K_1}$ distinct column vectors.
A similar argument applies to $\mb S^{(1)}_{(K_1+1):(2K_1),\bcolon}$ and $\mb S^{(1)}_{(2K_1+1):(3K_1),\bcolon}$, since the second $K_1$ rows and the third $K_1$ rows of $\GG^{(1)}$ also both form an identity matrix $I_m$ by \eqref{eq-ggm}.
This means the constraint matrix $\mb S^{(1)}$ vertically stacks three submatrices $\mb S^{(1)}_{1:K_1,\bcolon}$, $\mb S^{(1)}_{(K_1+1):(2K_1),\bcolon}$, $\mb S^{(1)}_{(2K_1+1):(3K_1),\bcolon}$, each having distinct column vectors. By Corollary \ref{cor-three}, this guarantees the identifiability of the following quantities: the $\mb S^{(1)}$, the parameters of each conditional distribution of $y_j \mid \aaa^{(1)}_{\pa(j)}$ (denoted by $\mathfrak L^{(1)}$), and the parameters of the categorical distribution of $\aaa^{(1)}$ (denoted by $\nnu^{(1)}$).
Here $\nnu^{(1)}=\{\nu_{\aaa};\,\aaa\in\{0,1\}^{K_1}\}$ is a $2^{K_1}$-dimensional vector with non-negative entries summing to one.
We claim that the identifiability of $\mb S^{(1)}$ also implies the identifiability of $\GG^{(1)}$ here. As shown earlier in the previous paragraph, the existence of an identity submatrix $\mb I_{K_1}$ in $\GG^{(1)}$ actually guarantees that the resulting $\mb S^{(1)}$ contains $2^{K_1}$ distinct column vectors. 
This means each column of $\GG^{(1)}$ can indeed be read off from each column of $\mb S^{(1)}$.
This proves the claim of the identifiability of $\GG^{(1)}$.

Now recall that part (b) of Proposition \ref{prop-relation} states the following: the distribution of $\aaa^{(1)}$ can be considered as a constrained latent class model with $2^{K_2}$ latent classes, which are characterized by configurations of latent variable $\aaa^{(2)}$.
Since the distribution of $\aaa^{(1)}$ is already identified with parameters $\nnu^{(1)}$, we can proceed as if $\aaa^{(1)}$ are the observed variables with a known probability tensor.
This means we can further examine the structure of the current constraint matrix $\mb S^{(2)}$ to establish identifiability of parameters for this second constrained latent class model.
By \eqref{eq-ggm}, $\GG^{(2)}$ also contains three copies of identity submatrices $\mb I_{K_2}$, so the same argument as in the last paragraph gives that $\mb S^{(2)}$ vertically stacks three submatrices $\mb S^{(2)}_{1:K_2,\bcolon}$, $\mb S^{(2)}_{(K_2+1):(2K_2),\bcolon}$, $\mb S^{(2)}_{(2K_2+1):(3K_2),\bcolon}$, each having distinct column vectors.
Invoking Corollary \ref{cor-three} again gives the identifiability of the $\GG^{(2)}$, the parameters of each conditional distribution of $\alpha^{(1)}_k \mid \aaa^{(2)}_{\pa(\alpha^{(1)}_k)}$ (denoted by $\mathfrak L^{(2)}$), and the parameters of the categorical distribution of $\aaa^{(2)}$ (denoted by $\nnu^{(2)}$).
Then in a similar and recursive fashion, we can establish identifiability for $\GG^{(m)}$, $\mathfrak L^{(m)}$, and $\nnu^{(m)}$ for each $m=2,\ldots,D-1$. This completes the proof of Theorem \ref{thm-stack}.

\subsection{Proof of Proposition \ref{prop-2layer}}\label{sec-pf2layer}
Rewrite the distribution of each $y_j$ given $\aaa^{(1)}$ under \eqref{eq-2layer} as
\begin{align*}
\lambda^{(j)}_{c,\,\aaa}
&=
\mathbb P(y_j=c\mid \aaa^{(1)}=\aaa) \\
&=
\frac{ \exp\left(\beta_{j,c,0} + \sum_{k=1}^{K_1} \beta_{j,c,k} g^{(1)}_{j,k}\alpha_k \right)}{ \sum_{m=1}^{d_j} \exp\left(\beta_{j,0,m} + \sum_{k=1}^{K_1} \beta_{j,k,m} g^{(1)}_{j,k}\alpha_k \right)},
% \\
% &= :
% f_{j,\aaa} \left(\beta_{j,c,0} + \sum_{k=1}^K \beta_{j,c,k} g^{(1)}_{j,k} \alpha_k \right),
\quad
c\in[d_j],~ \aaa\in\{0,1\}^{K_1},
\end{align*}
where  $\beta_{j,0,d_j} = \beta_{j,k,d_j} = 0$ for all $k\in[K_1]$.
For each variable $y_j$, consider a $d_j\times 2^{K_1}$ conditional probability matrix $\Lambda^{j} = (\lambda^{(j)}_{c,\,\aaa})$ whose rows are indexed by the $d_j$ categories and columns by the $2^{K_1}$ possible binary latent pattern configurations.
Under the condition $\GG^{(1)}_{1:K_1,\bcolon} = \GG^{(1)}_{(K_1+1):2K_1,\bcolon} = \GG^{(1)}_{(2K_1+1):3K_1,\bcolon} =\mb I_{K_1}$, consider the two Khatri-Rao product matrices 
$$\mb K^{1} = \odot_{j=1}^{K_1} \Lambda^{(j)},\quad
\mb K^{2} = \odot_{j=K_1+1}^{2K_1} \Lambda^{(j)},\quad
\mb K^{3} = \odot_{j=2K_1+1}^{3K_1} \Lambda^{(j)}.$$
We claim that each of $\mb K^1$, $\mb K^2$, and  $\mb K^3$ has full column rank $2^{K_1}$.
We next focus on proving this claim for $\mb K^1$ without loss of generality.
Note $\mb K^1$ is of size $\prod_{j=1}^{K_1} d_j \times 2^{K_1}$, where each of its row is indexed by a $K_1$-dimensional response pattern $\bo y_{1:K_1}\in\times_{j=1}^{K_1} [d_j]$.
There is $d_j\geq 2$ for each $j$, so to show $\mb K^1$ has full column rank, it suffices to find a $2^{K_1}\times 2^{K_1}$ submatrix of it that is invertible.
To this end, we next specifically consider the rows of $\mb K^1$ indexed by response pattern $\bo y = (y_1,\ldots,y_{K_1})$ where each $y_j\in\{1,d_j\}$. These $2^{K_1}$ response patterns naturally give rise to a $2^{K_1}\times 2^{K_1}$ submatrix of $\mb K^1$, which we denote by $\mb K^{1,\sub}$.
Now note that for each $j\in[K_1]$ there is $\bo g^{(1)}_{j,\bcolon} = \ee_j$, the $j$th standard basis vector, therefore when $\aaa$ varies all the $\lambda^{(j)}_{1,\,\aaa}$ only takes two possible values: %$f(\beta_{j,0,1} + \beta_{j,j,1})$ and $f(\beta_{j,0,1})$
\begin{align}\label{eq-lam1}
\lambda^{(j)}_{1,\,\aaa} =
\begin{cases}
    \dfrac{\exp\left(\beta_{j,0,1}\right)}{\sum_{m=1}^{d_j} \exp\left(\beta_{j,0,m} \right) }, & \text{if}~~\aaa\nsucceq\ee_j;
    \\[5mm]
    \dfrac{\exp\left(\beta_{j,0,1} + \beta_{j,j,1}\right)}{\sum_{m=1}^{d_j} \exp\left(\beta_{j,0,m} + \beta_{j,j,m} \right) }, & \text{if}~~\aaa\succeq\ee_j.
\end{cases}
\end{align}
Similarly there are
\begin{align}\label{eq-lamdj}
\lambda^{(j)}_{d_j,\,\aaa} =
\begin{cases}
    \dfrac{1}{\sum_{m=1}^{d_j} \exp\left(\beta_{j,0,m} \right) }, & \text{if}~~\aaa\nsucceq\ee_j;
    \\[5mm]
    \dfrac{1}{\sum_{m=1}^{d_j} \exp\left(\beta_{j,0,m} + \beta_{j,j,m} \right) }, & \text{if}~~\aaa\succeq\ee_j.
\end{cases}
\end{align}
Now an important observation is that the conditional probability tensor $\mb K^{1,\sub}$ can be written as a Kronecker product of $K_1$ matrices each of size $2\times 2$ as follows,
\begin{align}\label{eq-k1sub}
    \mb K^{1,\sub} = 
    \bigotimes_{j=1}^{K_1} 
    \begin{pmatrix}
        \dfrac{\exp\left(\beta_{j,0,1}\right)}{\sum_{m=1}^{d_j} \exp\left(\beta_{j,0,m} \right) }
        &
        \dfrac{\exp\left(\beta_{j,0,1} + \beta_{j,j,1}\right)}{\sum_{m=1}^{d_j} \exp\left(\beta_{j,0,m} + \beta_{j,j,m} \right)}
        \\[5mm]
        \dfrac{1}{\sum_{m=1}^{d_j} \exp\left(\beta_{j,0,m} \right) }
        &
        \dfrac{1}{\sum_{m=1}^{d_j} \exp\left(\beta_{j,0,m} + \beta_{j,j,m} \right) }
    \end{pmatrix}
    =:  \bigotimes_{j=1}^{K_1} \mb B^{(j)}.
\end{align}
Under the assumptions $\beta_{j,k,1} \neq 0$ for any  $g^{(1)}_{j,k}=1$ stated in the theorem, we have $\beta_{j,j,1} \neq 0$ for all $j\in[K_1]$.
Therefore, the $2\times 2$ matrix $\mathbf B^{(j)}$ in \eqref{eq-k1sub} has full rank. 
According to the property of the Kronecker product, the $\mb K^{1,\sub}$ then must have full rank $2^{K_1}$. The same conclusion also holds for $\mb K^2$,  $\mb K^3$ and we have proved the earlier claim that each of $\mb K^{1}$, $\mb K^{2}$, and  $\mb K^3$ has full column rank.
Note that after marginalizing out the deep latent variable $z$, the model for $\bo y$ is simply a constrained latent class model with $2^{K_1}$ latent classes indexed by $\aaa\in\{0,1\}^{K_1}$. 
Now invoking Corollary \ref{cor-three} gives the identifiability of $\LLambda^{(1)}$, $\GG^{(1)}$, and $\bo \nu^{(1)}$. 
This proves part (a) of the proposition.

We next prove part (b) of the proposition. The conclusion of part (a) guarantees the identifiability of $\nnu^{(1)} := \{\mathbb P(\aaa^{(1)} = \aaa);\,\aaa\in\{0,1\}^{K_1}\}$, so now the $\aaa^{(1)}$ can be viewed as if they are observed variables with known probability mass function $\nnu^{(1)}$.
Then it suffices to note that the distribution of $\aaa^{(1)}$ is simply an unconstrained latent class model with $H$ latent classes. Under $K_1 \geq 2\ceil{\log_2 B}+1$, Corollary 5 in \cite{allman2009} directly gives the generic identifiability conclusion for $\bo\tau$ and $\bo\eta$.
The proof of the proposition is complete.

\subsection{Proof of Proposition \ref{prop-gendiag}}
Denote the three sets of row indices forming $\GG^{(1),1}$, $\GG^{(1),2}$, and $\GG^{(1),3}$ by $\mc A_1$, $\mc A_2$, and $\mc A_3$; that is $\GG^{(1),i} = \mc A_i$ for each $i=1,2,3$.
We claim that for each $\mc A_i$ with $i=1,2,3$, the induced Khatri-Rao product matrix $\mb K_i = \odot_{j\in\mc A_i} \Lambda^{(j)}$ has full column rank for \textit{generic} parameters under the equality constraint \eqref{eq-equa}.
If this claim is true, then we can invoke Corollary \ref{cor-three} to establish the desired generic identifiability result.
So we next focus on proving this claim.
%We next prove this claim using a similar argument as that in the proof of Theorem \ref{thm-genid}.
The statement in the corollary says that each $\GG^{(1),i}$ takes the following form,
\begin{align}\label{eq-diag}
   \GG^{(1),i} =
   \begin{pmatrix}
       1 & * & \cdots & * \\
       * & 1 & \cdots & * \\
       \vdots & \vdots & \ddots & \vdots \\
       * & * & \cdots & 1 \\
   \end{pmatrix}.
\end{align}
We next consider a $2^{K_1} \times 2^{K_1}$ submatrix of the Khatri-Rao product matrix $\mb K^{\mc A_i} = \odot_{j\in \mc A_1} \Lambda^{(j)}$.
Similarly as in the proof of Proposition \ref{prop-2layer}, we consider the rows of $\mb K^{\mc A_i}$ indexed by response patterns $\bo y_{\mc A_i} \in \times_{j\in\mc A_i} \{1, d_j\}$.
First, in a special case where all the off-diagonal entries of $\GG^{(1),i}$ in \eqref{eq-diag} are equal to zero, there is $\GG^{(1),i}=I_{K_1}$ and Eq.~\eqref{eq-k1sub} in the proof of Proposition \ref{prop-2layer} states that the matrix 
$\mb K^{\mc A_i}(\bo\beta) := \otimes_{j\in \mc A_i}
\mb B^{(j)}(\bo\beta)$ with
\begin{equation}\label{eq-defbj}
\mb B^{(j)}(\bo\beta) = 
    \begin{pmatrix}
        \dfrac{\exp\left(\beta_{j,0,1}\right)}{\sum_{m=1}^{d_j} \exp\left(\beta_{j,0,m} \right) }
        &
        \dfrac{\exp\left(\beta_{j,0,1} + \beta_{j,j,1}\right)}{\sum_{m=1}^{d_j} \exp\left(\beta_{j,0,m} + \beta_{j,j,m} \right)}
        \\[5mm]
        \dfrac{1}{\sum_{m=1}^{d_j} \exp\left(\beta_{j,0,m} \right) }
        &
        \dfrac{1}{\sum_{m=1}^{d_j} \exp\left(\beta_{j,0,m} + \beta_{j,j,m} \right) }
    \end{pmatrix}
\end{equation}
has full rank $2^{K_1}$ when the $\bo\beta$-parameters vary in the $I_{K_1}$-constrained parameter space 
$$
\Omega(\bo\beta;\, \GG^{(1),i} =\mb I_{K_1}) := \{\beta_{1:K_1,1:K_1,1:(d_j-1)};\,\beta_{j,j,c}\neq 0~\text{and}~\beta_{j,c,k}=0 ~\text{if}~ k\neq j\}.
$$
Now an important observation is that for an arbitrary $\GG^{(1),i} =: \GG^\star$ taking the form of \eqref{eq-diag}, its corresponding constrained parameter space 
$\Omega(\bo\beta;\, \GG^{(1),i} = \GG^\star)$ 
must contain some $\bo\beta$ that is arbitrarily close to some $\bo\beta^\star$ belonging to $\Omega(\bo\beta;\, \GG^{(1),i} =\mb I_{K_1})$.
Mathematically, this means: 
\begin{itemize}
    \item[(C)] For any small positive number $\epsilon > 0$ and any $\bo\beta \in \Omega(\bo\beta;\, \GG^{(1),i} =\mb I_{K_1})$, there must be some $\bo\beta^\star \in \Omega(\bo\beta;\, \GG^{(1),i} = \GG^\star)$ such that $\norm{\bo\beta - \bo\beta^\star} < \epsilon$.
\end{itemize}
The reasoning behind this claim is as follows. Recall the definition of the parameter space $\Omega(\bo\beta;\,\GG^{(1)})$ for a general graphical matrix $\GG^{(1)}$ given in \eqref{eq-omegab} in the main text,
$$
\Omega(\bo\beta;\, \GG^{(1)}) = \{\beta_{1:p,\, 1:K_1,\, 1:(d_j-1)};\, \beta_{j,c,k} \neq 0~\text{if}~g^{(1)}_{j,k}=1;~\text{and}~\beta_{j,c,k}=0~\text{if}~g^{(1)}_{j,k}=0\}.
$$
By this definition, considering the structure of $\GG^\star$ with diagonal elements being one and some off-diagonal elements potentially also being one, the space $\Omega(\bo\beta;\, \GG^{(1),i} = \GG^\star)$  can be obtained from $\Omega(\bo\beta;\, \GG^{(1),i} =\mb I_{K_1})$ by performing the following operation:
\begin{itemize}
    \item[(P)] If $g^{(1),\star}_{j,k} = 1$ in $\GG^\star$, then convert the equality constraint $\beta_{j,c,k} = 0$ in $\Omega(\bo\beta;\, \GG^{(1),i} =\mb I_{K_1})$ into the equality constraint $\beta_{j,c,k} \neq 0$.
\end{itemize}
Based on the above (P) that $\Omega(\bo\beta;\, \GG^{(1),i} = \GG^\star)$ is obtained through changing certain equality constraints in $\Omega(\bo\beta;\, \GG^{(1),i} =\mb I_{K_1})$ to inequality constraints, it is not hard to see that the previous claim (C) is true.
Recall for any $\bo\beta\in \Omega(\bo\beta;\, \GG^{(1),i} = \GG^\star)$, the $\mb K^{\mc A_i}$ defined earlier has full rank, that is,
$\text{det}\left( \mb K^{\mc A_i}(\bo\beta) \right) \neq 0$, where $\text{det}(\mb M)$ denotes the determinant of matrix $\mb M$.
Since the $\text{det}\left( \mb K^{\mc A_i}(\bo\beta) \right)$ can be viewed as a continuous function of parameters $\bo\beta$, we obtain that there exists some $\bo\beta^\star\in \Omega(\bo\beta;\, \GG^{(1),i} = \GG^\star)$ close enough to $\bo\beta$ (see claim (C)) such that  
$\text{det}\left( \mb K^{\mc A_i}(\bo\beta^\star) \right) \neq 0$.
According to the multinomial logistic model specified in \eqref{eq-2layer}, the $\text{det}\left( \mb K^{\mc A_i}(\bo\beta) \right)$ is actually an {analytic} function of $\bo\beta$, and $\text{det}\left( \mb K^{\mc A_i}(\bo\beta^\star) \right) \neq 0$ for some $\bo\beta^\star \in \Omega(\bo\beta;\, \GG^{(1),i} = \GG^\star)$ implies that the following subset 
$$
\{\bo\beta^\star\in\Omega(\bo\beta;\, \GG^{(1),i} = \GG^\star):\, \text{det}\left( \mb K^{\mc A_i}(\bo\beta) \right) = 0 \}
$$
has measure zero with respect to the Lebesgue measure on $\Omega(\bo\beta;\, \GG^{(1),i} = \GG^\star)$.
Thus we have proved that for generic parameters, the $\odot_{j\in \mc A_i} \Lambda^{(j)}$ has full column rank $2^{K_1}$.
Having three full rank structures $\odot_{j\in \mc A_i} \Lambda^{(j)}$ for $i=1,2,3$, we can invoke the Kruskal's theorem in Lemma \ref{lem-kruskal} to conclude that the parameters $\GG^{(1)}$, $\bo\beta$, and $\nnu^{(1)} := \{\mathbb P(\aaa^{(1)} = \aaa);\,\aaa\in\{0,1\}^{K_1}\}$ are generically identifiable.
Then combined with the condition that $K_1 \geq 2\ceil{\log_2 B}+1$, as shown in the last part of the proof of Proposition \ref{prop-2layer} we also have the generic identifiability of $\bo\eta$ and $\bo\tau$.
This concludes the proof of Proposition \ref{prop-gendiag}.

%%%%%%%%%%%%%%%%%%
%% proof of lemma
\begin{proof}[Proof of Lemma \ref{lem-krprod}]
We prove the equality \eqref{eq-krprod} by showing the corresponding elements of the two matrices on the left and right hand sides (LHS and RHS) are equal. Suppose each $\mathbf A_j=(a^j_{m,\ell})$ has dimension $x_j\times y_j$ and each $\mathbf B_j=(b^j_{m,\ell})$ has dimension $y_j \times H$. Then both LHS and RHS have dimension $\prod_{j=1}^p x_j \times H$. For arbitrary $\bo c=(c_1,\ldots, c_p)^\top$ with each $c_j\in[x_j]$ and arbitrary $h\in[H]$, there is
\begin{align*}
	\text{\normalfont{LHS}}_{\bo c, h} 
	&= 
	\prod_{j=1}^p \{\mathbf A_j  \mathbf B_j\}_{c_j, h}
	=\prod_{j=1}^p\left( \sum_{m=1}^{y_j} a^j_{c_j,m} b^j_{m,h} \right)
	=\sum_{m_1=1}^{y_1}\cdots \sum_{m_p=1}^{y_p} \prod_{j=1}^p a^j_{c_j,m_j} b^j_{m_j,h};\\
	% RHS
	\text{\normalfont{RHS}}_{\bo c, h} 
	&= \sum_{\bo m=(m_1,\ldots,m_p),\atop m_j\in[y_j]}\left\{\otimes_{j=1}^p \mathbf A_j \right\}_{\bo c,\bo m} 
	\times 
	\left\{\odot_{j=1}^p \mathbf B_j\right\}_{\bo m,h}
	=\sum_{m_1=1}^{y_1}\cdots \sum_{m_p=1}^{y_p} \prod_{j=1}^p a^j_{c_j,m_j} b^j_{m_j,h},
\end{align*}
therefore $\text{\normalfont{LHS}}_{\bo c, h} = \text{\normalfont{RHS}}_{\bo c, h}$ for any $\bo c$ and $h$. This establishes the equality in the lemma and completes the proof.
\end{proof}

% proof of the mult-lemma
\begin{proof}[Proof of Lemma \ref{lem-mult}]
We can write out the explicit form of $\KK(\{\bo\Lambda^{(j)}-\bo\delta^{(j)}\bcdot\one^\top_k \}_{j\in[p]})$ as follows,
\begin{align*}
	&~ \odot_{j=1}^p 
	\begin{pmatrix}
		\lambda^{(j)}_{1,1}-\delta^{(j)}_{1} & \lambda^{(j)}_{2,1}-\delta^{(j)}_{1} & \cdots & \lambda^{(j)}_{H,1}-\delta^{(j)}_{1}\\
		\vdots & \vdots & \vdots & \vdots \\
		\lambda^{(j)}_{1,d_j-1}-\delta^{(j)}_{d_j-1} & \lambda^{(j)}_{2,d_j-1}-\delta^{(j)}_{d_j-1} & \cdots & \lambda^{(j)}_{H,d_j-1}-\delta^{(j)}_{d_j-1} \\
			1 & 1 & \cdots & 1
	\end{pmatrix} \\
	=&~ \odot_{j=1}^p 
	\left\{
	\begin{pmatrix}
		1 & 0 & \cdots & 0 & -\delta^{(j)}_1 \\
		0 & 1 & \cdots & 0 & -\delta^{(j)}_2 \\
		\vdots & \vdots & \ddots & \vdots & \vdots \\
		0 & 0 & \cdots & 1 & -\delta^{(j)}_{d_j-1}\\
		0 & 0 & \cdots & 0 & 1
	\end{pmatrix}
	\bcdot
	\begin{pmatrix}
		\lambda^{(j)}_{1,1} & \lambda^{(j)}_{2,1} & \cdots & \lambda^{(j)}_{H,1}\\
		\vdots & \vdots & \vdots & \vdots \\
		\lambda^{(j)}_{1,d_j-1} & \lambda^{(j)}_{2,d_j-1} & \cdots & \lambda^{(j)}_{H,d_j-1} \\
			1 & 1 & \cdots & 1
	\end{pmatrix} 
	\right\}
	\\
	%%%
	=&~
	\left\{
	\bigotimes_{j=1}^p 
	\begin{pmatrix}
		1 & 0 & \cdots & 0 & -\delta^{(j)}_1 \\
		0 & 1 & \cdots & 0 & -\delta^{(j)}_2 \\
		\vdots & \vdots & \ddots & \vdots & \vdots \\
		0 & 0 & \cdots & 1 & -\delta^{(j)}_{d_j-1}\\
		0 & 0 & \cdots & 0 & 1
	\end{pmatrix}
	\right\}
	\bcdot
	\left\{
	\odot_{j=1}^p 
	\begin{pmatrix}
		\lambda^{(j)}_{1,1} & \lambda^{(j)}_{2,1} & \cdots & \lambda^{(j)}_{H,1}\\
		\vdots & \vdots & \vdots & \vdots \\
		\lambda^{(j)}_{1,d_j-1} & \lambda^{(j)}_{2,d_j-1} & \cdots & \lambda^{(j)}_{H,d_j-1} \\
			1 & 1 & \cdots & 1
	\end{pmatrix} 
	\right\} 
	\\
	=&:
	\left\{\bigotimes_{j=1}^p \bo\Delta^{(j)}\right\} \bcdot \KK(\{\bo\Lambda^{(j)} \}_{j\in[p]}),
\end{align*}
where the second equality above also results from Lemma \ref{lem-krprod}.
Because each $\bo\Delta^{(j)}$ in the above display is an upper-triangular matrix with all the diagonal entries equal to one, the Kronecker product of them is an invertible matrix of size $\prod_{j=1}^p d_j \times \prod_{j=1}^p d_j$.
Therefore we can take $\mathbf B(\{\bo\delta^{(j)}\}_{j\in[p]})= \bigotimes_{j=1}^p \bo\Delta^{(j)}$ and this completes the proof of the lemma.
\end{proof}

\section{Posterior Computation Details}
\label{sec-poscomp}

\subsection{Gibbs Sampler for a Fixed $K=K_1$ (Number of Binary Latent Traits in the Middle Layer)}\label{sec-gibbs2layer}
% Prior specification for the weight parameters are
% \begin{align*}
% &\beta_{jck}
%     \mid(\sigma_{ck}^2,~ g_{j,k}=1) \sim N(0, \sigma_{ck}^2);
%     \\
%     %\qquad
% %
% &\beta_{jck}
%     \mid(\sigma_{ck}^2,~ g_{j,k}=0) \sim N(0, v_0^2);~\text{(pseudo prior)}\\
% %
% &\mathbb P(g_{j,k}=1) = 1- \mathbb P(g_{j,k}=0) = \gamma; \\
% %
% &\sigma_{ck}^2  \sim \text{InvGa}(a_{\sigma}, b_{\sigma}).
% \end{align*}
% %Here $v_0$ is a small positive number as suggested in \cite{ishwaran2005spike}.
% For the intercept parameters, let $\beta_{jc0} \sim N(\mu_{c0}, \sigma^2_{c0})$ where $\sigma^2_{c0} = 10$ would give a weakly informative prior as suggested in \cite{gelman2008weakly}  for the logistic regression intercepts.
% Introduce a notation 
% $$\theta_{jck} = \sigma_{ck}^2 (g_{j,k} + (1-g_{j,k}) v_0),$$
% then $\theta_{jck}$ is the  variance for $\beta_{jck}$ as a function of $g_{j,k}$.
The posterior distribution is
\begin{align}\label{eq-pos}
& p(\bo\beta,\GG^{(1)},\bo\tau,\bo\eta, \mb A, \mb Z\mid \mb Y)
\\ \notag
    \propto & ~\prod_{i=1}^n \prod_{j=1}^p \prod_{c=1}^d \prod_{k=1}^{K_1}
    \frac{\left[\exp\left(\beta_{jc0} + \sum_{k=1}^{K_1} \beta_{jck} g_{j,k} a_{i,k} \right)\right]^{\mathbbm{1}(y_{ij} = c )}}{\sum_{l=1}^{d} \exp\left(\beta_{jl0} + \sum_{k=1}^{K_1} \beta_{jlk} g_{j,k} a_{i,k} \right)}
    \times
    \prod_{i=1}^n
    \prod_{b=1}^B
    \left[ \tau_b \prod_{k=1}^{K_1} \eta_{kb}^{a_{ik}} (1-\eta_{kb})^{1-a_{ik}} \right]^{z_{ib}}
    \\ \notag
    &\times
    \prod_{c=1}^{d-1}\prod_{j=1}^p \exp\left\{ -1/2\sigma_{c0}^{-2} (\beta_{jc0} - \mu_{c0})^2 \right\}
    \\ \notag
    % \theta part begins
    &\times
    %\prod_{c=1}^{d-1}
    \prod_{k=1}^{K_1} \prod_{j=1}^p 
    \prod_{c=1}^{d-1}
    \left[\sigma_{ck}^{-1}  \exp\left( -1/2 \sigma_{ck}^{-2}
    \beta_{jck}^2 \right)\right]^{g_{j,k}}
    \left[v_0^{-1}  \exp\left( -1/2 v_0^{-2}
    \beta_{jck}^2 \right)\right]^{1-g_{j,k}}
    \\ \notag
    &\times
    \prod_{j=1}^p \prod_{k=1}^{K_1} \gamma^{g_{j,k}} (1-\gamma)^{1-g_{j,k}}
    \times
    \prod_{k=1}^{K_1} \prod_{c=1}^{d-1} 
    \sigma_{ck}^{-2a_{\sigma}-2} \exp\left( -b_{\sigma} \sigma_{ck}^{-2}\right)
    \\ \notag
    &\times
    \prod_{i=1}^n \prod_{j=1}^p \prod_{c=1}^{d-1} \text{PG}(w_{ijc}\mid 1,0)
    \times 
    \text{Dirichlet}(\bo\tau;\, \aaa_{\tau})
    \times \prod_{b=1}^B \prod_{k=1}^{K_1} \text{Beta}(\eta_{kb};\, \aaa_{\eta})
    \times \text{Beta}(\gamma;\, \aaa_{\gamma}).
\end{align}
Here the PG$(\bcdot\mid 1, 0)$ denotes the density of the Polya-Gamma distribution \citep{polson2013pg}.
We simply take $\aaa_{\tau} = \one_B$ and $\aaa_{\eta} = \aaa_{\gamma} = (1,1)^\top$.

%%%%
\bigskip
\noindent
\textbf{(1) Conditional distribution for $\bo\beta$.}
We use the auxiliary variable approach in \cite{holmes2006} for multinomial regression to derive the conditional distribution of each $\beta_{jck}$.
Recall the posterior $p(\bo\beta,\mb W \mid \bo y_{1:n})$ in Section \ref{sec-method}.
% Given data $y_{ij}\in[d_j]$, introduce binary indicators $y_{ijc} = \mathbbm{1}(y_{ij} = c)$.
% %
% For $c \in [d-1]$,
% \begin{align*}
% \mc L(\bo y_1,\ldots,\bo y_n\mid \bo\beta)
%     &\propto
%     \prod_{i=1}^n \prod_{l=1}^{d-1}
%     \frac{\left[\exp\left(\beta_{jl0} + \sum_{k=1}^{K_1} \beta_{jlk} g_{j,k} a_{i,k} \right)\right]^{\mathbbm{1}(y_{ij} = l )}}{\sum_{l'=1}^{d} \exp\left(\beta_{jl'0} + \sum_{k=1}^{K_1} \beta_{jl'k} g_{j,k} a_{i,k} \right)}
%     \\
%     &\propto
%     \prod_{i=1}^n \prod_{l=1}^{d-1} \xi_{ijl}^{y_{ijl}} (1-\xi_{ijl})^{1-y_{ijl}},
% \end{align*}
% where there is
% \begin{align}\label{eq-defphi}
%     \xi_{ijl} &= \frac{\exp\{ \beta_{jl0} + \sum_{k=1}^{K_1} \beta_{jlk} g_{j,k} a_{i,k} - C_{ij(l)}\}}{1 + \exp\{\beta_{jl0} + \sum_{k=1}^{K_1} \beta_{jlk} g_{j,k} a_{i,k} - C_{ij(l)}\}}
%     =: \frac{\exp(\phi_{ijl})}{1+\exp(\phi_{ijl})} ;
%     \\ \notag
%     C_{ij(l)} &= \log\left\{ \sum_{ 1\leq c\leq d,\,c\neq l} \exp\left( \beta_{jl0} + \sum_{k=1}^{K_1} \beta_{jlk} g_{j,k} a_{i,k} \right) \right\}.
% \end{align}

For an arbitrary given graphical vector $\bo g_{j,\bcolon} = (g_{j,1},\ldots,g_{j,K})$, define $\mc K_j = \{k\in[K]:\, g_{j,k}=1\}$, the set of parent latent variables of variable $y_j$.
Then denote $\bo\beta_{jc,\,\text{eff}} = (\beta_{jck};\, k\in \mc K_j)^\top$, 
$\tilde{\bo\beta}_{jc,\,\text{eff}} = (\beta_{jc0},\bo\beta_{jc,\,\text{eff}}^\top)^\top$;
also denote ${\bo a}_{i,\,\text{eff}} = (a_{i,k};\, k\in \mc K_j)^\top$, 
$\tilde{\bo a}_{i,\,\text{eff}} = (1,{\bo a}_{i,\,\text{eff}}^\top)^\top$.
Then we can write 
%\begin{equation}\label{eq-writephi}
$\phi_{ijc} 
= \tilde{\bo a}_i^\top \tilde{\bo\beta}_{jc,\bcolon} - C_{ij(c)}$, where
$$
C_{ij(c)} = 
\log\left\{ \sum_{ 1\leq \ell\leq d_j,\,\ell\neq c} \exp\left( \beta_{j\ell 0} + \sum_{k=1}^{K_1} \beta_{j\ell k} g_{j,k} \alpha_{i,k} \right) \right\}.
$$
%\end{equation}
By the property of the Polya-Gamma augmentation \citep{polson2013pg}, the conditional distribution of $\tilde{\bo\beta}_{jc,\bcolon}$ is
\begin{align*}
 &~ p(\tilde{\bo\beta}_{jc,\,\text{eff}} \mid \mb Y, \mb W)\\
= &~ 
p(\tilde{\bo\beta}_{jc,\,\text{eff}}) \prod_{i=1}^n \mc L(\bo y_i\mid\bo\beta)
\\
\propto &~
p(\tilde{\bo\beta}_{jc,\,\text{eff}})
    \prod_{i=1}^n  \xi_{ijc}^{y_{ijc}} (1-\xi_{ijc})^{1-y_{ijc}}
=p(\tilde{\bo\beta}_{jc,\,\text{eff}})
    \prod_{i=1}^n  \frac{[\exp(\phi_{ijc})]^{y_{ijc}}}{1+\exp(\phi_{ijc})}
\\
\propto &~
p(\tilde{\bo\beta}_{jc,\,\text{eff}})
\prod_{i=1}^n \exp\left\{(y_{ijc}-1/2) \phi_{ijc} \right\}
\exp\left\{ - 1/2 w_{ijc} \phi_{ijc}^2 \right\}
\\
\propto &~
p(\tilde{\bo\beta}_{jc,\,\text{eff}})
\prod_{i=1}^n 
\exp\left\{ - \frac{1}{2} w_{ijc} \left( \phi_{ijc} - \frac{y_{ijc}-1/2}{w_{ijc}} \right)^2 \right\} 
\\
&~(\text{use}~\phi_{ijc} = \tilde{\bo a}_i^\top \tilde{\bo\beta}_{jc,\,\text{eff}} - C_{ij(c)} ~\text{and define}~ x_{ijc} = \frac{y_{ijc}-1/2}{w_{ijc}} + C_{ij(c)})
\\
\propto &~
p(\tilde{\bo\beta}_{jc,\,\text{eff}})
\prod_{i=1}^n 
\exp\left\{ - \frac{1}{2} w_{ijc} \left( \tilde{\bo a}_{i,\,\text{eff}}^\top \tilde{\bo\beta}_{jc,\,\text{eff}} - x_{ijc} \right)^2 \right\}
\\
\propto &~
p(\tilde{\bo\beta}_{jc,\,\text{eff}})
%\prod_{i=1}^n 
\exp\left\{ - \frac{1}{2} 
\left(
 x_{\bcolon,jc} - \tilde{\mb A}_{\bcolon,\,\text{eff}}\tilde{\bo\beta}_{jc,\,\text{eff}}
\right)^\top
\diag(\mb W_{\bcolon,jc})
\left(
 x_{\bcolon,jc} - \tilde{\mb A}_{\bcolon,\,\text{eff}}\tilde{\bo\beta}_{jc,\,\text{eff}}
\right)
\right\}
\\
\propto &~
p(\tilde{\bo\beta}_{jc,\,\text{eff}})
%\prod_{i=1}^n 
\exp\left\{ - \frac{1}{2} 
\left(
\tilde{\bo\beta}_{jc,\,\text{eff}} -
 \tilde{\mb A}_{\bcolon,\,\text{eff}}^{-1} x_{\bcolon,jc} 
\right)^\top
\tilde{\mb A}_{\bcolon,\,\text{eff}}^{\top}\diag(\mb W_{\bcolon,jc})\tilde{\mb A}_{\bcolon,\,\text{eff}}
\left(
\tilde{\bo\beta}_{jc,\,\text{eff}} -
 \tilde{\mb A}_{\bcolon,\,\text{eff}}^{-1} x_{\bcolon,jc} 
\right)\right\},
\end{align*}
where $\tilde{\mb A}_{\bcolon,\,\text{eff}}$ is a $n\times (|\mc K_j|+1)$ matrix with the $i$th row being $\tilde{\bo a}_{i,\,\text{eff}}$.
Define $\bo\mu_{0c} = (\mu_{0c},0,\ldots,0)^\top$, a vector of length $K_1+1$, and ${\bo\Sigma}_{0jc} = \diag\left(\sigma_{c0}^{2}, \{\sigma^2_{ck};~k\in\mc K_j\} \right)$, a diagonal matrix of size $(|\mc K_j|+1) \times (|\mc K_j|+1)$.
Considering the prior for $\tilde{\bo\beta}_{jc,\,\text{eff}}$, we have
\begin{align*}
p(\tilde{\bo\beta}_{jc,\,\text{eff}} \mid -)
\propto &~
\exp\left\{ - \frac{1}{2} \left(\tilde{\bo\beta}_{jc,\,\text{eff}} - \bo\mu_{0c} \right)^\top 
{\bo\Sigma}_{0jc}^{-1}
\left(\tilde{\bo\beta}_{jc,\,\text{eff}} - \bo\mu_{0c} \right) \right\}
\\
&~ \times
\exp\left\{ - \frac{1}{2} 
\left(
\tilde{\bo\beta}_{jc,\,\text{eff}} -
 \tilde{\mb A}_{\bcolon,\,\text{eff}}^{-1} x_{\bcolon,jc} 
\right)^\top
\tilde{\mb A}_{\bcolon,\,\text{eff}}^{\top}\diag(\mb W_{\bcolon,jc})\tilde{\mb A}_{\bcolon,\,\text{eff}}
\left(
\tilde{\bo\beta}_{jc,\,\text{eff}} -
 \tilde{\mb A}_{\bcolon,\,\text{eff}}^{-1} x_{\bcolon,jc} 
\right)\right\},
\end{align*} 
so 
$(\tilde{\bo\beta}_{jc,\,\text{eff}} \mid -) \sim
N(\bo\mu_{jc}, {\bo\Sigma}_{jc})$
where 
\begin{align}\label{eq-posbeta}
\bo\mu_{jc} = {\bo\Sigma}_{jc}\left(
\tilde{\mb A}^{\top}\diag(\mb W_{\bcolon,jc}) x_{\bcolon,jc}
 + {\bo\Sigma}_{0jc}^{-1}\bo\mu_{0c} \right),
\quad
{\bo\Sigma}_{jc} = \left(\tilde{\mb A}^{\top}\diag(\mb W_{\bcolon,jc})\tilde{\mb A} + {\bo\Sigma}_{0jc}^{-1}\right)^{-1}.
\end{align}
As for those $k\notin \mathcal K_j$ with $g_{j,k} = 0$, we sample $\beta_{jkc}$ from the pseudo prior $N(0,v_0^2)$. In summary, for any $j\in[p]$, we sample $\bo\beta_{j,\bcolon,\bcolon} = (\tilde{\bo\beta}_{jc,\,\text{eff}}, ~ \bo\beta_{jc,\,\text{ine}})$ as follows
\begin{align*}
    \tilde{\bo\beta}_{jc,\,\text{eff}} \sim N(\bo\mu_{jc}, {\bo\Sigma}_{jc});
    \quad
    \bo\beta_{jc,\,\text{ine}} \sim N(0, v_0^2 \mb I_{p-|\mc K_j|}).
\end{align*}

%%%%
\bigskip
\noindent
\textbf{(2) Conditional distribution for Polya-Gamma random variables $w_{ijc}$.}
\begin{align*}
    p(w_{ijc}\mid -)
    \propto
    \text{PolyaGamma}\mleft(w_{ijc} \;\middle|\; 1,~ \beta_{jc0} + \sum_{k=1}^{K_1} \beta_{jck} g_{j,k} a_{i,k} - C_{ij(c)}\mright).
\end{align*}

%%%%%
\bigskip
\noindent
\textbf{(3) Conditional distribution for $g_{j,k}$.}
\begin{align*}
    &~\mathbb P(g_{j,k} = 1\mid -)
    =
    \frac{\gamma}
    {\gamma 
    +
    (1-\gamma) O^{01}_{jk}};
    \\
    &~O^{01}_{jk}
    =
    \prod_{c}
    \frac{v_0^{-1}  \exp\left( -1/2 v_0^{-2}
    \beta_{jck}^2 \right)}
    {\sigma_{ck}^{-1}  \exp\left( -1/2 \sigma_{ck}^{-2}
    \beta_{jck}^2 \right)}
    \prod_{i}\prod_{c} 
    \frac{ p(y_{ij} = c\mid \mb -, g_{j,k}=0)}
    {p(y_{ij} = c\mid \mb -, g_{j,k}=1)};
    % \frac{\gamma \prod_{c=1}^{d-1} \sigma_{ck}^{-1} \exp\left(-1/2 \sigma_{ck}^{-2} \beta^2_{jck}\right) 
    % \prod_{i}\prod_{c} p(y_{ij} = c\mid \mb -, g_{j,k}=1)}
    % {\gamma \prod_{c=1}^{d-1} \sigma_{ck}^{-1} \exp\left(-1/2 \sigma_{ck}^{-2} \beta^2_{jck}\right) \prod_{i}\prod_{c} p(y_{ij} = c\mid \mb -, g_{j,k}=1)
    % +
    % (1-\gamma) \prod_{c=1}^{d-1} v_0^{-1/2}\sigma_{ck}^{-1} \exp\left(-1/2 v_0^{-1}\sigma_{ck}^{-2} \beta^2_{jck}\right)\prod_{i}\prod_{c} p(y_{ij} = c\mid \mb -, g_{j,k}=0)}
    %%%%%%%
    % \frac{\gamma}
    % {\gamma 
    % +
    % (1-\gamma) v_0^{-(d-1)/2}
    % \prod_{c}\sigma_{ck}^{-1} \exp\left\{-1/2 (v_0^{-1}-1) \sum_{c}\sigma_{ck}^{-2} \beta^2_{jck}\right\}
    % \prod_{i}\prod_{c} 
    % \frac{ p(y_{ij} = c\mid \mb -, g_{j,k}=0)}
    % {p(y_{ij} = c\mid \mb -, g_{j,k}=1)}}.
\end{align*}
Also there is $\gamma\sim \text{Beta}(1+\sum_{j=1}^p\sum_{k=1}^K g_{j,k},\, 1 + pK -\sum_{j=1}^p\sum_{k=1}^K g_{j,k})$.

%%%%
\bigskip
\noindent
\textbf{(4) Conditional distribution for $\sigma_{ck}^2$.}
\begin{align*}
p(\sigma_{ck}^2 \mid -)
    \propto &~
    \prod_{j=1}^p
    \left[\sigma_{ck}^{-1}  \exp\left( -1/2 \sigma_{ck}^{-2}
    \beta_{jck}^2 \right)\right]^{g_{j,k}}
    \times
    \sigma_{ck}^{-2(a_{\sigma}+1)} \exp\left( -b_{\sigma} \sigma_{ck}^{-2}\right)
    \\
    \propto &~
\sigma_{ck}^{-2(1/2 \sum_{j=1}^p g_{j,k} + a_{\sigma}+1)}
\exp\left\{ - \sigma_{ck}^{-2} \left(\frac{1}{2} \sum_{j=1}^p  g_{j,k} \beta_{jck}^2  + b_\sigma \right) \right\},
\end{align*}
% \begin{align*}
% p(\sigma_{ck}^2 \mid -)
%     \propto &~
%     \prod_{j=1}^p \sigma_{ck}^{-1} (g_{j,k} + (1-g_{j,k}) v_0)^{-1/2} \exp\left\{ - \frac{1}{2} \sigma_{ck}^{-2} (g_{j,k} + (1-g_{j,k}) v_0)^{-1}\beta_{jck}^2 \right\}
%     \\
%     &~\times
%     \sigma_{ck}^{-2(a_{\sigma}+1)} \exp\left( -b_{\sigma} \sigma_{ck}^{-2}\right)
%     \\
%     \propto &~
% \sigma_{ck}^{-2(p/2 + a_{\sigma}+1)}
% \exp\left\{ - \sigma_{ck}^{-2} \left(\frac{1}{2} \sum_{j=1}^p  (g_{j,k} + (1-g_{j,k}) v_0)^{-1} \beta_{jck}^2  + b_\sigma \right) \right\},
% \end{align*}
so 
$$
\sigma_{ck}^2 \mid - \sim \text{InvGa}
\left(
\frac{1}{2}\sum_{j=1}^p g_{j,k} + a_\sigma,~
\frac{1}{2} \sum_{j=1}^p  g_{j,k} \beta_{jck}^2 + b_\sigma
\right).
$$

%%%%
\bigskip
\noindent
\textbf{(5) The conditional distribution for parameters of the deeper latent layer.}
For parameters $\bo\tau$ and $\bo\eta$ underlying the deeper latent layer, 
\begin{align*}
    (\tau_1,\ldots,\tau_B)
    &\sim
    \text{Dirichlet}\left(1+\sum_{i=1}^n z_{i1},\,\ldots,\,1+\sum_{i=1}^n z_{iB}\right);\\
    \eta_{kb} 
    &\sim \text{Beta}\left(1 + \sum_{i=1}^n a_{ik}z_{ib}, 
    ~ 1 + \sum_{i=1}^n (1-a_{ik})z_{ib}\right).
\end{align*}

%%%%
\bigskip
\noindent
\textbf{(6) The conditional distribution for subject-specific local parameters $\mb A$ and $\mb Z$.}
\begin{align*}
    &\mathbb P(a_{ik}=1\mid -)
    =\\
    % &\frac{\prod_{b} [\tau_b\eta_{kb}]^{z_{ib}} \prod_{j}\prod_{c} p(y_{ij} = c\mid -, a_{ik}=1)}{\prod_{b} [\tau_b\eta_{kb}]^{z_{ib}} \prod_{j}\prod_{c} p(y_{ij} = c\mid \mb -, a_{ik}=1) + \prod_{b} [\tau_b (1-\eta_{kb})]^{z_{ib}} \prod_{j}\prod_{c} p(y_{ij} = c\mid \mb -, a_{ik}=0 )};
     &\frac{\prod_{b} \eta_{kb}^{z_{ib}} \prod_{j}\prod_{c} p(y_{ij} = c\mid -, a_{ik}=1)}{\prod_{b} \eta_{kb}^{z_{ib}} \prod_{j}\prod_{c} p(y_{ij} = c\mid \mb -, a_{ik}=1) + \prod_{b} (1-\eta_{kb})^{z_{ib}} \prod_{j}\prod_{c} p(y_{ij} = c\mid \mb -, a_{ik}=0 )};
    \\[3mm]
    &(z_{i1},\ldots,z_{iB})
    \sim\text{Categorical}
    \left(\tau_1 \prod_{k=1}^{K_1} \eta_{k1}^{a_{ik}} (1-\eta_{k1})^{1-a_{ik}},
    \,\ldots,\,
    \tau_B \prod_{k=1}^{K_1} \eta_{kB}^{a_{ik}} (1-\eta_{kB})^{1-a_{ik}}\right).
\end{align*}

Based on the full conditionals in the above (1)--(6), we can have a Gibbs sampler that cycle through these steps. All the conditional distributions are standard distributions that are easy to sample from.

% \section*{Posterior Computation under a Cumulative Shrinkage Process (CSP) Prior}
\subsection{Gibbs Sampler under a Cumulative Shrinkage Process Prior when $K_1$ is Unknown}

Introduce weight $w_\ell := v_\ell\prod_{m=1}^{\ell-1}(1-v_m)$ and introduce latent indicators $h_k$ with $\mathbb P(h_k=\ell\mid w_\ell) = w_{\ell}$, for $\ell\in[K]$.
The joint conditional posterior of $\{\sigma_{ck}^2\}$, $\{v_{k}\}$, and $\{h_k\}$ are
\begin{align*}
    p(\{\sigma_{ck}^2\}, \{h_{k}\}, \{v_{k}\}\mid-)
    \propto &~
    \prod_{k=1}^K\prod_{j=1}^p\left[\prod_{c=1}^{d-1}\sigma_{ck}^{-1}  \exp\left( -1/2 \sigma_{ck}^{-2}
    \beta_{j1k}^2 \right)\right]^{g_{j,k}}
    \\
    &~\times
    \prod_{k=1}^K \left[\prod_{c=1}^{d-1}(\sigma_{ck}^{-2})^{a_{\sigma}-1}\exp\left( -b_{\sigma} \sigma_{ck}^{-2} \right)\right]^{\mathbbm{1}(h_{k}>k)}
    \left[\prod_{c=1}^{d-1}\mathbbm{1}(\sigma_{ck}^2=\theta_{\infty})\right]^{\mathbbm{1}(h_{k}\leq k)}
    \\
    &~\times
    \prod_{\ell=1}^K \left[v_\ell\prod_{m=1}^{\ell-1}(1-v_m)\right]^{\sum_{k=1}^K \mathbbm{1}(h_k=\ell)}
\end{align*}
Based on the above display, we can derive the full conditional distribution for each quantity as follows.
\begin{align*}
(\sigma_{ck}^2 \mid -) 
&\sim
\begin{cases}
    \text{InvGa}\left(a_\sigma + \frac{1}{2}\sum_{j=1}^p g_{j,k},\,
    b_\sigma + \frac{1}{2}\sum_{j=1}^p g_{j,k}\beta^2_{jck}\right), & \text{if}~~ h_k\leq k;\\
    \sigma_{ck}^2 = \theta_{\infty}, & \text{if}~~ h_k > k;
\end{cases}
\\[2mm]
(v_k \mid -) &\sim \text{Beta}\left(1 + \sum_{\ell=1}^K \mathbbm{1}(h_\ell=k), ~
    \alpha + \sum_{\ell=1}^K \mathbbm{1}(h_\ell > k) \right);
    \\[2mm]
\mathbb P(h_k = \ell \mid -) &\propto
\begin{cases}
w_\ell \prod_{c=1}^{d-1} N_p (\beta_{\eff,\, ck};\, 0, \theta_{\infty} \mb I_{p_{\eff}}), & \text{if} ~~ \ell\leq k; \quad(\text{spike})
    \\[2mm]
w_\ell \prod_{c=1}^{d-1} t_{2a_\sigma}\{\beta_{\eff, \, ck};\,  0, (b_\sigma/a_\sigma)\mb I_{p_{\eff}}\}, & \text{if} ~~ \ell> k. \quad(\text{slab})
\end{cases}
\end{align*}

% The $K$ above should be understood as the upper bound of $K_1$ plus one, because the last latent indicator $z_K$ always satisfies $z_K\leq K$ and hence implies a spike distribution for the $K$th column.
With the indicator variables $\{h_k;\, k\in[K]\}$,
the posterior mean of the inferred number of effective latent variables takes the form of 
\begin{equation}\label{eq-kz}
    \mathbb E[K^\star\mid \mb Y] = \mathbb P\mleft(\sum_{k=1}^K \sum_{\ell=k+1}^K \mathbbm{1}(h_k =\ell) \;\middle|\; \mb Y\mright).
\end{equation}

% \color{blue!70!black}
% \subsection{Posterior Computation for Constrained Latent Class Models}
% Parameterize the conditional probability $\lambda_{j,c,k}$ as
% \begin{align*}
% \lambda_{j,c,k} = \frac{\exp(\zeta_{j,c,0} + \zeta_{j,c,k}s_{j,k})}{1 + \sum_{c'=1}^{d_j-1} \exp(\zeta_{j,c',0} + \zeta_{j,c',k}s_{j,k})},\quad c=1,\ldots,d_j, 
% \end{align*}
% where $\zeta_{j,d_j,0} = \zeta_{j,d_j,1}= \cdots = \zeta_{j,d_j,K} = 0$. Such a multinomial logit formulation also allows for developing 

% \color{black}

%\color{blue!70!black}
\subsection{Gibbs Sampler for One-wide-latent-layer Models}
Consider a single-latent-layer model alternative to the multilayer Bayesian pyramids, where the latent layer consists of multiple independent binary latent variables. Compared to the two-latent-layer Bayesian Pyramid, this model does not have a deep latent class variable $z$ beneath the $\aaa$-layer.
We can slightly modify the Gibbs sampler for Bayesian Pyramids (described in Section 4.2 of the main text and Section \ref{sec-gibbs2layer} in this Supplementary Material) to develop a similar Gibbs sampler for the single-latent-layer model.
With similar priors as those for the Bayesian Pyramids and assuming each $\alpha_k\sim \text{Bernoulli}(\eta_k)$, the current posterior distribution can be written as
\begin{align}\label{eq-pos}
& p(\bo\beta,\GG^{(1)},\bo\tau,\bo\eta, \mb A\mid \mb Y)
\\ \notag
    \propto & ~\prod_{i=1}^n \prod_{j=1}^p \prod_{c=1}^d \prod_{k=1}^{K_1}
    \frac{\left[\exp\left(\beta_{jc0} + \sum_{k=1}^{K_1} \beta_{jck} g_{j,k} a_{i,k} \right)\right]^{\mathbbm{1}(y_{ij} = c )}}{\sum_{l=1}^{d} \exp\left(\beta_{jl0} + \sum_{k=1}^{K_1} \beta_{jlk} g_{j,k} a_{i,k} \right)}
    \times \textcolor{black}{\underbrace{\prod_{k=1}^{K_1} \eta_{k}^{a_{ik}} (1-\eta_{k})^{1-a_{ik}}}_{\text{the only difference from }\eqref{eq-pos}}}
    % \times
    % \prod_{i=1}^n
    % \prod_{b=1}^B
    % \left[ \tau_b \prod_{k=1}^{K_1} \eta_{kb}^{a_{ik}} (1-\eta_{kb})^{1-a_{ik}} \right]^{z_{ib}}
    \\ \notag
    &\times
    \prod_{c=1}^{d-1}\prod_{j=1}^p \exp\left\{ -1/2\sigma_{c0}^{-2} (\beta_{jc0} - \mu_{c0})^2 \right\}
    \\ \notag
    % \theta part begins
    &\times
    %\prod_{c=1}^{d-1}
    \prod_{k=1}^{K_1} \prod_{j=1}^p 
    \prod_{c=1}^{d-1}
    \left[\sigma_{ck}^{-1}  \exp\left( -1/2 \sigma_{ck}^{-2}
    \beta_{jck}^2 \right)\right]^{g_{j,k}}
    \left[v_0^{-1}  \exp\left( -1/2 v_0^{-2}
    \beta_{jck}^2 \right)\right]^{1-g_{j,k}}
    \\ \notag
    &\times
    \prod_{j=1}^p \prod_{k=1}^{K_1} \gamma^{g_{j,k}} (1-\gamma)^{1-g_{j,k}}
    \times
    \prod_{k=1}^{K_1} \prod_{c=1}^{d-1} 
    \sigma_{ck}^{-2a_{\sigma}-2} \exp\left( -b_{\sigma} \sigma_{ck}^{-2}\right)
    \\ \notag
    &\times
    \prod_{i=1}^n \prod_{j=1}^p \prod_{c=1}^{d-1} \text{PG}(w_{ijc}\mid 1,0)
    \times 
    \text{Dirichlet}(\bo\tau;\, \aaa_{\tau})
    \times
    \textcolor{black}{\prod_{k=1}^{K_1} \text{Beta}(\eta_{k};\, \aaa_{\eta})}
    \times \text{Beta}(\gamma;\, \aaa_{\gamma}).
\end{align}
Therefore, the posterior distribution of $\eta_{k}$ is 
$$
\eta_{k} 
\sim \text{Beta}\left(1 + \sum_{i=1}^n a_{ik}, 
~ 1 + \sum_{i=1}^n (1-a_{ik})\right).
$$
All the other steps of the Gibbs sampler for sampling other parameters are the same as those detailed in Section \ref{sec-gibbs2layer}.

\color{black}

\section{Additional Numerical Analyses}\label{sec-nume}

%\color{blue!70!black}
\subsection{Overfitted Mixture Model applied to the Deep Latent Class Layer in Bayesian Pyramids for Selecting $B$}

For finite mixture models with an unknown number of mixture components $B$ where only an upper bound of $B$ is known as $B_{\text{upper}}$, \cite{rousseau2011overfit} proposed to use shrinkage priors for the mixture proportion parameters, such as the Dirichlet priors with small Dirichlet hyperparameters. More specifically, denote the mixture proportion parameters corresponding to the $B_{\text{upper}}$ ``overfitted'' mixture components by $\bo\pi := (\pi_1,\ldots,\pi_{B_{\text{upper}}})$, then $\bo\pi$ lives on the $(B_{\text{upper}}-1)$-dimensional probability simplex. The overfitted mixture framework in \cite{rousseau2011overfit} guarantees that with the prior $\bo\pi\sim \text{Dirichlet}(a_1, \ldots,  a_{B_{\text{upper}}})$ where the Dirichlet parameters are small enough with respect to the dimension of the observations, the redundant mixture components will be ``emptied out'' in the posterior. %Hence, the $B$ true mixture components underlying the data will be identified from the posterior distribution. 

Such an overfitted-mixture-model framework provides a convenient way to select cardinality $B$ of the deep latent variable $z$ in our model.
%We have performed numerical experiments for both the simulated data the real data of DNA nucleotide sequences, and in the scenarios we set the hyperparameters of the Dirichlet distribution of the deep latent class proportion to be $0.01$. In particular,
We have performed numerical experiments for both the simulated data and the real data of DNA nucleotide sequences using an overfitted mixture; see the following for details.
In all the scenarios, we set the hyperparameters of the Dirichlet distribution of the deep latent class proportion to be $0.01$.

\medskip
\noindent
\textbf{Simulations.} We have conducted {a new simulation study} using simulated data generated from the setting described in Section 5 in the manuscript, while overfitting $B_{\text{upper}}$ to be 5 (note the true $B=2$ in the simulation). The true deep latent class proportions are $(0.5,~0.5)$. In Figure \ref{fig-overfitmix-simu}, we provide the boxplots of the posterior means of the sorted $(\tau_b;\; b\in[B])$ across the 50 simulation replications. In particular, the first column in  Figure \ref{fig-overfitmix-simu} is the boxplot for the largest deep latent class proportion $\tau_{(1)}$ across replicates, second column the second largest proportion $\tau_{(2)}$, etc. One can see that in most of the simulation trials, only the largest two deep latent classes have significant posterior means while the last three are very close to zero.
    
    \begin{figure}[h!]
        \centering
        \includegraphics[width=0.6\textwidth]{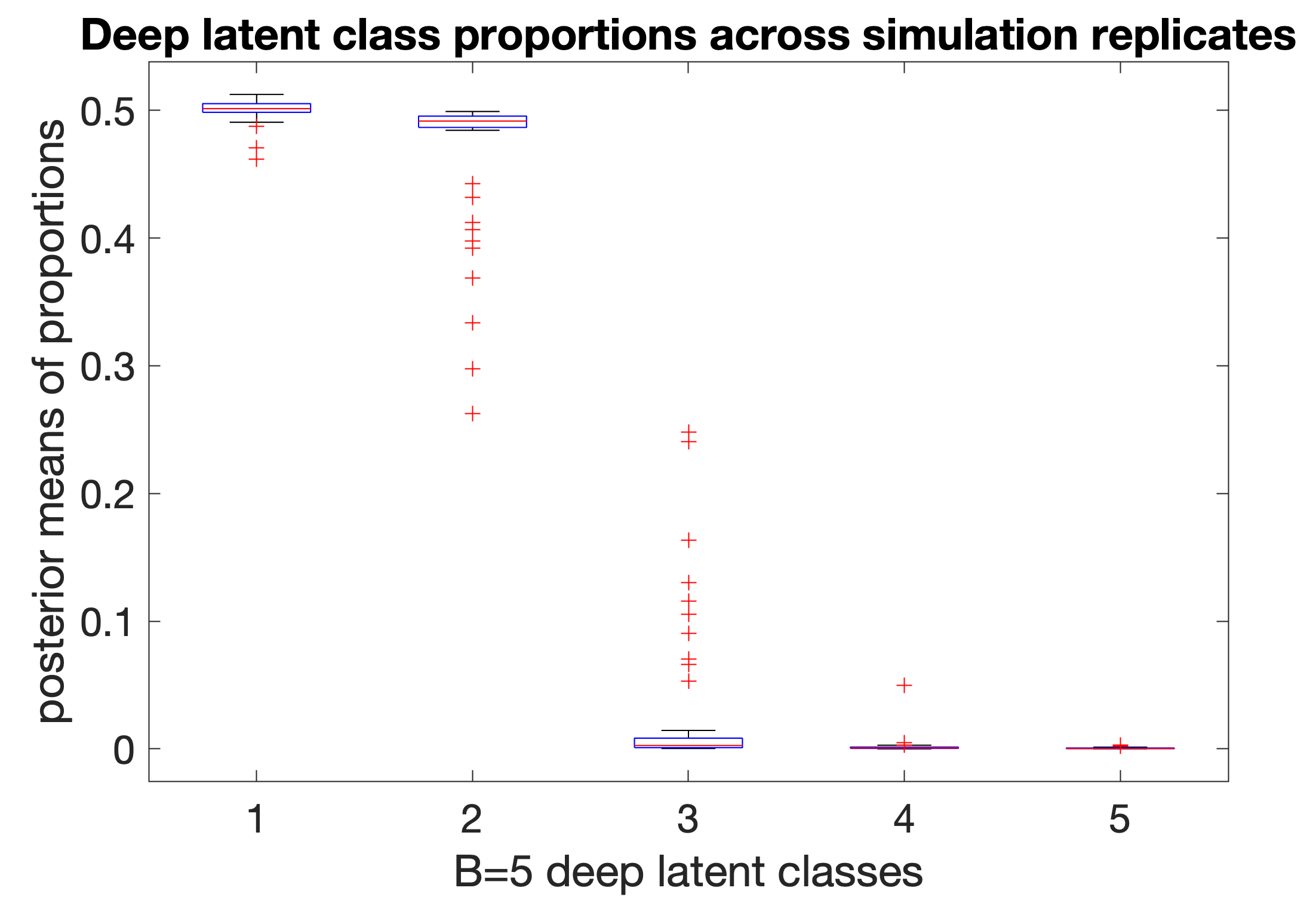}%
        \caption{Simulation results for the overfitted mixture model applied to the deep latent class variable $z\in[B]$. The true deep latent class proportions are $(0.5,~0.5)$ with $B_{\text{true}=2}$, while in the simulation the $B$ is specified to be 5.}
        \label{fig-overfitmix-simu}
    \end{figure}

\medskip
\noindent
\textbf{Splice data application.} We reanalyzed the splice junction data, now specifying the number of deep latent categories to be $B=5$ ($B=2$ for the splice data in the main paper). Results have been obtained under the overfitted mixture framework, where the posterior means for the $B=5$ deep latent classes are: 
    $$
    \bo\tau_{\text{post. mean}} = (8.2\times 10^{-6},~  1.2\times 10^{-5},~    6.0\times 10^{-6},~    0.7517,~    0.2482)^\top,
    $$
    with the first three mixture components having probability weights very close to zero. Hence, three of the classes are effectively deleted, leading to  $B_{\text{eff}}=2$. In the plot of the new data analysis results in Figure \ref{fig-overfitmix-splice}, the first three columns of the posterior mode of the $n\times B$ matrix $\mathbf Z$ are all zero, indicating these categories are effectively emptied out in the posterior.
    From the rightmost two plots, our estimated deep latent component with posterior proportion $0.2482$ seems to be practically meaningful and capture the sequence type ``EI'' of the splice junction.
    We also used the ``rule lists'' approach \citep{angelino2017rule} for downstream classification, and the classification accuracy for all three junction types are similar to those presented in the main text.

\begin{figure}[h!]
    \centering
    \resizebox{0.9\textwidth}{!}{%
    \includegraphics[height=3cm]{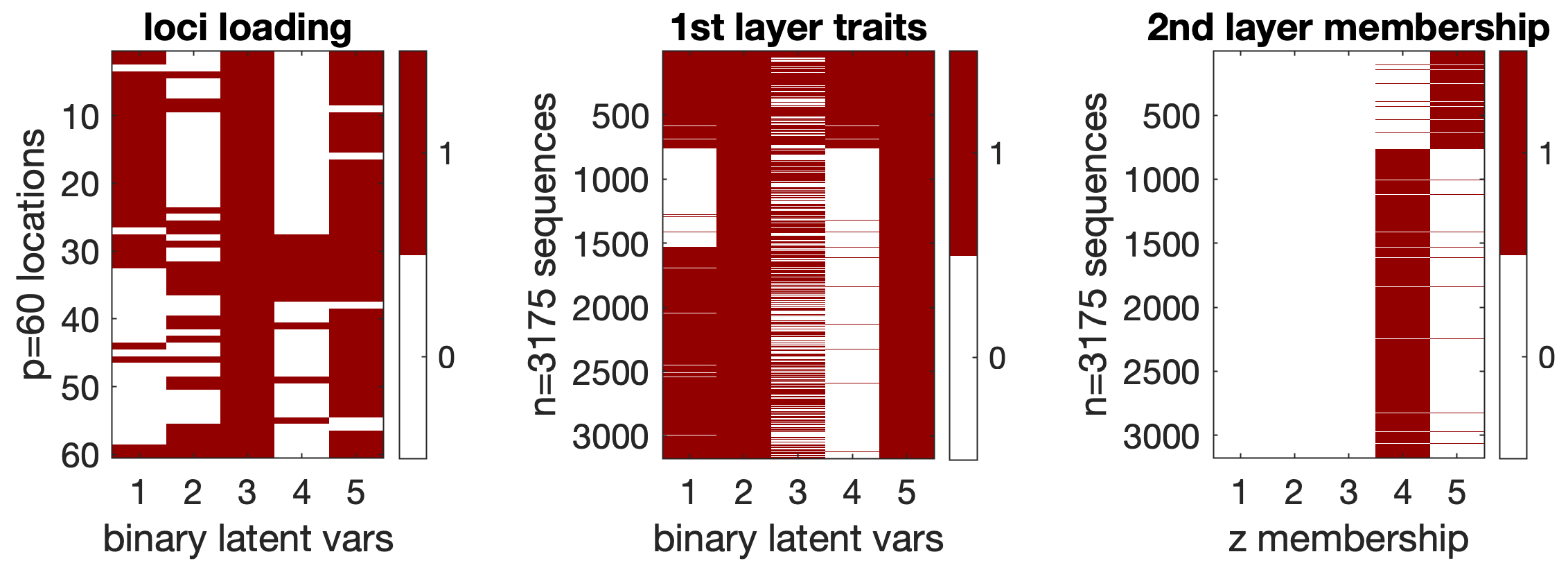}%
    \quad
    \includegraphics[height=3cm]{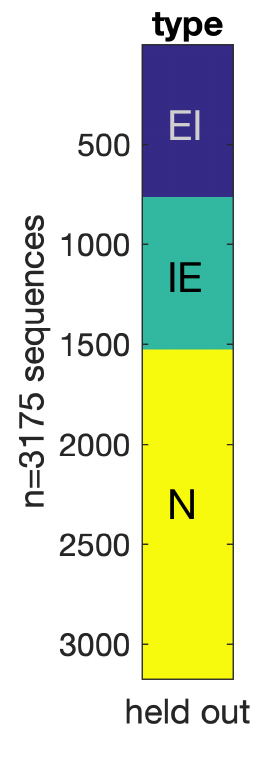}%
    }
    \caption{{\textit{Overfitted} mixture framework applied to the deep latent class variable $z\in\{1,\ldots,B\}$ for the splice junction data. Here $B$ is specified to be $B=5$ while only $B_{\text{eff}}=3$ effective latent categories remain in the posterior, as shown by the third plot.}}
    \label{fig-overfitmix-splice}
\end{figure}

% \end{itemize}

\color{black}
\subsection{Analysis of the Promoter Dataset}
\label{sec-promoter}

% \subsection{Splice Junction Dataset Continued}
% \begin{figure}[h!]
% 	\resizebox{.98\textwidth}{!}{%
% 	 \begin{subfigure}[c]{0.46\textwidth}\centering
% 	\includegraphics[width=\textwidth]{figures/splice_trace_tau.png}
% 	\caption{Traceplot for $\bo\tau$.}
% 	\end{subfigure}
% 	%
% 	\quad
% 	\begin{subfigure}[c]{0.46\textwidth}\centering
% 	\includegraphics[width=\textwidth]{figures/splice_trace_deep_arm3.png}
% 	\caption{Traceplot for $\bo\eta$.}
% 	\end{subfigure}
% 	}
% \caption{Traceplots of the Gibbs sampler when applied to the splice junction data.}
% \label{fig-trace-splice}
% \end{figure}

% \subsection{Promoter Dataset}
The promoter data consist of A, C, G, T nucleotides at $p = 57$ positions for $n = 106$ DNA nucleotide sequences, along with a binary variable for each sequence indicating whether it is a promoter or a non-promoter sequence. 
In this dataset, the first 53 sequences are promoters and remaining 53 are non-promoter sequences.
We let the variable $z$ in the deepest latent layer have $B=2$ categories, since there are two types of sequences; we do not use any information of which sequence belongs to which type when fitting the two-latent-layer SG-LCM to data.

Similarly to the analysis of the splice junction data, we  take $K_{\upper}=7$ and $\alpha_0 = 2$ in the CSP prior. 
% We find through examining the traceplots that the proposed Gibbs sampler has good mixing behavior when applied to this dataset.
We ran the chain for 15000 iterations, discarding the first 5000 as burn-in, and retaining every fifth sample post burn-in to thin the chain.
Our method selects a number of $K^\star = 4$ binary latent variables in the middle layer.
% After collecting  the posterior samples indexed by $r\in\{1,\ldots,R=2000\}$, for each $r$ denote the samples of $\GG$ by $(g^{(r)}_{j,k};\, j\in[60], k\in[5])$. 
Denote the collected posterior samples of the nucleotide sequences' latent binary traits by $(a^{(r)}_{i,k};\, i\in[106], k\in[4])$. and denote those of the nucleotide sequences' deep latent category by $(z_i^{(r)};\, i\in[106])$ where $z_i^{(r)}\in\{1,2\}$.
Similarly as for the splice junction data, we define the final estimators as %of the graphical matrix to be 
$\widehat\GG = \left(\widehat g_{j,k}\right)$,  $\widehat{\mb A} = \left(\widehat a_{i,k}\right)$, and $\widehat{\mb Z} = (\widehat z_{i,b})$ to be
\begin{align*}
    &\widehat g_{j,k} = \mathbbm{1}\left(\frac{1}{R} \sum_{r=1}^R g^{(r)}_{j,k} > \frac{1}{2}\right),\quad
    \widehat a_{i,k} = \mathbbm{1}\left(\frac{1}{R} \sum_{r=1}^R a^{(r)}_{i,k} > \frac{1}{2}\right),\\[3mm]
    &
    \widehat z_{i,b} =
    \begin{dcases}
        1, & \text{if}~ b = \argmax_{b\in[B]} \frac{1}{R} \sum_{r=1}^R \mathbbm{1}\left(z^{(r)}_{i}=b\right);\\
        0, & \text{otherwise}.
    \end{dcases}
\end{align*}
% In summary, the $\widehat g_{j,k}$, $ \widehat a_{i,k}$, and $\widehat z_{i,b}$ gather information of the element-wise posterior modes of the discrete latent  structures in our model.
%Then the $60\times 5$  matrix $\widehat\GG$ depicts how the loci load on the binary latent traits, the $3175\times 5$ matrix $\widehat{\mb A}$ depicts the presence or absence of each binary latent trait in each  nucleotide sequence, and the $3175\times 3$ matrix $\widehat{\mb Z}$ depicts which deep latent group each nucleotide sequence belongs to. 
%The $\widehat{\mb G}$, $\widehat{\mb A}$, and $\widehat{\mb Z}$ are all binary matrices, but the first two are binary feature matrices while the last one $\widehat{\mb Z}$ has each row having exactly one entry of ``1'' indicating group membership.
In Fig.~\ref{fig-promoter}, the three plots display the three estimated matrices $\widehat\GG$, $\widehat{\mb A}$, and $\widehat{\mb Z}$, respectively.
As for the estimated loci loading matrix $\widehat\GG$, Fig.~\ref{fig-splice}(a) provides information on how the $p=60$ loci depend on the four selected binary latent traits. 

The gray reference lines in Fig.~\ref{fig-promoter}(b)-(c) mark the $x$-axis value of $i=53$, which separates the two nucleotide sequence types in the sample: the promoter sequences $i\in\{1,\ldots,53\}$ and non-promoter sequences $i\in\{54,\ldots,106\}$.
As for the learned binary latent representations of the nucleotide sequences shown in Fig.~\ref{fig-promoter}(b), we found that the first 53 promoter sequences generally possess more binary latent traits while the last 53 non-promoter sequences do not have any binary latent traits.
The following eight promoter sequences indexed from $i=10,11,\ldots,17$ possess the densest binary trait representations, \verb|PRNAB_P1|, \verb|PRNAB_P2|, \verb|PRNDEX_P2|, \verb|RRND_P1|, \verb|RRNE_P1|, \verb|RRNG_P1|, \verb|RRNG_P2|, \verb|RRNX_P1|.
In particular, Fig.~\ref{fig-promoter}(b) shows that seven out of these eight sequences, except the fourth one \verb|RRND_P1|, are the only seven sequences in the dataset that possess at least three binary traits.
Interestingly, the aforementioned eight sequences indexed from $i=10,11,\ldots,17$ are also the only eight out of the $n=106$ ones that have their names starting with ``\verb|PRN|''.
This result shows the usefulness of the proposed methodology in learning interpretable and meaningful latent features.
Such insight also suggests that domain scientists could conduct future research to investigate how these eight promoter sequences function differently from the other promoter sequences.

% On the other hand, Fig.~\ref{fig-splice}(b)--(c) show that the two matrices $\widehat{\mb A}$ and $\widehat{\mb Z}$ corresponding to the $n=106$ nucleotide sequences exhibit a clear pattern of clustering the same type of genes and distinguishing the different types of genes among EI, IE, and N.

\begin{figure}[h!]
	\centering
	
    \includegraphics[width=0.9\textwidth]{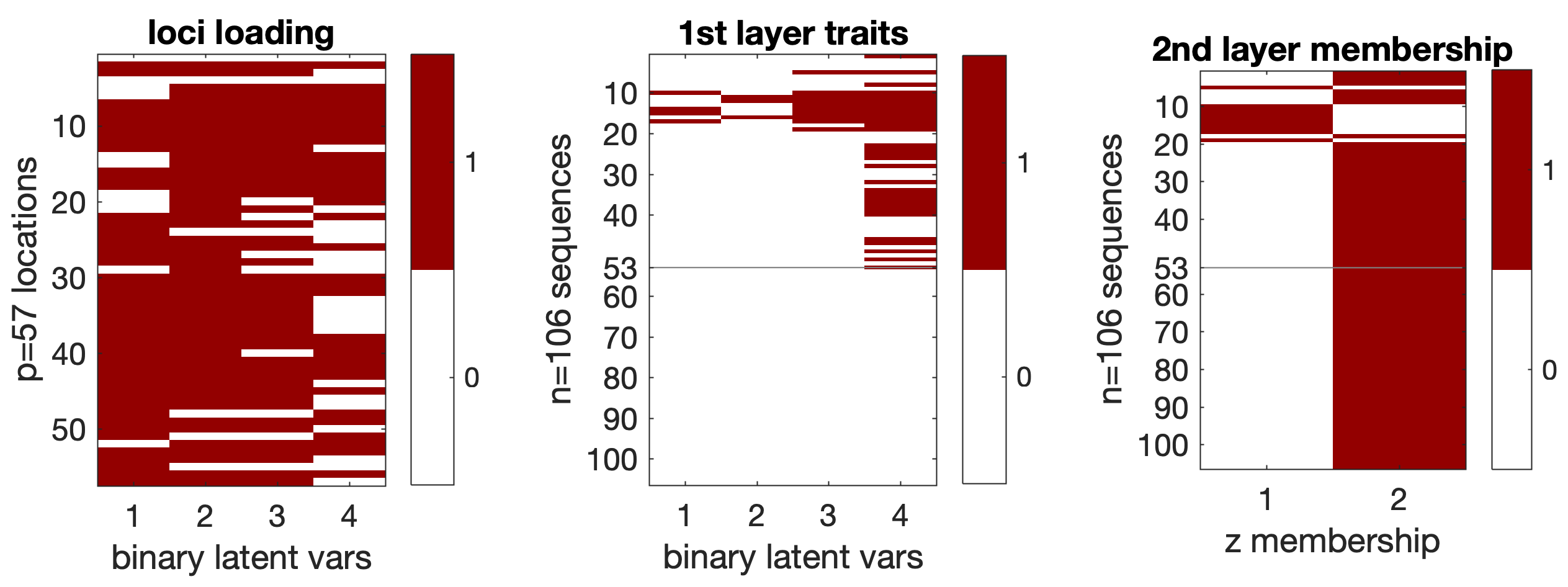}%
    
    \begin{minipage}[c]{0.3\textwidth}\centering
    (a) $\widehat\GG_{60\times 4}$
    \end{minipage}
    ~
    \begin{minipage}[c]{0.3\textwidth}\centering
    (b) $\widehat{\mb A}_{106\times 4}$
    \end{minipage}
    ~
    \begin{minipage}[c]{0.3\textwidth}\centering
    (c) $\widehat{\mb Z}_{106\times 2}$
    \end{minipage}

	\caption{Promoter data analysis under the CSP prior with $K_{\upper} = 7$. 
	Plots are presented with the $K^\star = 4$ binary latent traits selected by our proposed method. The gray reference lines in (b) and (c) mark the $x$-axis value of $i=53$, which separates the two nucleotide sequence types in the sample: the promoter sequences $i\in\{1,\ldots,53\}$ and non-promoter sequences $i\in\{54,\ldots,106\}$.}
	\label{fig-promoter}
\end{figure}

%\end{document}

% \begin{table}
% % \caption{caption}
% % \label{tab-promoter}
% 	\begin{tabular}{llll}
% 		\toprule
% 		$i$ & Genes & $\widehat z_i$ & $\widehat{\bo a}_i$\\[1mm]
% 		\midrule
%         &\verb|PRNAB_P1| & & \\[1mm]
%         &\verb|PRNAB_P2| & & \\[1mm]
%         &\verb|PRNDEX_P2| & & \\[1mm]
%         &\verb|RRND_P1| & & \\[1mm]
%         &\verb|RRNE_P1| & & \\[1mm]
%         &\verb|RRNG_P1| & & \\[1mm]
%         &\verb|RRNG_P2| & & \\[1mm]
%         &\verb|RRNX_P1| & & \\[1mm]
% 		\bottomrule
% 	\end{tabular}
% \end{table}

% \begin{table}0
% % \caption{\label{tab01}A table as an example}
% \centering
% \fbox{%
% \begin{tabular}{*{10}{c}}
% \em a&\em b&\em c&\em da&\em b&\em c&\em da&\em b&\em c&\em d\\
% \hline
% a&b&c&da&b&c&da&b&c&d\\
% a&b&c&da&b&c&da&b&c&d\\
% a&b&c&da&b&c&da&b&c&d\\
% a&b&c&da&b&c&da&b&c&d\\
% \end{tabular}}
% \end{table}

%% file: Bayesian_pyramids_R2.bbl
\begin{thebibliography}{}

\bibitem[Allman et~al., 2009]{allman2009}
Allman, E.~S., Matias, C., and Rhodes, J.~A. (2009).
\newblock Identifiability of parameters in latent structure models with many
  observed variables.
\newblock {\em The Annals of Statistics}, 37(6A):3099--3132.

\bibitem[Anandkumar et~al., 2013]{anandkumar2013dag}
Anandkumar, A., Hsu, D., Javanmard, A., and Kakade, S. (2013).
\newblock Learning linear {B}ayesian networks with latent variables.
\newblock In {\em International Conference on Machine Learning}, pages
  249--257.

\bibitem[Anderson and Rubin, 1956]{anderson1956fa}
Anderson, T.~W. and Rubin, H. (1956).
\newblock Statistical inference in factor analysis.
\newblock In {\em Proceedings of the Third Berkeley Symposium on Mathematical
  Statistics and Probability}, volume~5, pages 111--150.

\bibitem[Angelino et~al., 2017]{angelino2017rule}
Angelino, E., Larus-Stone, N., Alabi, D., Seltzer, M., and Rudin, C. (2017).
\newblock Learning certifiably optimal rule lists for categorical data.
\newblock {\em The Journal of Machine Learning Research}, 18(1):8753--8830.

\bibitem[Blei et~al., 2003]{blei2003lda}
Blei, D.~M., Ng, A.~Y., and Jordan, M.~I. (2003).
\newblock Latent {D}irichlet allocation.
\newblock {\em Journal of Machine Learning Research}, 3(Jan):993--1022.

\bibitem[Chen et~al., 2020a]{chen2020cdm}
Chen, Y., Culpepper, S., and Liang, F. (2020a).
\newblock A sparse latent class model for cognitive diagnosis.
\newblock {\em Psychometrika}, 85:121–153.

\bibitem[Chen et~al., 2020b]{chen2019factor}
Chen, Y., Li, X., and Zhang, S. (2020b).
\newblock Structured latent factor analysis for large-scale data:
  Identifiability, estimability, and their implications.
\newblock {\em Journal of the American Statistical Association},
  115(432):1756–1770.

\bibitem[Chen et~al., 2015]{chen2015qmat}
Chen, Y., Liu, J., Xu, G., and Ying, Z. (2015).
\newblock Statistical analysis of {$Q$}-matrix based diagnostic classification
  models.
\newblock {\em Journal of the American Statistical Association},
  110(510):850--866.

\bibitem[Collins and Lanza, 2009]{collins2009lca}
Collins, L.~M. and Lanza, S.~T. (2009).
\newblock {\em Latent Class and Latent Transition Analysis: With Applications
  in the Social, Behavioral, and Health Sciences}, volume 718.
\newblock John Wiley \& Sons.

\bibitem[de~la Torre, 2011]{dela2011gdina}
de~la Torre, J. (2011).
\newblock The generalized {DINA} model framework.
\newblock {\em Psychometrika}, 76(2):179--199.

\bibitem[Dellaportas et~al., 2002]{dellaportas2002bayesian}
Dellaportas, P., Forster, J.~J., and Ntzoufras, I. (2002).
\newblock On {B}ayesian model and variable selection using {MCMC}.
\newblock {\em Statistics and Computing}, 12(1):27--36.

\bibitem[Doshi-Velez and Ghahramani, 2009]{doshi2009correlated}
Doshi-Velez, F. and Ghahramani, Z. (2009).
\newblock Correlated non-parametric latent feature models.
\newblock {\em Uncertainty in Artificial Intelligence}.

\bibitem[Doshi-Velez and Kim, 2017]{doshi2017int}
Doshi-Velez, F. and Kim, B. (2017).
\newblock Towards a rigorous science of interpretable machine learning.
\newblock {\em arXiv preprint arXiv:1702.08608}.

\bibitem[Drton et~al., 2011]{drton2011global}
Drton, M., Foygel, R., and Sullivant, S. (2011).
\newblock Global identifiability of linear structural equation models.
\newblock {\em The Annals of Statistics}, 39(2):865--886.

\bibitem[Dua and Graff, 2017]{uci2019}
Dua, D. and Graff, C. (2017).
\newblock {UCI} machine learning repository.

\bibitem[Dunson and Xing, 2009]{dunson2009}
Dunson, D.~B. and Xing, C. (2009).
\newblock Nonparametric {B}ayes modeling of multivariate categorical data.
\newblock {\em Journal of the American Statistical Association},
  104(487):1042--1051.

\bibitem[Erosheva et~al., 2004]{erosheva2004pnas}
Erosheva, E., Fienberg, S., and Lafferty, J. (2004).
\newblock Mixed-membership models of scientific publications.
\newblock {\em Proceedings of the National Academy of Sciences}, 101(suppl
  1):5220--5227.

\bibitem[Eysenck et~al., 2020]{eysenck2020person}
Eysenck, S.~B., Barrett, P.~T., and Saklofske, D.~H. (2020).
\newblock The junior {E}ysenck personality questionnaire.
\newblock {\em Personality and Individual Differences}, page 109974.

\bibitem[Fang et~al., 2019]{fang2019cdm}
Fang, G., Liu, J., and Ying, Z. (2019).
\newblock On the identifiability of diagnostic classification models.
\newblock {\em Psychometrika}, 84(1):19--40.

\bibitem[Gelman et~al., 2013]{gelman2013bayesian}
Gelman, A., Carlin, J.~B., Stern, H.~S., Dunson, D.~B., Vehtari, A., and Rubin,
  D.~B. (2013).
\newblock {\em Bayesian Data Analysis}.
\newblock CRC Press.

\bibitem[Gelman et~al., 2008]{gelman2008weakly}
Gelman, A., Jakulin, A., Pittau, M.~G., and Su, Y.-S. (2008).
\newblock A weakly informative default prior distribution for logistic and
  other regression models.
\newblock {\em The Annals of Applied Statistics}, 2(4):1360--1383.

\bibitem[Gelman and Rubin, 1992]{gelman1992inference}
Gelman, A. and Rubin, D.~B. (1992).
\newblock Inference from iterative simulation using multiple sequences.
\newblock {\em Statistical Science}, 7(4):457--472.

\bibitem[Goodman, 1974]{goodman1974lcm}
Goodman, L.~A. (1974).
\newblock Exploratory latent structure analysis using both identifiable and
  unidentifiable models.
\newblock {\em Biometrika}, 61(2):215--231.

\bibitem[Gu and Xu, 2019]{gu2019learning}
Gu, Y. and Xu, G. (2019).
\newblock Learning attribute patterns in high-dimensional structured latent
  attribute models.
\newblock {\em Journal of Machine Learning Research}, 20(115):1--58.

\bibitem[Gu and Xu, 2020]{gu2020partial}
Gu, Y. and Xu, G. (2020).
\newblock Partial identifiability of restricted latent class models.
\newblock {\em Annals of Statistics}, 48(4):2082--2107.

\bibitem[Gyllenberg et~al., 1994]{gyllenberg1994non}
Gyllenberg, M., Koski, T., Reilink, E., and Verlaan, M. (1994).
\newblock Non-uniqueness in probabilistic numerical identification of bacteria.
\newblock {\em Journal of Applied Probability}, 31:542--548.

\bibitem[Haertel, 1989]{haertel1989rlcm}
Haertel, E.~H. (1989).
\newblock Using restricted latent class models to map the skill structure of
  achievement items.
\newblock {\em Journal of Educational Measurement}, 26(4):301--321.

\bibitem[Hinton, 2009]{hinton2009dbn}
Hinton, G.~E. (2009).
\newblock Deep belief networks.
\newblock {\em Scholarpedia}, 4(5):5947.

\bibitem[Hinton et~al., 2006]{hinton2006dbn}
Hinton, G.~E., Osindero, S., and Teh, Y.-W. (2006).
\newblock A fast learning algorithm for deep belief nets.
\newblock {\em Neural Computation}, 18(7):1527--1554.

\bibitem[Holmes and Held, 2006]{holmes2006}
Holmes, C.~C. and Held, L. (2006).
\newblock Bayesian auxiliary variable models for binary and multinomial
  regression.
\newblock {\em Bayesian Analysis}, 1(1):145--168.

\bibitem[Kitagawa and Tetenov, 2018]{kitagawa2018ewm}
Kitagawa, T. and Tetenov, A. (2018).
\newblock Who should be treated? {E}mpirical welfare maximization methods for
  treatment choice.
\newblock {\em Econometrica}, 86(2):591--616.

\bibitem[Kolda and Bader, 2009]{koldabader2009}
Kolda, T.~G. and Bader, B.~W. (2009).
\newblock Tensor decompositions and applications.
\newblock {\em SIAM Review}, 51(3):455--500.

\bibitem[Koopmans and Reiersol, 1950]{koopmans1950identification}
Koopmans, T.~C. and Reiersol, O. (1950).
\newblock The identification of structural characteristics.
\newblock {\em The Annals of Mathematical Statistics}, 21(2):165--181.

\bibitem[Kruskal, 1977]{kruskal1977}
Kruskal, J.~B. (1977).
\newblock Three-way arrays: rank and uniqueness of trilinear decompositions,
  with application to arithmetic complexity and statistics.
\newblock {\em Linear Algebra and its Applications}, 18(2):95--138.

\bibitem[Lazarsfeld, 1950]{lazarsfeld1950}
Lazarsfeld, P.~F. (1950).
\newblock The logical and mathematical foundation of latent structure analysis.
\newblock {\em Studies in Social Psychology in World War II Vol. IV:
  Measurement and Prediction}, pages 362--412.

\bibitem[Lee et~al., 2007]{lee2007sdpn}
Lee, H., Ekanadham, C., and Ng, A. (2007).
\newblock Sparse deep belief net model for visual area {V2}.
\newblock {\em Advances in Neural Information Processing Systems}, 20:873--880.

\bibitem[Lee et~al., 2009]{lee2009dbn}
Lee, H., Grosse, R., Ranganath, R., and Ng, A.~Y. (2009).
\newblock Convolutional deep belief networks for scalable unsupervised learning
  of hierarchical representations.
\newblock In {\em Proceedings of the 26th annual International Conference on
  Machine Learning}, pages 609--616.

\bibitem[Legramanti et~al., 2020]{legramanti2020}
Legramanti, S., Durante, D., and Dunson, D.~B. (2020).
\newblock Bayesian cumulative shrinkage for infinite factorizations.
\newblock {\em Biometrika}, 107(3):745--752.

\bibitem[Li and Wong, 2003]{li2003dna}
Li, J. and Wong, L. (2003).
\newblock Using rules to analyse bio-medical data: a comparison between {C4. 5}
  and {PCL}.
\newblock In {\em International Conference on Web-Age Information Management},
  pages 254--265. Springer.

\bibitem[McLachlan and Basford, 1988]{mclachlan1988mixture}
McLachlan, G.~J. and Basford, K.~E. (1988).
\newblock {\em Mixture Models: Inference and Applications to Clustering},
  volume~38.
\newblock M. Dekker New York.

\bibitem[Miettinen and Vreeken, 2014]{miettinen2014bmf}
Miettinen, P. and Vreeken, J. (2014).
\newblock Mdl4bmf: Minimum description length for {B}oolean matrix
  factorization.
\newblock {\em ACM Transactions on Knowledge Discovery from Data (TKDD)},
  8(4):1--31.

\bibitem[Mourad et~al., 2013]{mourad2013ltm}
Mourad, R., Sinoquet, C., Zhang, N.~L., Liu, T., and Leray, P. (2013).
\newblock A survey on latent tree models and applications.
\newblock {\em Journal of Artificial Intelligence Research}, 47:157--203.

\bibitem[Murdoch et~al., 2019]{murdoch2019int}
Murdoch, W.~J., Singh, C., Kumbier, K., Abbasi-Asl, R., and Yu, B. (2019).
\newblock Interpretable machine learning: definitions, methods, and
  applications.
\newblock {\em Proceedings of the National Academy of Sciences},
  116(44):22071--22080.

\bibitem[Murphy, 2012]{murphy2012machine}
Murphy, K.~P. (2012).
\newblock {\em Machine Learning: A Probabilistic Perspective}.
\newblock MIT press.

\bibitem[Nguyen et~al., 2016]{nguyen2016dna}
Nguyen, N.~G., Tran, V.~A., Ngo, D.~L., Phan, D., Lumbanraja, F.~R., Faisal,
  M.~R., Abapihi, B., Kubo, M., Satou, K., et~al. (2016).
\newblock {DNA} sequence classification by convolutional neural network.
\newblock {\em Journal of Biomedical Science and Engineering}, 9(05):280.

\bibitem[Ovaskainen and Abrego, 2020]{ovaskainen2020eco}
Ovaskainen, O. and Abrego, N. (2020).
\newblock {\em Joint Species Distribution Modelling: With Applications in R}.
\newblock Ecology, Biodiversity and Conservation. Cambridge University Press.

\bibitem[Ovaskainen et~al., 2016]{ovaskainen2016eco}
Ovaskainen, O., Abrego, N., Halme, P., and Dunson, D.~B. (2016).
\newblock Using latent variable models to identify large networks of
  species-to-species associations at different spatial scales.
\newblock {\em Methods in Ecology and Evolution}, 7(5):549--555.

\bibitem[Papastamoulis, 2016]{papastamoulis2016label}
Papastamoulis, P. (2016).
\newblock label. switching: An {R} package for dealing with the label switching
  problem in {MCMC} outputs.
\newblock {\em Journal of Statistical Software}, 69(c01).

\bibitem[Pearl, 2014]{pearl2014}
Pearl, J. (2014).
\newblock {\em Probabilistic Reasoning in Intelligent Systems: Networks of
  Plausible Inference}.
\newblock Elsevier.

\bibitem[Pokholok et~al., 2005]{pokholok2005cell}
Pokholok, D.~K., Harbison, C.~T., Levine, S., Cole, M., Hannett, N.~M., Lee,
  T.~I., Bell, G.~W., Walker, K., Rolfe, P.~A., Herbolsheimer, E., et~al.
  (2005).
\newblock Genome-wide map of nucleosome acetylation and methylation in yeast.
\newblock {\em Cell}, 122(4):517--527.

\bibitem[Polson et~al., 2013]{polson2013pg}
Polson, N.~G., Scott, J.~G., and Windle, J. (2013).
\newblock Bayesian inference for logistic models using {P\'o}lya--{G}amma
  latent variables.
\newblock {\em Journal of the American Statistical Association},
  108(504):1339--1349.

\bibitem[Poon and Domingos, 2011]{poon2011spn}
Poon, H. and Domingos, P. (2011).
\newblock Sum-product networks: A new deep architecture.
\newblock In {\em 2011 IEEE International Conference on Computer Vision
  Workshops (ICCV Workshops)}, pages 689--690. IEEE.

\bibitem[Rousseau and Mengersen, 2011]{rousseau2011overfit}
Rousseau, J. and Mengersen, K. (2011).
\newblock Asymptotic behaviour of the posterior distribution in overfitted
  mixture models.
\newblock {\em Journal of the Royal Statistical Society: Series B (Statistical
  Methodology)}, 73(5):689--710.

\bibitem[Rudin, 2019]{rudin2019nature}
Rudin, C. (2019).
\newblock Stop explaining black box machine learning models for high stakes
  decisions and use interpretable models instead.
\newblock {\em Nature Machine Intelligence}, 1(5):206--215.

\bibitem[Rupp and Templin, 2008]{rupp2008review}
Rupp, A.~A. and Templin, J.~L. (2008).
\newblock Unique characteristics of diagnostic classification models: A
  comprehensive review of the current state-of-the-art.
\newblock {\em Measurement}, 6(4):219--262.

\bibitem[Semenova and Chernozhukov, 2021]{semenova2021debiased}
Semenova, V. and Chernozhukov, V. (2021).
\newblock Debiased machine learning of conditional average treatment effects
  and other causal functions.
\newblock {\em The Econometrics Journal}, 24(2):264--289.

\bibitem[Skinner, 2019]{skinner2019survey}
Skinner, C. (2019).
\newblock Analysis of categorical data for complex surveys.
\newblock {\em International Statistical Review}, 87:S64--S78.

\bibitem[Teicher, 1961]{teicher1961mix}
Teicher, H. (1961).
\newblock Identifiability of mixtures.
\newblock {\em The Annals of Mathematical Statistics}, 32(1):244--248.

\bibitem[von Davier and Lee, 2019]{von2019handbook}
von Davier, M. and Lee, Y.-S. (2019).
\newblock {\em Handbook of Diagnostic Classification Models}.
\newblock Springer.

\bibitem[Xu, 2017]{xu2017}
Xu, G. (2017).
\newblock Identifiability of restricted latent class models with binary
  responses.
\newblock {\em The Annals of Statistics}, 45(2):675--707.

\bibitem[Zhao et~al., 2015]{zhao2015relate}
Zhao, H., Melibari, M., and Poupart, P. (2015).
\newblock On the relationship between sum-product networks and {B}ayesian
  networks.
\newblock In {\em International Conference on Machine Learning}, pages
  116--124.

\bibitem[Zhou et~al., 2015]{zhoudunson2015}
Zhou, J., Bhattacharya, A., Herring, A.~H., and Dunson, D.~B. (2015).
\newblock Bayesian factorizations of big sparse tensors.
\newblock {\em Journal of the American Statistical Association},
  110(512):1562--1576.

\bibitem[Zwiernik, 2018]{zwiernik2018ltm}
Zwiernik, P. (2018).
\newblock Latent tree models.
\newblock In {\em Handbook of Graphical Models}, pages 283--306. CRC Press.

\end{thebibliography}
